\documentclass[12pt]{article}
\pdfoutput=1
\usepackage[square,comma,numbers,sort&compress]{natbib}
\usepackage{graphicx,epstopdf,amssymb,amsfonts,amsmath,amsthm,array,
mathrsfs,amscd,enumitem}
\usepackage{tikz}
\usetikzlibrary{shapes.misc}
\tikzset{cross/.style={cross out, draw=black, minimum size=2*(#1-\pgflinewidth), inner sep=0pt, outer sep=0pt},
cross/.default={3pt}}
\usetikzlibrary{decorations.pathmorphing}
\usepackage{float}
\usepackage{xcolor}
\usepackage{appendix}
\usepackage{fnpct}
\usepackage{lmodern}
\usepackage[T1]{fontenc}
\usepackage[utf8]{inputenc}
\usepackage{hyperref}
\newcommand{\pbs}[1]{\let\temp=\\#1\let\\=\temp}
\numberwithin{equation}{section}
\def\be{\begin{equation}}\def\ee{\end{equation}}
\providecommand{\abs}[1]{\lvert#1\rvert}
\def\cvp{\raise 2pt\hbox{,}}

\def\re{\mathop{\text{Re}}\nolimits}  

 \def\d{{\rm d}} 
\def\la{\lambda}\def\La{\Lambda}
 \def\uone{\text{U}(1)}

\DeclareMathOperator{\diff}{Diff}

\DeclareMathOperator{\Vol}{\text{Vol}}\def\mani{\mathscr M}\def\disk{\mathscr D}

\def\Htwo{\text{H}^{2}}\def\Sone{\text{S}^{1}}

\def\PslR{\text{PSL}(2,\mathbb R)}

\def\Seff{S_{\text{eff}}}

\def\PslR{\text{PSL}(2,\mathbb R)}

\def\Htwo{\text{H}^{2}}

\def\la{\lambda}
\def\La{\Lambda}

\def\mani{\mathscr M}\def\disk{\mathscr D}
\def\Sdil{S_{\text{dil}}}

\def\uone{\text{U}(1)}

\def\SigB{\Sigma_{\text B}}\def\XiB{\Xi_{\text B}}

\newtheorem{proposition}{Proposition}[section]

\def\plb#1#2#3{{\it Phys.\ Lett.\ }{\bf B #1} (#2) #3}
\def\npb#1#2#3{{\it Nucl.\ Phys.\ }{\bf B #1} (#2) #3}

\def\prl#1#2#3{{\it Phys.\ Rev.\ Lett.\ }{\bf #1} (#2) #3}
\def\jhep#1#2#3{{\it J. High Energy Phys.\ }{\bf #1} (#2) #3}
\def\prd#1#2#3{{\it Phys.\ Rev.\ }{\bf D #1} (#2) #3}

\def\cmp#1#2#3{{\it Comm.\ Math.\ Phys.\ }{\bf #1} (#2) #3}

\def\jmp#1#2#3{{\it J.\ Math.\ Phys.\ }{\bf #1} (#2) #3}

\def\imath#1#2#3{{\it Invent math }{\bf #1} (#2) #3}

\def\jpa#1#2#3{{\it J.\ Phys.\ }{\bf A #1} (#2) #3}

\begin{document}
%
%
{\pagestyle{empty}
\parskip 0in
\

\vfill
\begin{center}
{\LARGE Finite cut-off JT and Liouville quantum gravities}

\bigskip

{\LARGE on the disk at one loop}


\vspace{0.4in}

Soumyadeep C{\scshape haudhuri} and Frank F{\scshape errari}

\medskip
{\it Service de Physique Th\'eorique et Math\'ematique\\
Universit\'e Libre de Bruxelles (ULB) and International Solvay Institutes\\
Campus de la Plaine, CP 231, B-1050 Bruxelles, Belgique}

\smallskip
{\tt chaudhurisoumyadeep@gmail.com, frank.ferrari@ulb.be}
\end{center}
\vfill\noindent

Within the path integral formalism, we compute the disk partition functions of two dimensional Liouville and JT quantum gravity theories coupled to a matter CFT of central charge $c$, with cosmological constant $\La$, in the limit $c\rightarrow -\infty$, $|\La|\rightarrow\infty$, for fixed $\La/c$ and fixed and finite disk boundary length $\ell$, to leading and first subleading order in the $1/|c|$ expansion. In the case of Liouville theory, we find perfect agreement with the asymptotic expansion of the known exact FZZT partition function. In the case of JT gravity, we obtain the first explicit results for the partition functions at finite cut-off, in the three versions (negative, zero and positive curvature) of the model. Our findings are in agreement with predictions from the recent proposal for a microscopic definition of JT gravity, including the $c\rightarrow -\infty$ expansion of the Hausdorff dimension of the boundary. In the negative curvature case, we also provide evidence for the emergence of an effective Schwarzian description at length scales much greater than the curvature length scale.

\vfill

\medskip
%
%
\newpage\pagestyle{plain}
\baselineskip 16pt
\setcounter{footnote}{0}

}

\tableofcontents

\section{Introduction}

\subsection{Generalities}

Recently, a general framework has been proposed for the study of finite cut-off JT quantum gravity from first principles \cite{Fer1,Fer2,Fer3}, from both a discretized, statistical physics point of view and a continuum, field theoretic point of view. This framework provides a rigorous microscopic definition of the model and thus allows, in principle, to study the full quantum theory for any finite values of the parameters. 

JT quantum gravity may be defined in negative, zero or positive curvature and we shall deal with the three cases in parallel. We limit ourselves to Euclidean signature and to the disk topology. We consider the case in which JT gravity is coupled to an arbitrary matter CFT of central charge $c$. We focus on the computation of the partition functions, denoted by $Z^{0}$, $Z^{-}$ and $Z^{+}$ in the zero, negative and positive curvature theories, respectively, using the path integral formalism.

According to \cite{Fer1}, the boundary in JT gravity is a fractal curve and thus has no standard geometrical length. The naive length is replaced by a renormalized quantum length parameter $\beta$,\footnote{The parameter $\beta$ is a microscopic parameter of the theories at finite cut-off and must not be confused with another parameter, often denoted by $\beta$ as well in the literature, that appears in the so-called Schwarzian limit of the negative curvature model. To avoid confusion, we shall denote by $\beta_{\text S}$ the Schwarzian limit parameter in the following.} that has anomalous Hausdorff dimension $d_{\text H}=1/\nu$, such that $\beta^{\nu}$ has the dimension of length. The critical exponent $\nu$ is predicted to be \cite{Fer1}
\be\label{critexponents} \nu = \frac{1}{2}\Biggl[1+\sqrt{\frac{c}{c-24}}\Biggr]\, .\ee
The parameter $\beta$ is very similar to the diffusion time $\beta$ for a Brownian process. The critical exponent $\nu=1/2$ for pure, $c=0$, JT gravity, actually coincides with the one for Brownian curves, even though the JT gravity boundaries are not arbitrary random closed loops \cite{Fer1,Fer2}. The zero curvature theory, with cosmological constant $\La$, thus has a unique dimensionless parameter $\La\beta^{2\nu}$ on top of the central charge $c$. The partition function at fixed boundary quantum length is predicted in \cite{Fer1} to take the general form
\be\label{Zzerexact} Z^{0} = \beta^{-1-2\nu\vartheta}g^{0}\bigl(\La\beta^{2\nu},c\bigr)
\ee
for some non-trivial function $g^{0}$ that has a smooth small $\La$ limit and a critical exponent
\be\label{varthetaexp}\vartheta = 2-\frac{c}{12}\,\cdotp\ee
In non-zero curvature, the curvature scale $L$ allows to build a second dimensionless parameter $\beta/L^{1/\nu}$. The partition function then takes the general form
\be\label{Zpmerexact} Z^{\pm} = \beta^{-1-2\nu\vartheta}g^{\pm}\bigl(\La\beta^{2\nu},\beta/L^{1/\nu},c\bigr)\ee
for functions $g^{+}$ and $g^{-}$ that must be such that
\be\label{flatlimgen} \lim_{L\rightarrow +\infty} g^{\pm}\bigl(\La\beta^{2\nu},\beta/L^{1/\nu},c\bigr) = g^{0}\bigl(\La\beta^{2\nu},c\bigr)\, .\ee

The vast majority of the published works on JT gravity focuses on the negative curvature case in the so-called Schwarzian limit, see e.g.\ \cite{JTappli1} and \cite{Fer2} for a review. The limit is defined as
\be\label{Schlimit} \La L^{2}\rightarrow -\infty\, ,\quad \frac{\ell_{\text S}}{L}\rightarrow\infty\, ,\quad \frac{2\ell_{\text S}}{|\La|L^{3}} = \beta_{\text S}\quad\text{fixed,}\ee
assuming that a long-distance effective smooth boundary length parameter $\ell_{\text S}$ exists. In this limit, it is conjectured that an effective description of the theory, in terms of the so-called reparameterization ansatz and the Schwarzian action, becomes valid. A microscopic definition of the model is then not needed to derive its leading order properties. This effective description is similar to the IR Schwarzian description in the SYK or tensor models \cite{JTappli2}. It should emerge when $\ell_{\text S}\gg L$. The microscopic properties, like the fractal structure of the boundary, become relevant below the curvature length scale $L$. The limit $\ell_{\text S}/L\rightarrow\infty$ is usually called the infinite cut-off limit, consistently with the holographic interpretation. Note that, strictly speaking, the smooth length parameter $\ell_{\text S}$ itself emerges only in this limit and should be expressed in terms of the microscopic parameters $\beta$, $L$ and $\La$ \cite{Fer2}.

\subsection{The $c\rightarrow -\infty$ limit}

One of the main motivation for the present work is to use some of the ideas presented in \cite{Fer1,Fer2,Fer3}, combined with an interesting semiclassical approximation scheme explained below, to derive, from first principles and for the first time, explicit results in negative curvature JT gravity away from the Schwarzian limit. We shall also use the same set of ideas and tools to study the positive and zero curvature models. 

Instrumental in our analysis is an old idea of Zamolodchikov \cite{Zamolod}, that was applied in the past to derive the first non-trivial quantum results for the Liouville quantum gravity theory on the sphere. Zamolodchikov's idea was discussed further recently in the context of Liouville gravity \cite{WittLiousemicl,Muhlmann,StanLiousecl} and also to study generalisations of the Liouville model that occur when non-conformal matter is coupled to gravity \cite{FerK1,FerK2}. In our context, it amounts to studying the limit
\be\label{Zamolodlimit} c\rightarrow -\infty\, ,\quad |\La|\rightarrow +\infty\, ,\quad \frac{\La}{|c|} = \frac{2\mu}{3}\quad \text{and}\quad \beta\quad \text{fixed.}\footnote{The factor of $2/3$ is a convenient convention in the definition of the parameter $\mu$, see below.}\ee
In this limit, it turns out that the path integral is dominated by a smooth classical geometry, at least for some range of the parameter $\mu$. We can then, in principle, compute all the observables in the model in a systematic loopwise $1/|c|$ expansion, using the saddle-point method. Consistent with this idea, the exact result \eqref{critexponents} admits an expansion
\be\label{nuexp} \nu = \sum_{k\geq 0}\nu_{k}|c|^{-k} = 1 - \frac{6}{|c|} +\frac{108}{|c|^{2}}  + \cdots\, ,\ee
the term proportional to $|c|^{-k}$ being a $k$-loop contribution. Note that the leading order value $\nu=1$ corresponds to a smooth boundary, consistently with the fact that the limit $c\rightarrow -\infty$ is semiclassical. In this limit we can thus identify $\beta$ with a genuine geometrical length at leading order. For this reason, we shall use the notation
\be\label{betaeqell} \beta = \ell\, .\ee
The loop corrections to the exponent $\nu$ will give $\ell$ an anomalous dimension. This anomalous dimension is itself the signature of the fractal nature of the boundary in the microscopic description of the models. Our one-loop computation below will confirm this expectation and we will find a perfect match with the $-6/|c|$ term in \eqref{nuexp}, see in particular the discussion in Sec.\ \ref{finalSec}. More generally, the logarithm of the partition functions are expected to have a large $|c|$ expansion of the form
\be\label{Zsemiclexpansion} \ln Z^{\eta} = \sum_{k\geq 0}|c|^{1-k}f^{\eta}_{k}\, .\ee
The main new result of our paper is to compute the leading and first subleading terms in this expansion, see Eqs.\ \eqref{fzeromin}, \eqref{fzerozero}, \eqref{fzeroplus} and \eqref{fzeroplusbis} for the tree-level contributions and Eqs.\ \eqref{Zzerofinal}, \eqref{Zminusfinal} and \eqref{Zplusfinal} for the one-loop contributions, the dimensionless parameters $x$ and $y$ being defined in \eqref{xydefJTmin}. Combining the one-loop results with the exact statements \eqref{critexponents} and \eqref{Zzerexact}, one can even extract predictions at any loop order, see Sec.\ \ref{finalSec}.

Interestingly, the usual Liouville quantum gravity theory on the disk can also be studied in the limit \eqref{Zamolodlimit}. The path integral turns out to be dominated by the same classical geometry relevant to the negative curvature JT model. As for the JT models, the $1/|c|$ expansion can thus be studied systematically, at least in principle, generalising the sphere computation of \cite{Zamolod}. In particular, one expects an expansion 
\be\label{ZLsemicl} \ln Z_{\text L} = \sum_{k\geq 0}|c|^{1-k} f_{k}\ee
for the Liouville disk partition function $Z_{\text L}$, of the same form as for the JT theory, Eq.\ \eqref{Zsemiclexpansion}.

Although our main interest is in JT, we have also studied the expansion \eqref{ZLsemicl} in the Liouville theory, having in mind two main motivations. First, rigorous path integral computations in Liouville theory are very few \cite{WittLiousemicl,Muhlmann,StanLiousecl} and the calculation of the one-loop contribution $f_{1}$ was never done before. This calculation is indeed much more involved on the disk than on the sphere, both for technical and conceptual reasons. But one must overcome essentially the same difficulties to do the JT gravity calculations and thus one can get $f_{1}$ for Liouville with very little additional effort if one has gotten $f_{1}^{\pm}$ for JT. Actually, the calculations in Liouville and JT can be naturally presented in parallel, and this is what we shall do below. Second, an exact formula for the partition function of Liouville gravity on the disk was proposed by FZZT in \cite{FZZT}. The ideas used by FZZT to obtain their results, based on general assumptions in the operator formalism, generalising earlier works by DOZZ \cite{DOZZ} which apply to the sphere, are independent of the path integral approach that we use in our work. Comparing their prediction with our explicit tree-level and one-loop computations then allows to check our methods in a very non-trivial way. Explicitly, in terms of the dimensionless parameter\footnote{Liouville theory is defined only for positive or zero cosmological constant.}
\be\label{defxpara} x = \frac{\ell\sqrt{\mu}}{2\pi}\,\cvp\ee
FZZT predicts
\be\label{exactWLioudiskc}  Z_{\text L} = w \ell^{\gamma-3}x^{1-\gamma} K_{1-\gamma}\Biggl[\sqrt{\frac{\pi|c|}{6\sin\frac{\pi}{1-\gamma}}}\, x\Biggr]\, ,\ee
where $K$ is the modified Bessel function of the second kind and $\gamma$ the so-called string susceptibility exponent, given by the famous KPZ formula
\be\label{gammaKPZ} \gamma = \frac{1}{12}\Bigl(c-1-\sqrt{(1-c)(25-c)}\Bigr)\, .\ee
The constant $w$ is an overall $\ell$-independent and non-universal, scheme-dependent, proportionality factor. We conventionally fix it by imposing that
\be\label{normalW} Z_{\text L}(\ell,0) = \Bigl(\frac{\ell}{2\pi}\Bigr)^{\gamma-3}\, ,\ee
which yields
\be\label{wpropL} w = \frac{8\pi^{3-\gamma}}{\Gamma(\gamma-1)}\biggl(\frac{6\sin\frac{\pi}{1-\gamma}}{\pi|c|}\biggr)^{\frac{\gamma-1}{2}}\, .\ee
In the limit \eqref{Zamolodlimit}, the exact result \eqref{exactWLioudiskc} yields, as expected, a semiclassical expansion of the form \eqref{ZLsemicl}, with explicit predictions for the $k$-loop contributions. At tree-level and one-loop, one straightforwardly gets, from the asymptotics of the Bessel function,
\begin{align}\label{ZLloopcoeff0} f_{0} & = -\frac{1}{6}\ln\frac{\ell}{2\pi} +\frac{1}{6}\ln\frac{1+\sqrt{1+x^{2}}}{2}-\frac{1}{6}\bigl(\sqrt{1+x^{2}}-1\bigr)\\
\label{ZLloopcoeff1} f_{1} & = -\frac{25}{6}\ln\frac{\ell}{2\pi} + \frac{13}{6}\ln\Bigl[\frac{1}{2}\bigl(1+\sqrt{1+x^{2}}\bigr)\Bigr]-\frac{13}{12}\bigl(\sqrt{1+x^{2}}-1\bigr) - \frac{1}{4}\ln\bigl(1+x^{2}\bigr)\, .
\end{align}
Our direct semiclassical evaluation of the path integral below will reproduce very precisely the above coefficients $f_{0}$ and $f_{1}$, see Sec.\ \ref{saddleLSec} and \ref{LoneSec}.

\subsection{Summary of the key findings}

At the technical level, computing the one-loop contributions requires the explicit evaluation of bulk and boundary determinants, which we do using $\zeta$-function techniques. The calculation of the bulk determinants in a finite size geometry turns out to be particularly involved. In order to improve the clarity of our presentation, we have decided to present this part of the calculation in a separate paper \cite{detpaper} whose results are summarised in Sec.\ \ref{detSec}. Also for ease of reading, we have given the calculations of the boundary determinants in the Appendix \ref{bddetApp}.

The main subtlety, however, is conceptual. After going to conformal gauge, the metric degrees of freedom are encoded in the conformal factor, the logarithm of which is called the Liouville field. The crucial point, emphasized in \cite{Fer1,Fer3}, is that \emph{the bulk and the boundary Liouville fields must be treated as independent integration variables in the path integral.} In other words, the Liouville field has free boundary conditions.\footnote{A discussion of the case of the Liouville theory on the disk in the grand canonical ensemble, with boundary cosmological constant $\la$ and in the limit $\la/|c|\rightarrow -\infty$, appears in \cite{StanLiousecl} (the limit $\la/|c|\rightarrow -\infty$ in the grand canonical ensemble is similar to the limit $\ell\rightarrow\infty$ in the fixed $\ell$ ensemble on which we focus  in our work). As far as we understand, the crucial boundary degrees of freedom were not taken into account in \cite{StanLiousecl}.} We show in this paper how to use explicitly these new boundary conditions. In particular, we check that they yield results that are consistent with the known exact formulas in the case of the Liouville theory.

The most important and novel application is to JT gravity. On the one hand, our results confirm that the UV structure of the models, which is independent of the choice of bulk curvature, is governed by an effective boundary action derived from the Liouville action. This yields in particular the expected $1/|c|$ expansion for the Hausdorff dimension of the boundary whose exact form was predicted in \cite{Fer1}. On the other hand, in the case of the negative curvature theory, our results are also consistent with the idea that the Schwarzian description emerges as an IR effective description on length scales much larger than the curvature length scale. Let us emphasise that the Schwarzian theory does not describe correctly the short scale structure of the JT gravity boundaries. This is a fundamentally new result that was presented in \cite{Fer1} and that we check here explicitly by a direct semi-classical path integral calculation.

\subsection{Plan of the paper}

The plan of the paper is as follows. In Section \ref{defbasicSec}, we present the models, introducing the relevant classical geometries and actions and treating the basics of the path integral formulation, including a discussion of conformal gauge, the separation between the bulk and boundary degrees of freedom, the role of the Liouville action and the gauge fixing of the $\PslR$ group of disk automorphisms. We are also putting in place some essential ingredients of the one-loop calculations. In Section \ref{clSec}, we determine the classical geometries that dominate in the limit \eqref{Zamolodlimit} by solving the relevant field equations in Liouville gravity and in the various versions of JT gravity. This yields the tree-level coefficients $f_{0}^{0}$, $f_{0}^{\pm}$ and $f_{0}$ in the expansions \eqref{Zsemiclexpansion} and \eqref{ZLsemicl}. In Section \ref{quadSec}, we compute the quadratic fluctuations around the classical backgrounds. These fluctuations determine the one-loop quantum corrections. We also discuss the stability of the various backgrounds. In Section \eqref{oneloopSec}, we perform the relevant Gaussian integrations. The bulk fluctuations produce a finite-size determinant $\det_{\text D} (\Delta + M^{2})$, where $\Delta$ is the positive Laplacian with Dirichlet boundary condition on the round disk of circumference $\ell$ embedded in hyperbolic space, the two-sphere or the Euclidean plane and $M^{2}$ is fixed to the particular values $+2/L^{2}$, $-2/L^{2}$ or zero, respectively. The determinants for arbitrary values of $M^{2}$ are studied in \cite{detpaper}. Remarkably, they can be evaluated explicitly in terms of elementary functions and the Euler $\Gamma$ function for the special mass parameters that are relevant for our purposes. The boundary fluctuations yield a class of functional determinants that we compute in Appendix \ref{bddetApp}. Putting all together (Faddeev-Popov factor for the $\PslR$ gauge fixing, contributions of the gauge-fixed modes, on-shell Liouville action contribution from reparameterization ghosts, bulk determinant and boundary determinant) yields the one-loop coefficients $f_{1}^{0}$, $f_{1}^{\pm}$ and $f_{1}$. Finally, in Section \ref{finalSec}, we discuss our results in the light of the exact predictions \eqref{critexponents} and \eqref{Zzerexact} and comment on the naive $\La\rightarrow -\infty$, fixed $\ell$, semiclassical limit of the models and on the Schwarzian limit. In particular, we show explicitly that the usual Schwarzian action is only an effective description valid on length scales much larger than the curvature scale $L$, whereas the correct UV description is governed by a completely different action that comes from the boundary Hilbert-Liouville term, as discussed in \cite{Fer1}. We also find nice evidence that the effective Schwarzian description emerges from the microscopic theory, if one uses a renormalized version of the naive scaling used to define the Schwarzian limit, consistently with the general ideas presented in \cite{Fer2}.

\section{\label{defbasicSec}Definition of the models and path integrals}

The models we are studying are formally defined by giving a base manifold, a classical action and a path integral measure. In subsections \ref{baseSec}, \ref{SclSec}, \ref{pathintSec}, we follow the discussion in \cite{Fer1,Fer3}, in which much more details can be found, providing a synthetic and self-contained account of the points relevant to the present paper. In subsection \ref{loopSec}, we specialise to the one-loop approximation we are interested in, reducing the problem to the computation of the classical saddles and the evaluation of the Gaussian path integrals around the saddles, which will be done explicitly in subsequent Sections.

\subsection{\label{baseSec}Base manifold}

In this paper, the base manifold is the disk
\be\label{diskdef} \disk = \bigl\{ z\in\mathbb C \bigm| \abs{z} < 1\bigr\}\, .\ee
The boundary is $\Sone=\partial\disk=\{z\in\mathbb C\, |\, \abs{z} =1\}$ and the closed disk is denoted by $\bar\disk = \disk\cup\partial\disk$. We often use polar coordinates $(\rho,\theta)$ or Cartesian coordinates $(x^{a})$ defined by
\be\label{polardef} z = \rho e^{i\theta} = x^{1}+ix^{2}\, .\ee
The group of disk automorphisms, which are bijective holomorphic maps from $\disk$ to $\disk$, is $\PslR$, acting as
\be\label{pslRdef} z' = e^{i\alpha}\frac{z-z_{0}}{1-\bar z_{0}z}\,\cvp\quad |z_{0}|<1\, ,\ (z,z')\in\bar\disk^{2}\, .\ee
Note that these transformations extend smoothly to $\bar\disk$.

The canonical flat round metric of perimeter $2\pi$ on $\disk$ is denoted by
\be\label{deltametdef} \delta = |\d z|^{2} = \d\rho^{2} + \rho^{2}\d\theta^{2}\, .\ee
The similar metric of perimeter $\ell$ is
\be\label{deltaldef} \delta_{\ell} = \Bigl(\frac{\ell}{2\pi}\Bigr)^{2}\, \delta\, .\ee
The canonical negative curvature round metric of perimeter $\ell$ and Ricci scalar $R=-2/L^{2}$ is
\be\label{delmindef} \delta^{-}_{\ell} = \frac{4L^{2}r_{0}^{2}|\d z|^{2}}{(1-r_{0}^{2}|z|^{2})^{2}}\ee
with 
\be\label{ellminrelneg} \frac{\ell}{L} = \frac{4\pi r_{0}}{1-r_{0}^{2}}\ee
or equivalently
\be\label{ellminusrel} r_{0} = \frac{2\pi L}{\ell}\Biggl[-1+\sqrt{1+\frac{\ell^{2}}{4\pi^{2}L^{2}}}\Biggr]\, .\ee
Note that $0\leq r_{0}<1$. The area $A^{-}$ and extrinsic curvature of the boundary $k^{-}$ are
\be\label{Akmin} A^{-} = 4\pi L^{2}\frac{r_{0}^{2}}{1-r_{0}^{2}} = L\ell r_{0}\, ,\quad k^{-} = \frac{2\pi}{\ell}\frac{1+r_{0}^{2}}{1-r_{0}^{2}} = \frac{2\pi}{\ell}\sqrt{1+\frac{\ell^{2}}{4\pi^{2}L^{2}}}\, .\ee
The metric \eqref{delmindef} may be seen as arising from the natural embedding of the disk into hyperbolic space.

The two canonical positive curvature round metrics of perimeter $\ell$ and Ricci scalar $R=+2/L^{2}$ are of the form
\be\label{delplusdef} \delta^{+}_{\ell} = \frac{4L^{2}r_{0}^{2}|\d z|^{2}}{(1+r_{0}^{2}|z|^{2})^{2}}\ee
with 
\be\label{ellminrelpos} \frac{\ell}{L} = \frac{4\pi r_{0}}{1+r_{0}^{2}}\,\cdotp\ee
Note that $\ell$ is bounded above by $2\pi L$. The area $A^{+}$ and extrinsic curvature of the boundary $k^{+}$ are
\be\label{Akplus} A^{+} = 4\pi L^{2}\frac{r_{0}^{2}}{1+r_{0}^{2}} = L\ell r_{0}\, ,\quad k^{+} = \frac{2\pi}{\ell}\frac{1-r_{0}^{2}}{1+r_{0}^{2}}\, \cdotp\ee
There are two possible solutions for $r_{0}$, which coincide if and only if $\ell = 2\pi L$,
\be\label{ellplusrel} r_{0}^{>} = \frac{2\pi L}{\ell}\Biggl[1 + \sqrt{1-\frac{\ell^{2}}{4\pi^{2}L^{2}}}\Biggr]\, ,\quad r_{0}^{<} = \frac{2\pi L}{\ell}\Biggl[1 - \sqrt{1-\frac{\ell^{2}}{4\pi^{2}L^{2}}}\Biggr]\, ,\ee
that we shall call the ``large'' and ``small'' disks. Note that $r_{0}^{>}r_{0}^{<}=1$. The solution $r_{0}^{>}$ corresponds to the large disk, having the greatest area. These metrics may be seen as arising from embedding the disk into the round two-sphere of radius $L$. The embeddings may be chosen in such a way that the two embedded disks share the same $\Sone$ boundary; the two metric correspond to ``filling'' one side of the sphere or the other, see Fig.\ \ref{posdisFig}.

\begin{figure}
\centerline{\includegraphics[width=6in]{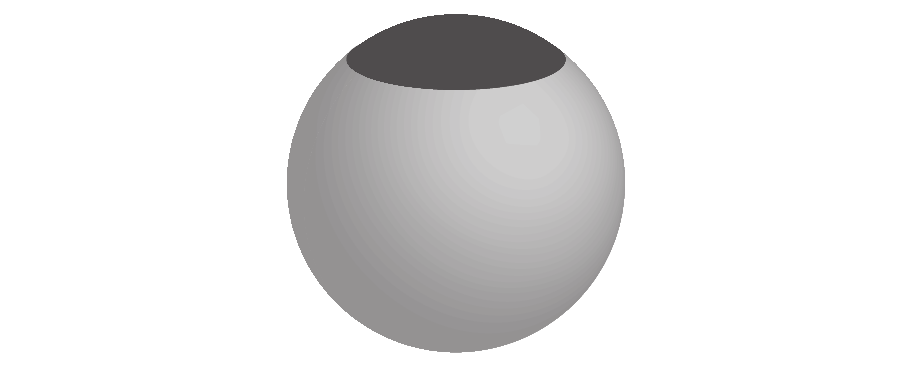}}
\caption{The two distinct positive curvature round disk metrics corresponding to the same boundary length $\ell < 2\pi L$, depicted as embeddings of the disk into the round two-sphere. The ``large'' disk, with parameter $r_{0}^{>}$, is shaded in light grey, the ``small'' disk, with parameter $r_{0}^{<}$, in dark grey.\label{posdisFig}}
\end{figure}

Finally, let us note that the Gauss-Bonnet formula on the disk
\be\label{GaussBonnet} \int_{\disk}\!\d^{2}x\sqrt{g}\, R + 2\oint_{\partial\disk}\!\d s\, k=4\pi\ee
has the following consequences, which are valid for any constant curvature metric on $\disk$ (and thus also for the canonical examples that we have just mentioned). In zero curvature, the integral of the extrinsic curvature on the boundary is fixed to $2\pi$,
\be\label{kintzerocurv} \oint_{\partial\disk}k\d s = 2\pi\, .\ee
In positive ($\eta=+1$) and negative ($\eta=-1$) curvature, this integral is directly related to the area $A$ according to
\be\label{kintarea} \oint_{\partial\disk}k\d s = 2\pi-\frac{\eta}{L^{2}} A\, .\ee
\subsection{\label{SclSec}Classical actions and counterterms}

\subsubsection{Classical actions}

\paragraph{Einstein-Hilbert term} The Gauss-Bonnet formula \eqref{GaussBonnet} implies that the Einstein-Hilbert term is purely topological in two dimensions. It has no dynamical role. It can be generated as a counterterm, see below.

\paragraph{The cosmological constant term} This is
\be\label{Scosmodef} S_{\La} = \frac{\La}{16\pi}\int_{\disk}\d^{2}x\,\sqrt{g} = \frac{\La}{16\pi} A[g]\, .\ee
In Liouville gravity, classical stability implies that $\La$ must be positive, because there exist metrics of fixed boundary length that have an arbitrary large area. The condition $\La\geq 0$ is also necessary and sufficient for the stability of the quantum theory.\footnote{Classical stability and quantum stability are not necessarily equivalent, because in the quantum theory one must take into account the measure, i.e.\ the entropy of the configurations one integrates over in the path integral.} In JT gravity in negative and zero curvature, classical stability is ensured for any $\La\in\mathbb R$ thanks to the isoperimetric inequalities, which put an upper bound on the value of the area for a given boundary length,
\begin{align}\label{isoineq1} \frac{A}{2\pi L^{2}} & \leq -1 +\sqrt{1+\Bigl(\frac{\ell}{2\pi L}\Bigr)^{2}}\quad\text{if $R=-\frac{2}{L^{2}}$}\, \cvp\\
\label{isoineq2} A & \leq\frac{\ell^{2}}{4\pi}\quad \text{if $R=0$}\, .
\end{align}
In JT gravity in positive curvature, the theory is classically non-perturbatively unstable if $\La<0$, as explained in \cite{Fer2}. We shall discuss perturbative stability in the context of the Zamolodchikov limit \eqref{Zamolodlimit} in Section \ref{StabSec}.

\paragraph{The dilaton, or Lagrange multiplier, term}

In JT gravity, the constraint of constant bulk curvature is imposed by a Lagrange mutiplier term in the action. The Lagrange multiplier field is also called in this context the ``dilaton.'' The associated term is traditionally written as
\be\label{Sdil} \Sdil = -\frac{1}{16\pi}\Biggl[\int_{\disk}\!\d^{2}x\sqrt{g}\,\Phi\Bigl(R -\frac{2\eta}{L^{2}}\Bigr) + 2\oint_{\partial\disk}\!\d s\, \Phi k\Biggr]\, .\ee
Note that the role of $\Phi$ is to impose $R=2\eta/L^{2}$ in the bulk, but the boundary extrinsic curvature $k$ must be allowed to fluctuate. As a consequence, $\Phi$ does not have boundary degrees of freedom and one imposes Dirichlet boundary conditions
\be\label{Phibc} \Phi_{|\partial\disk} = \Phi_{\text b} = \text{constant.}\ee
The boundary term in \eqref{Sdil} ensures that the variational principle is well-defined with this boundary condition.\footnote{See e.g.\ \cite{bcJTpapers} for a detailed discussion of the boundary conditions and the variational principle in JT gravity.}

\noindent\emph{Remark}: the aforementioned distinction between bulk and boundary degrees of freedom, which, following \cite{Fer1,Fer2,Fer3}, will play a central role in our analysis, may naively seem immaterial in the continuum formulation that we use in the present paper. However, it becomes very clear in a regularized lattice version of the models, as described in \cite{Fer2}. For instance, the bulk curvature $R$ is entirely localized on bulk vertices, i.e.\ vertices that are in the interior of the discretized disk, whereas the boundary vertices are associated with the extrinsic curvature.

The dilaton term \eqref{Sdil} can be conveniently rewritten as follows. Taking into account that $R=2\eta/L^{2}$ is imposed in the bulk after integrating out $\Phi$ and using \eqref{Phibc} and the Gauss-Bonnet formula \eqref{GaussBonnet}, the dilaton action \eqref{Sdil} can be replaced by an equivalent action
\be\label{Sdilprime} \Sdil' = -\frac{1}{16\pi}\int_{\disk}\!\d^{2}x\sqrt{g}\,\Phi\Bigl(R -\frac{2\eta}{L^{2}}\Bigr) + \frac{\eta\Phi_{\text b}}{8\pi L^{2}}A - \frac{\Phi_{\text b}}{4}\ee
in the path integral. When $\eta\not = 0$, we thus see that the cosmological constant action \eqref{Scosmodef} is redundant if the dilaton action is taken into account, with the identification
\be\label{philambdarel} \La = \frac{2\eta\Phi_{\text b}}{L^{2}}\,\cdotp\ee

\paragraph{The classical action for Liouville gravity coupled to a matter CFT}

This is simply
\be\label{SclLiou} S^{\text L}_{\text{cl}} = S_{\La} + S_{\text{CFT}}\, ,\ee
where $S_{\text{CFT}}$ is the action for the matter CFT coupled to gravity.

\paragraph{The classical action for JT gravity coupled to a matter CFT}

In non-zero curvature, one may use the action \eqref{Sdil} plus the matter CFT action. In zero curvature, one must then also add by hand the cosmological constant term \eqref{Scosmodef}. Equivalently, and more conveniently, we can use the action \eqref{Sdilprime}, supplemented with the cosmological constant term in zero curvature. In all cases, and discarding the trivial constant $-\Phi_{\text b}/4$ in \eqref{Sdilprime},\footnote{This constant is not physical because it can be absorbed in a counterterm, see below.} we can work with
\be\label{SJT} S^{\text{JT}}_{\text{cl}} = -\frac{1}{16\pi}\int_{\disk}\!\d^{2}x\sqrt{g}\,\Phi\Bigl(R -\frac{2\eta}{L^{2}}\Bigr) + S_{\La} + S_{\text{CFT}}\, .\ee

\subsubsection{\label{ctSec}Counterterms}

As usual, the list of possible counterterms is established by looking at all the possible local terms that respect diffeomorphism invariance and are associated with a coupling of positive mass dimension.

In the Liouville gravity theory, the only possible counterterms are the bulk terms $\int\d^{2}x\sqrt{g}R$, $\int\d^{2}x\sqrt{g} = A$ and the boundary terms $\int\d s = \ell$ and $\int\d s k$. Thanks to the Gauss-Bonnet formula, the term $\int\d^{2}x\sqrt{g}R$ can be replaced by a constant and by the boundary term $\int\d s k$. We do not need to discuss additional counterterms associated with the matter CFT; the treatment of the path integral over the matter CFT, which is well-known, will be recalled below. The counterterm action is thus of the form
\be\label{SLct} S^{\text L}_{\text{c.t.}} = \alpha_{1} + \alpha_{2}\ell + \alpha_{3}A + \alpha_{4}\int_{\partial\disk}\d s\, k \, ,\ee
where the coefficients $\alpha_{j}$ may depend on the short-distance lattice cut-off $\ell_{0}$, the cosmological constant $\La$ and the CFT central charge $c$.

Finite pieces in the counterterm coefficients $\alpha_{j}$ correspond to undetermined parameters in the theory. The coefficient $\alpha_{1}$ is fixed by imposing the normalization condition \eqref{normalW}. The coefficient $\alpha_{2}$ yields an arbitrary boundary cosmological constant factor $e^{-\la\ell}$, for some finite $\la$, in the partition function. The coefficient $\alpha_{3}$ renormalises the bulk cosmological constant $\La$. Finally, at the one-loop order that we consider in the present paper, the effect of $\alpha_{4}$ can be absorbed in $\alpha_{1}$ and $\alpha_{3}$, see below.

In JT gravity, all the counterterms of the Liouville theory are of course allowed as well, but $\alpha_{4}$ can always be absorbed in $\alpha_{1}$ and $\alpha_{3}$ thanks to the constant bulk curvature constraint. However, based on power counting, more general counterterms of the form
\be\label{countertermJTgen}\int\!\d^{2}x\sqrt{g}\, g^{\mu\nu}\partial_{\mu}\Phi\partial_{\nu}\Phi\, ,\quad
 \int\!\d^{2}x\sqrt{g}\, V(\Phi)\, ,\quad \int\!\d^{2}x\sqrt{g}\, W(\Phi)R\, ,\ee
for arbitrary polynomials $V$ and $W$, may also a priori be generated. If a kinetic term for $\Phi$ or non-linear terms in $V$ and $W$ were generated by quantum corrections, this would imply that $\Phi$ is no longer a Lagrange multiplier field and thus that the constraint of constant bulk curvature cannot be enforced at the quantum level. In other words, JT gravity, defined as a theory of random constant curvature metrics, would not make sense quantum mechanically. This, of course, goes against all evidence, in particular the fact that a microscopic, UV-complete definition of the models exists \cite{Fer1,Fer2,Fer3}. The only new allowed counterterms are thus $\int\d^{2}x\sqrt{g}\Phi$ and $\int\d^{2}x\sqrt{g} \Phi R$. In the flat theory, the second term is absent and the first forbidden, since it would modify the $R=0$ constraint to $R\not = 0$. In the curved theories, the two terms are equivalent since $R$ is fixed. They correspond to a renormalization of $L$. The counterterm action in the JT models will thus take the general form
\be\label{SJTct} S^{\text{JT}}_{\text{c.t.}} = \alpha_{1} + \alpha_{2}\ell + \alpha_{3}A + |\eta|\alpha_{4}\int_{\disk}\d^{2}x\sqrt{g}\,\Phi \, .\ee

\subsection{\label{pathintSec}Path integrals}

\subsubsection{Formal definition}

The partition functions we wish to compute have a formal path integral representation of the form
\be\label{Zformal1} Z = \frac{1}{\Vol(\diff(\disk))}\int_{\mani_{\ell}} Dg\, Z_{\text{CFT}}[g]\, e^{-S_{\La}[g]}\, .\ee
The action $S_{\La}$ is the cosmological constant term given in \eqref{Scosmodef}. The CFT partition function $Z_{\text{CFT}}[g]$ is obtained after integrating out the matter CFT degrees of freedom in a fixed background metric $g$. In the case of the Liouville theory, the space $\mani_{\ell}$ is the full space of metrics on the disk $\disk$, with a fixed boundary length $\ell$. In the case of the JT models, it is the space of constant bulk curvature metrics, $R=2\eta/L^{2}$, with a fixed boundary length $\ell$. The measure $D g$ is a formal ultralocal integration measure over these spaces. The prefactor $1/\Vol(\diff(\disk))$ indicates formally that the gauge redundancy associated with the group of disk diffeomorphisms must be fixed.

It is important to note that, in the full quantum theory, the very notion of geometric bulk area or boundary length must be revisited, since the typical quantum geometry is not smooth. For instance, it is emphasized in \cite{Fer1,Fer2,Fer3} that, in JT gravity, the boundary is a fractal, of non-trivial Hausdorff dimension $d_{\text H}=1/\nu$ given by \eqref{critexponents}, and thus its geometric length $\ell$ is ill-defined. The length parameter is replaced by a quantum length parameter $\beta$ which is such that $\beta^{\nu}$ has the dimension of length. As already discussed in the introduction, it does make sense to talk about the smooth boundary length $\ell$ in the framework of the $c\rightarrow -\infty$ expansion, since the limit is dominated by a smooth geometry for which $\nu=1$, consistently with the expansion \eqref{nuexp} at leading order. One-loop corrections are expected to produce an anomalous dimension for $\ell$ and this is what we shall find, see the discussion in Sec.\ \ref{finalSec}.

\subsubsection{Conformal gauge}

To make \eqref{Zformal1} more precise, and as advocated in \cite{Fer1,Fer3} in the case of JT gravity, we go to conformal gauge, in which the metric takes the form
\be\label{confgaugedef} g = e^{2\Sigma}\delta\, ,\ee
where $\delta$ is the round flat metric defined in Eq.\ \eqref{deltametdef}. This can be done without loss of generality, because any metric on the disk can be put in the form \eqref{confgaugedef} by acting with a disk diffeomorphism. We call $\Sigma:\bar\disk\rightarrow\mathbb R$ the ``Liouville field'' and we note
\be\label{sigbddef} \sigma = \Sigma_{|\partial\disk}\ee
its boundary value. 

The crucial point, emphasized in \cite{Fer1,Fer3}, is that the Liouville field is not constrained by boundary conditions on $\partial\disk$. One says that $\Sigma$ has ``free'' boundary conditions.\footnote{This terminology might be a bit misleading, since free boundary conditions really means no boundary condition at all.} Actually, in the case of JT gravity in negative or zero curvature, the degrees of freedom are entirely described by the boundary Liouville field $\sigma$ \cite{Fer1,Fer3}. In positive curvature, this is still true semi-classically around a given classical solution, which is the case of interest for the present paper. Constraining the boundary Liouville field in any way (Neumann, Dirichlet or any other way) would be totally inconsistent.\footnote{In the case of the Liouville theory, for which both the bulk and the boundary Liouville fields fluctuate, a subtle coincidence makes the free and the Neumann boundary conditions equivalent, see e.g.\ the discussion in \cite{Fer3}. But we insist that, from first principles, the correct boundary conditions are the free boundary conditions, for both the usual Liouville model and JT gravity.} The bulk and boundary degrees of freedom in the Liouville field must thus be carefully treated as independent variables. Following \cite{Fer1,Fer3}, we perform the explicit separation between the bulk and the boundary by writing 
\be\label{sepBb} \Sigma = \SigB + \Sigma_{\sigma}\, ,\ee
where the bulk piece $\SigB$ satisfies Dirichlet boundary conditions,
\be\label{SigBbc} \Sigma_{\text{B}|\partial\disk} = 0\, ,\ee
and the piece $\Sigma_{\sigma}$ coincides, by definition, with $\sigma$ on the boundary, and is extended into the bulk by using a convenient prescription that fixes the bulk values in terms of the boundary values,
\be\label{Sigbbc} \Sigma_{\sigma|\partial\disk} = \sigma\, ,\quad \Sigma_{\sigma}(x) \ \text{depends only on the boundary $\sigma$ for all $x\in\disk$}\, .\ee
At the conceptual level, the precise prescription for defining $\Sigma_{\sigma}$ is, to a large extent irrelevant, but, in practice, it is chosen to simplify the analysis of the particular problem at hand, see \cite{Fer1,Fer3} and below. 

The gauge fixing \eqref{confgaugedef} has the effect to break the diffeomorphism gauge redundancy down to the $\PslR$ group of automorphisms of the disk, Eq.\ \eqref{pslRdef}, which corresponds to the three-dimensional subgroup of $\diff(\disk)$ that maps a metric in conformal gauge to another metric in conformal gauge.\footnote{Strictly speaking, there is an additional $\mathbb Z_{2}$ redundancy $z\mapsto\bar z$ that remains unbroken in conformal gauge and which produces an overall factor of 2 in the partition function. But overall multiplication constants in $Z$ are unphysical, because they can be absorbed in counterterms; thus, we do not indicate them explicitly.} In conformal gauge, the formal path integral \eqref{Zformal1} reduces to
\be\label{Zformalcg} Z = \frac{1}{\Vol(\PslR)}\int_{\mathscr L_{\text B}\times\mathscr L_{\text b}} D\SigB D\sigma\, 
\delta\Bigl(\int_{0}^{2\pi}e^{\sigma}\d\theta - \ell\Bigr)\, J[g]
Z_{\text{gh}}[g] Z_{\text{CFT}}[g]\, e^{-S_{\La}[g]}\, .\ee
The factor $Z_{\text{gh}}[g]$ is the Faddeev-Popov determinant which, as is well-known, corresponds to the partition function of a ghost CFT of central charge $c_{\text{gh}}=-26$.  The $\delta$-function enforces explicitly the constraint of fixed boundary length
\be\label{fixedlength} \int_{\partial\disk}\d s = \int_{0}^{2\pi}e^{\sigma}\d\theta = \ell\, ,\ee
since we work in this ensemble in the present paper.\footnote{As explained in \cite{Fer3}, imposing the constraint of fixed length in this precise way and not, for instance, using $\smash{\delta\bigl(\frac{1}{\ell}\int_{0}^{2\pi}e^{\sigma}\d\theta - 1\bigr)}$, corresponds to counting the metrics with no marked point on the boundary.} The space $\mathscr L_{\text B}\times\mathscr L_{\text b}$ over which we integrate is the full space of Liouville fields on the closed disk $\bar\disk$, which decomposes into a bulk piece $\mathscr L_{\text B}$ and a boundary piece $\mathscr L_{\text b}$, following \eqref{sepBb}. The measures $D\SigB$ and $D\sigma$ are formal reparameterization invariant and background-independent ultralocal measures that we can identify with the volume forms associated with the metrics 
\begin{align}\label{LBmetricspace} ||\delta\SigB||^{2} & = \int_{\disk}\d^{2}x\sqrt{g}\, (\delta\SigB)^{2} = \int_{0}^{1}\d\rho\, \rho\int_{0}^{2\pi}\d\theta\,  e^{2\Sigma}  (\delta\SigB)^{2}\\
\label{Lbmetricspace} ||\delta\sigma||^{2} & = \int_{\partial\disk}\d s\, (\delta\sigma)^{2} = \int_{0}^{2\pi}\d\theta\, e^{\sigma}  (\delta\sigma)^{2}
\end{align}
on the spaces $\mathscr L_{\text B}$ and $\mathscr L_{\text b}$ respectively. The contribution
\be\label{Jgdef} J[g] = \delta\Bigl(R - \frac{2\eta}{L^{2}}\Bigr)\ee
is an additional piece that enters in the case of JT gravity, enforcing the constraint of constant bulk curvature and coming from integrating out the dilaton field $\Phi$.

\subsubsection{The Liouville action}

The Liouville action on the disk is defined by
\be\label{SLdef} S_{\text L}[\Sigma] = \int_{\disk}\partial_{a}\Sigma\partial_{a}\Sigma\, \d^{2} x + 2\int_{\partial\disk} \sigma\, \d\theta\, .\ee
The metric dependence of the disk partition function $\mathscr Z[g]$ of an arbitrary CFT of central charge $\mathscr C$ in the gauge \eqref{confgaugedef} is given by
\be\label{ZCFTLiou}  \mathscr Z[g] = e^{\frac{\mathscr C}{24\pi}S_{\text L}[\Sigma]}\mathscr Z[\delta]\, .\ee
The combined ghost plus matter system entering in \eqref{Zformalcg} has total central charge $\mathscr C = c-26$. Discarding the overall constant $Z_{\text{gh}}[\delta]Z_{\text{CFT}}[\delta]$, we can thus rewrite \eqref{Zformalcg} as
\be\label{Zformalcg2} Z = \frac{1}{\Vol(\PslR)}\int_{\mathscr L_{\text B}\times\mathscr L_{\text b}} D\SigB D\sigma\, 
\delta\Bigl(\int_{0}^{2\pi}e^{\sigma}\d\theta - \ell\Bigr)\, J[g]\, e^{-S_{\La}[\Sigma]-\frac{26-c}{24\pi}S_{\text L}[\Sigma]}\, .\ee

\subsubsection{Gauge fixing $\PslR$}

The $\PslR$ group of automorphisms of the disk being non-compact, it is necessary to impose additional gauge conditions, that supplement the conformal gauge condition \eqref{confgaugedef}, to make sense of \eqref{Zformalcg2}. As explained in detail in \cite{Fer3}, a convenient set of gauge conditions fixing the non-compact part of $\PslR$ is\footnote{These gauge conditions are actually globally defined, in particular, they don't have Gribov ambiguities.}
\be\label{PsLgc} \int_{0}^{\frac{2\pi}{3}}e^{\sigma}\d\theta =  \int_{\frac{2\pi}{3}}^{\frac{4\pi}{3}}e^{\sigma}\d\theta = \frac{\ell}{3}\, \cdotp\ee
The associated Faddeev-Popov determinant is readily calculated to be $\exp(\sigma(\frac{2\pi}{3}) + \sigma(\frac{4\pi}{3}))$ \cite{Fer3}. Rearranging the $\delta$-functions in a nice symmetric way, we obtain
\begin{multline}\label{Zformalcg3} Z = \int_{\mathscr L_{\text B}\times\mathscr L_{\text b}} D\SigB D\sigma\, 
\delta\Bigl(\int_{0}^{\frac{2\pi}{3}}e^{\sigma}\d\theta - \frac{\ell}{3}\Bigr)\delta\Bigl(\int_{\frac{2\pi}{3}}^{\frac{4\pi}{3}}e^{\sigma}\d\theta - \frac{\ell}{3}\Bigr)\delta\Bigl(\int_{\frac{4\pi}{3}}^{2\pi}e^{\sigma}\d\theta - \frac{\ell}{3}\Bigr)\\
e^{\sigma(\frac{2\pi}{3}) + \sigma(\frac{4\pi}{3})} J[g]\, e^{-S_{\La}[\Sigma]-\frac{26-c}{24\pi}S_{\text L}[\Sigma]}\, .
\end{multline}

\subsubsection{The path integral for Liouville gravity}

The path integral formula for the Liouville partition function $Z_{\text L}$ is simply given by Eq.\ \eqref{Zformalcg3} with $J[g]=1$.

\subsubsection{\label{JTpathsubSec}The path integral for JT gravity}

In the case of JT gravity, the insertion of the functional $\delta$-function \eqref{Jgdef} in the path integral enforces the constant bulk curvature constraint which, in conformal gauge, reads
\be\label{Curvconscg} R - \frac{2\eta}{L^{2}} = 2e^{-2\Sigma}\Delta_{\delta}\Sigma - \frac{2\eta}{L^{2}} = 0\, ,\ee
where $\Delta_{\delta} = -\partial_{a}\partial_{a}$ is the positive Laplacian for the metric $\delta = |dz|^{2}$. The Liouville field thus satisfies, in the bulk, the so-called Liouville equation
\be\label{LioueqSig} \Delta_{\delta}\Sigma = \frac{\eta}{L^{2}}e^{2\Sigma}\, .\ee
As reviewed in \cite{Fer1,Fer3}, in the case of zero or negative curvature, this equation implies that the bulk Liouville field $\Sigma$ is uniquely expressed in terms of the boundary Liouville field $\sigma$,
\be\label{Sigsigrel} \Sigma (x) = \Sigma[\sigma](x)\, .\ee
Note that in zero curvature, the Liouville equation is simply the Laplace equation, which means that $\Sigma$ is a harmonic function. It is then well-known that $\Sigma$ is uniquely specified by its boundary value $\sigma$, and the relation \eqref{Sigsigrel} can be written down explicitly using the Poisson kernel. This is explained and fully exploited in \cite{Fer1,Fer3}. In negative curvature, the Liouville equation is non-linear and thus much more complicated than the Laplace equation. However, the crucial fact that its solution is uniquely specified in terms of the boundary data is still valid \cite{Lioumath}, even though a closed-form formula expressing $\Sigma$ in terms of $\sigma$ no longer exists. This fundamental result shows that, in zero and negative curvature JT gravity, the degrees of freedom are entirely encoded in $\sigma $ alone, a fundamental feature already stressed above. In the case of positive curvature, the situation is more subtle. For instance, as we exemplify in Section \ref{clSec}, the uniqueness of the solution for given boundary value $\sigma$ is violated; we shall find two solutions (called the small and the large disks) corresponding to the same constant boundary value $\sigma$. However, this subtlety is irrelevant for the present paper, because we are performing the $c\rightarrow -\infty$ semi-classical expansion around a particular classical solution, and for this expansion the existence of other saddles does not play any role.

This being understood, it is clear \cite{Fer1,Fer3} that the most natural choice of the decomposition \eqref{sepBb} in JT gravity is obtained by setting
\be\label{SigbdefJT} \Sigma_{\sigma} = \Sigma[\sigma]\, .\ee
The $\delta$-function in \eqref{Jgdef} then imposes $\SigB=0$. The path integral over $\SigB$ is then evaluated straightforwardly by expanding the argument of the $\delta$-function to linear order in $\SigB$,
\be\label{XilinJg} R - \frac{2\eta}{L^{2}} = 2\Bigl(\Delta_{\sigma}-\frac{2\eta}{L^{2}}\Bigr)\SigB + O(\SigB^{2})\, ,\ee
where $\Delta_{\sigma}$ is the positive Laplacian for the constant curvature metric with boundary Liouville field $\sigma$. We thus get
\be\label{deltafuncfin} J[g] = \frac{\delta(\SigB)}{2|\det_{\text D}(\Delta_{\sigma} - \frac{2\eta}{L^{2}})|} \ee
where $\det_{\text D}(\Delta_{\sigma} - \frac{2\eta}{L^{2}})$ denotes the functional determinant of $\Delta_{\sigma} - \frac{2\eta}{L^{2}}$ with Dirichlet boundary conditions, consistently with \eqref{SigBbc}. Integrating out $\SigB$ in \eqref{Zformalcg3} thus yields a simplified formula for the JT gravity partition function in terms of a path integral over the boundary Liouville field only,
\begin{multline}\label{ZpathJT} Z^{(\eta)} = \int_{\mathscr L_{\text b}} D\sigma\, 
\delta\Bigl(\int_{0}^{\frac{2\pi}{3}}e^{\sigma}\d\theta - \frac{\ell}{3}\Bigr)\delta\Bigl(\int_{\frac{2\pi}{3}}^{\frac{4\pi}{3}}e^{\sigma}\d\theta - \frac{\ell}{3}\Bigr)\delta\Bigl(\int_{\frac{4\pi}{3}}^{2\pi}e^{\sigma}\d\theta - \frac{\ell}{3}\Bigr)\\
e^{\sigma(\frac{2\pi}{3}) + \sigma(\frac{4\pi}{3})} \, 
\frac{e^{-S_{\La}[\Sigma_{\sigma}]-\frac{26-c}{24\pi}S_{\text L}[\Sigma_{\sigma}]}}{|\det_{\text D}(\Delta_{\sigma} - \frac{2\eta}{L^{2}})|}\,\cdotp
\end{multline}
\subsection{\label{loopSec}One-loop path integrals}

\subsubsection{The case of Liouville gravity}

We now focus on a one-loop calculation in the limit \eqref{Zamolodlimit}. In this limit, it is natural to group together the terms proportional to $|c|$ in the action, writing
\be\label{Seff1} S_{\La}+\frac{26-c}{24\pi}S_{\text L} = \frac{|c|}{24\pi} \Seff + \frac{13}{12\pi}S_{\text L}\ee
for an effective action
\be\label{Seffdef} \Seff[\Sigma] =S_{\text L}[\Sigma] + \mu\int_{\disk}e^{2\Sigma}\,\d^{2}x = \int_{\disk}\bigl(\partial_{a}\Sigma\partial_{a}\Sigma + \mu e^{2\Sigma}\bigr)\, \d^{2} x + 2\int_{\partial\disk} \sigma\, \d\theta\, .\ee
Recall that $\mu$ was defined in \eqref{Zamolodlimit}. The path integral is dominated by a saddle, computed in Section \ref{clSec}, that we note $\Sigma_{*}$. More generally, we note with a $*$ any quantity evaluated on the saddle. We expand around the saddle by writing
\be\label{Sigsadexp} \Sigma = \Sigma_{*} + \frac{\Xi}{\sqrt{|c|}}\ee
which, using the decomposition \eqref{sepBb}, is equivalent to expanding
\be\label{Sigsadexp2} \SigB = \Sigma_{\text{B}*} + \frac{\XiB}{\sqrt{|c|}}\,\cvp\quad \sigma = \sigma_{*}+ \frac{\xi}{\sqrt{|c|}}\, \cdotp\ee

There are then two equivalent ways to deal with the $\delta$-functions in the path integral \eqref{Zformalcg3}. A first way amounts to using the integral representation $\delta(x) =\frac{1}{2\pi} \int e^{i\lambda x}\d\lambda$ in terms of a Lagrange multiplier $\lambda$. One introduces a Lagrange multiplier for each $\delta$-function, adds the associated terms to the effective action \eqref{Seffdef} and, together with the fields $\SigB$ and $\sigma$, expands the Lagrange multipliers around their saddle point values. The advantage of this approach is that, by construction, the expansion of the effective action does not contain linear terms. It can also be used systematically to any loop order.

A second way, which is a bit more straightforward at one-loop, is to deal directly with the $\delta$ functions, without introducing Lagrange multipliers. The expansion of the effective action \eqref{Seffdef}, to the order we need for a one-loop calculation, has the form
\be\label{Seffquad} \Seff[\Sigma] = S_{\text{eff}*} + \frac{1}{|c|^{1/2}}\Seff^{(1)}[\xi] + \frac{1}{|c|}\Seff^{(2)}[\XiB,\xi] + O\bigl(1/|c|^{3/2}\bigr)\, ,\ee
where the linear term $\Seff^{(1)}$ turns out to depend on the boundary field fluctuations $\xi$ only.\footnote{This is easy to check, see Section \ref{oneloopSec}.} The arguments of the $\delta$ functions are themselves expanded as
\be\label{expdelta}
\int_{\frac{2\pi k}{3}}^{\frac{2\pi(k+1)}{3}}e^{\sigma}\d\theta - \frac{\ell}{3}= \frac{1}{\sqrt{|c|}}\Biggl[\int_{\frac{k\ell}{3}}^{\frac{(k+1)\ell}{3}}\xi\, ds_* + \frac{1}{2\sqrt{|c|}}
\int_{\frac{k\ell}{3}}^{\frac{(k+1)\ell}{3}}\xi^{2}\, ds_* + O\bigl(1/|c|\bigr)\Biggr]\, .\ee
At one-loop, the $\delta$-function constraints thus ensure that we can make the substitutions 
\be\label{substidelta} \int_{\frac{k\ell}{3}}^{\frac{(k+1)\ell}{3}}\xi\, ds_* \longmapsto -\frac{1}{2|c|^{1/2}}\int_{\frac{k\ell}{3}}^{\frac{(k+1)\ell}{3}}\xi^{2}\, ds_*\ee
in the effective action \eqref{Seffquad} without changing the value of the path integral. This allows to kill all the linear terms in the expansion. At the one-loop order, we are left with a purely quadratic action that we denote by $\tilde S^{(2)}_{\text{eff}}$. The quadratic or higher order terms in the $\delta$-functions can then be dropped. 

Moreover, still at one-loop, the non-linear metrics \eqref{LBmetricspace}, \eqref{Lbmetricspace} can be replaced by the usual linear metrics
\be\label{oneloopmet} ||\delta\XiB||_{*}^{2} = \frac{1}{|c|}\int_{\disk}\d^{2}x\sqrt{g_{*}}\, (\delta\XiB)^{2}\, ,\quad ||\delta\xi||_{*}^{2} = \frac{1}{|c|}\int_{\partial\disk}\d s_{*}\, (\delta\xi)^{2}\, ,\ee
which yield the measures $D_{*}\XiB$ and $D_{*}\xi$. Putting everything together in \eqref{Zformalcg3}, we finally get
\be\label{ZLoneloop1} \ln Z_{\text L} = \ln Z_{\text{L},\,\text{tree}} + \ln Z_{\text{L},\,\text{1-loop}} + O(1/|c|)\ee
with
\be\label{Ztree1} Z_{\text{L},\,\text{tree}} = e^{-\frac{|c|}{24\pi} S_{\text{eff}*}}\ee
and
\begin{multline}\label{Zloop1} Z_{\text{L},\,\text{1-loop}} = e^{-\frac{13}{12\pi}S_{\text L *}+\sigma_{*}(\frac{2\pi}{3}) + \sigma_{*}(\frac{4\pi}{3})}\\ \int D_{*}\XiB D_{*}\xi\,
\delta\Bigl(\int_{0}^{\frac{\ell}{3}}\xi\, \d s_{*}\Bigr)\delta\Bigl(\int_{\frac{\ell}{3}}^{\frac{2\ell}{3}}\xi\, \d s_{*}\Bigr)\delta\Bigl(\int_{\frac{2\ell}{3}}^{\ell}\xi\,\d s_{*}\Bigr)\, e^{-\frac{1}{24\pi}\tilde{S}^{(2)}_{\text{eff}}[\XiB,\xi]}\, .
\end{multline}

\subsubsection{The case of JT gravity}

The same analysis can be repeated for JT gravity, starting from \eqref{ZpathJT}. The main difference with the Liouville gravity case is that we only expand the boundary field
\be\label{sigexp} \sigma = \sigma_{*}+\frac{\xi}{\sqrt{|c|}}\, \cvp\ee
the bulk fluctuations being fixed in terms of $\xi$ as explained in subsection \eqref{JTpathsubSec}. We obtain
\be\label{ZJToneloop1} \ln Z^{(\eta)} = \ln Z^{(\eta)}_{\text{tree}} + \ln Z^{(\eta)}_{\text{1-loop}} + O(1/|c|)\ee
with
\be\label{Ztree2} Z^{(\eta)}_{\text{tree}} = e^{-\frac{|c|}{24\pi} S_{\text{eff}*}}\ee
and
\begin{multline}\label{ZJTloop1} Z^{(\eta)}_{\text{1-loop}} = \frac{e^{-\frac{13}{12\pi}S_{\text L *}+\sigma_{*}(\frac{2\pi}{3}) + \sigma_{*}(\frac{4\pi}{3})}}{|\det_{\text D}(\Delta_{*}-\frac{2\eta}{L^{2}})|}
\\ \int D_{*}\xi\,
\delta\Bigl(\int_{0}^{\frac{\ell}{3}}\xi\, \d s_{*}\Bigr)\delta\Bigl(\int_{\frac{\ell}{3}}^{\frac{2\ell}{3}}\xi\, \d s_{*}\Bigr)\delta\Bigl(\int_{\frac{2\ell}{3}}^{\ell}\xi\,\d s_{*}\Bigr) e^{-\frac{1}{24\pi}\tilde{S}^{(2)}_{\text{eff}}[\xi]}\, .
\end{multline}

The similarity between \eqref{Zloop1} and \eqref{ZJTloop1} is startling, but there is also a fundamental difference: in Liouville gravity, we have to integrate over both bulk and boundary metric fluctuations, whereas in the case of JT only the boundary degrees of freedom fluctuate.

We now turn to the explicit evaluation of the path integrals \eqref{Zloop1} and \eqref{ZJTloop1}. The saddle-point solutions and the associated tree-level contributions are given in the next Section. The computation of the quadratic fluctuations $\tilde{S}^{(2)}_{\text{eff}}$ is presented in Section \ref{oneloopSec}. The calculation of the associated Gaussian path integrals, which amounts to evaluating some functional determinants, is detailed in Section \ref{detSec}, using in particular results that are published separately \cite{detpaper}.

\section{\label{clSec}Leading order: classical solutions}

\subsection{Equations of motion}

From the discussion of the previous section, we know that, at leading order in the $c\rightarrow -\infty$ limit, the path integral is dominated by the minimum of the effective action \eqref{Seffdef}, taking into account the constraint of fixed boundary length
\be\label{clcons1} \int_{0}^{2\pi}e^{\sigma}\d\theta = \ell\ee
and, in the case of JT gravity, the additional constraint of constant bulk curvature
\be\label{clcons2} R =- 2e^{-2\Sigma}\partial_{a}\partial_{a}\Sigma  = \frac{2\eta}{L^{2}}\, \cvp \quad \text{$\eta = -1$, $1$ or $0$}\, .\ee
To impose these constraints at the classical level, it is convenient to introduce Lagrange multipliers $\zeta$ and $\varphi$. We thus use the classical action
\be\label{Sclanadef} S = \Seff[\Sigma] + \zeta\Bigl(\int_{0}^{2\pi}e^{\sigma}\d\theta - \ell\Bigr) + \vartheta\int_{\disk}\d^{2}x\, \varphi\Bigl(\partial_{a}\partial_{a}\Sigma + \frac{\eta}{L^{2}}e^{2\Sigma}\Bigr)\, ,\ee
where $\vartheta =0$ in the case of the Liouville theory and $\vartheta =1$ in the case of the JT theory. The Lagrange multiplier field $\varphi$ satisfies Dirichlet boundary conditions
\be\label{varphiDirich} \varphi_{|\partial\disk}=0\, ,\ee
since it imposes the constant curvature constraint in the bulk only. The Liouville field $\Sigma$ has free boundary conditions, as already emphasized several times above. The variation of \eqref{Sclanadef} with respect to $\Sigma$ reads
\begin{multline}\label{varScl} \delta S = \int_{\disk}\d^{2}x\,\biggl[-2\partial_{a} \partial_{a}\Sigma + 2\Bigl(\mu + \frac{\vartheta\eta}{L^{2}}\varphi\Bigr) e^{2\Sigma} + \vartheta\partial_{a}\partial_{a}\varphi\biggr]\delta\Sigma \\ +
\delta\zeta \Bigl(\int_{0}^{2\pi}e^{\sigma}\d\theta - \ell\Bigr) + 
\vartheta\int_{\disk}\d^{2}x\, \delta\varphi\Bigl(\partial_{a}\partial_{a}\Sigma + \frac{\eta}{L^{2}}e^{2\Sigma}\Bigr) \\ +
\int_{0}^{2\pi}\d\theta\, \biggl[\Bigl(2\frac{\partial\Sigma}{\partial\rho}(\rho=1,\theta) - \vartheta\frac{\partial\varphi}{\partial\rho}(\rho=1,\theta)+\zeta e^{\sigma}+2\Bigr)\delta\sigma + \vartheta\varphi\frac{\partial\delta\Sigma}{\partial\rho}(\rho=1,\theta)\biggr]\, .
\end{multline}
From this we derive the field equations for the saddle point solution, which is denoted with a $*$ consistently with our previous conventions,
\begin{align}\label{eom1} &\partial_{a}\partial_{a}\Sigma_{*} = \frac{1}{2}\vartheta\partial_{a}\partial_{a}\varphi_{*} + \Bigl(\mu + \frac{\vartheta\eta}{L^{2}}\varphi_{*}\Bigr) e^{2\Sigma_{*}}\, ,\\
\label{eom2} &\frac{\partial\Sigma_{*}}{\partial\rho}(\rho=1,\theta) + 1 = \frac{1}{2} \vartheta\frac{\partial\varphi_{*}}{\partial\rho}(\rho=1,\theta)-\frac{1}{2}\zeta_{*} e^{\sigma_{*}}\, ,
\end{align}
together with the constraint \eqref{clcons1} and, for JT, the additional constraint \eqref{clcons2}. Note that Eq.\ \eqref{eom2} comes from the boundary term in \eqref{varScl}, which does not automatically vanish because the Liouville field has free boundary conditions. The fact that boundary terms in the variation of the action yield additional equations of motion may seem unfamiliar. Traditionally, one imposes boundary conditions precisely in such a way that the boundary terms identically vanish. However, this is a general feature of the variational principle in conformal gauge and it is perfectly consistent. A detailed discussion of this point is provided in \cite{Fer3}.

The analysis of the equations of motion \eqref{eom1} and \eqref{eom2} are simplified by rewritting them in a more geometric way. Note that the bulk curvature $R$ of the metric $g=e^{2\Sigma}\delta$ is expressed in terms of $\Sigma$ as in Eq.\ \eqref{clcons2}. Similarly, the extrinsic curvature of the disk boundary for the same metric $g$ is given by
\be\label{kSigform} k = e^{-\sigma}\Bigl(\frac{\partial\Sigma}{\partial\rho}(\rho=1,\theta) + 1\Bigr)\, .\ee
Note also that the positive Laplacian in the metric $g$ is 
\be\label{posLapla} \Delta = -e^{-2\Sigma}\partial_{a}\partial_{a}\ee
and the normal derivative on the boundary
\be\label{normalder} \frac{\partial}{\partial n} = e^{-\sigma}\frac{\partial}{\partial\rho}\,\cdotp\ee
Equipped with these formulas, we find that \eqref{eom1} and \eqref{eom2} are equivalent to
\begin{align}\label{eomgeom1} R_{*} & = -2\mu + \vartheta\Bigl(\Delta_{*} - \frac{2\eta}{L^{2}}\Bigr)\varphi_{*}\, ,\\\label{eomgeom2}
k_{*} & = -\frac{1}{2}\zeta_{*} + \frac{1}{2}\vartheta\frac{\partial\varphi_{*}}{\partial n_{*}}\,\cdotp
\end{align}

Finally, let us mention that the equations of motion in conformal gauge always have a family of solutions related by the $\PslR$ automorphisms of the disk. This redundancy is fixed by further imposing the gauge conditions \eqref{PsLgc}.

Now let us turn to the analysis of the different cases we are interested in.

\subsection{\label{saddleLSec}Solution for Liouville gravity}

Recall that $\mu\geq 0$ in Liouville gravity. We choose $\mu>0$, since the case $\mu=0$ can be obtained by taking the limit $\mu\rightarrow 0^{+}$. Eq.\ \eqref{eomgeom1}, $R_{*}=-2\mu$,  implies that the solution has constant negative curvature. Locally, it is given by a patch of hyperbolic space $\Htwo$. Eq.\ \eqref{eomgeom2} implies that $k_{*}=-\frac{1}{2}\zeta_{*}$ is also a constant. The disk boundary can thus be seen as a curve of constant extrinsic curvature in $\Htwo$ whose interior must be the disk. Together with the conditions \eqref{PsLgc}, this identifies uniquely the metric to be $g=\delta_{\ell}^{-}$ given in Eq.\ \eqref{delmindef}, with $L^{2}= 1/\mu$. Note that one may also use a more pedestrian approach to solve the equations of motion, starting from a radial ansatz $\Sigma= \Sigma(\rho)$ and directly solving the resulting differential equations. Explicitly, using the parameter $x$ defined in \eqref{defxpara}, we have
\be\label{Liousaddlesol} \Sigma_{*} = \ln\frac{2r_{0}}{\sqrt{\mu}} - \ln \bigl(1-r_{0}^{2}\rho^{2}\bigr)\ee
for
\be\label{r0saddleL} r_{0} = \frac{1}{x}\bigl(-1+\sqrt{1+x^{2}}\bigr)\, .\ee

It is then straightforward to evaluate, on the saddle point solution, the area, the Liouville action,
\be\label{Lsaddleactions1}  A_{*} = \frac{4\pi}{\mu}\frac{r_{0}^{2}}{1-r_{0}^{2}} = \frac{\ell r_{0}}{\sqrt{\mu}}\,\cvp \quad S_{\text L *} = \ell\sqrt{\mu}\, r_{0} + 4\pi\ln\frac{2r_{0}}{\sqrt{\mu}} \, \cvp\ee
and the effective action
\begin{multline}
\label{Lsaddleactions2} 
S_{\text{eff}*} = S_{\text L*} + \mu A_{*} =2\ell\sqrt{\mu}\, r_{0} + 4\pi\ln\frac{2r_{0}}{\sqrt{\mu}}  \\= 4\pi\ln\frac{\ell}{2\pi}+4\pi\Bigl[\sqrt{1+x^{2}}-1-\ln\frac{1+\sqrt{1+x^{2}}}{2}\Bigr]\, .
\end{multline}
Using \eqref{Ztree1}, this yields the tree-level partition function $\ln Z_{\text{L},\,\text{tree}}(\ell,\mu) = -\frac{|c|}{24\pi}S_{\text{eff}*} = |c|f_{0}$, the tree-level coefficient $f_{0}$ being precisely as in \eqref{ZLloopcoeff0}.

\subsection{\label{clsolJTSec}Solutions for JT gravity}

The constraint \eqref{clcons2} implies that, locally, the solution is given by a patch of hyperbolic space, Euclidean space or the round two-sphere, if $\eta=-1$, $\eta=0$ or $\eta=+1$, respectively. The equations of motion \eqref{eomgeom1} and \eqref{eomgeom2} are then equivalent to
\begin{align}\label{eomJT1} & \Bigl(\Delta_{*}-\frac{2\eta}{L^{2}}\Bigr)\varphi_{*} = 2\mu + \frac{2\eta}{L^{2}}\, \cvp\\\label{eomJT2}
&R_{*}=\frac{2\eta}{L^{2}}\, \cvp\quad  k_{*} = -\frac{1}{2}\zeta_{*} + \frac{1}{2}\frac{\partial\varphi_{*}}{\partial n_{*}}\,\cvp
\end{align}
together with the boundary length constraint \eqref{clcons1}.

Let us note that, if we had not included the Liouville action term in the effective action \eqref{Seffdef}, the equations of motion would take the simplified form
\begin{align}\label{JTnotL1} & \Bigl(\Delta_{*}-\frac{2\eta}{L^{2}}\Bigr)\varphi_{*} = 2\mu\, ,\\\label{JTnotL2}
&R_{*}=\frac{2\eta}{L^{2}}\, \cvp\quad \frac{\partial\varphi_{*}}{\partial n_{*}} = \zeta_{*}\, .
\end{align}
One can straightforwardly check that, as expected, these equations also follow from varying the traditional dilaton action \eqref{Sdil} in conformal gauge, with the identification
\be\label{varphiPhirel} \Phi = \frac{|c|}{3}\bigl(\eta\mu L^{2} +\varphi\bigr)\, .\ee
The solutions to the JT classical equations of motion are well-known, even though their analysis in conformal gauge, which is presented in \cite{Fer3}, may be unfamiliar. They imply that $\smash{\frac{\epsilon^{\mu\nu}}{\sqrt{g_{*}}}\partial_{\nu}\varphi}$ is a Killing vector of the metric.\footnote{In conformal gauge, this is obtained by combining Eq.\ \eqref{JTnotL1} and the fact that $\varphi_{*}$ satisfies both Dirichlet, Eq.\ \eqref{varphiDirich}, and Neumann, second equation in \eqref{JTnotL2}, boundary conditions, see \cite{Fer3} for detailed explanations.} The form of the Killing vectors in constant curvature are of course well-known and, together with the Dirichlet boundary condition \eqref{varphiDirich} and the fact that the boundary must be topologically $\Sone$, this fixes the classical metrics to be $\delta_{\ell}^{-}$, $\delta_{\ell}$ or one of the two $\delta_{\ell}^{+}$, as given by Eq.\ \eqref{delmindef}, \eqref{deltaldef} and \eqref{delplusdef}, when $\eta=-1$, $0$ or $+1$, respectively. In particular, in positive curvature, a smooth classical solution exists if and only if $\ell\leq 2\pi L$.

When the Liouville action is taken into account, the Neumann constraint $\partial\varphi_{*}/\partial n_{*} = \text{constant}$ on the boundary is replaced by the second equation in \eqref{eomJT2}. This more general equation allows a priori for a normal derivative varying along the boundary, and thus also a varying extrinsic curvature. Nevertheless, it is clear that the classical solutions for the metric in  the JT theory will remain solutions of our more general set of equations of motion. Indeed, if we \emph{assume} that $\partial\varphi_{*}/\partial n_{*}$ is a constant, then clearly the two sets of equations of motion \eqref{eomJT1}, \eqref{eomJT2} and \eqref{JTnotL1}, \eqref{JTnotL2} are equivalent, up to trivial constant shifts in $\varphi_{*}$ and $\zeta_{*}$.

The fact that the equations \eqref{eomJT1} and \eqref{eomJT2} may have more general solutions,  breaking the $\uone$ symmetry of the round metrics $\delta_{\ell}$ or $\delta_{\ell}^{\pm}$, is an interesting possibility. In particular, one may expect that such solutions become relevant in the positive curvature case when $\ell>2\pi L$.\footnote{The breaking of the $\uone$ invariance in this context is reminiscent of a Gross-Witten phase transition  \cite{Fer3}.} However, for the purpose of the present work, we shall assume that the relevant solutions are the symmetric ones. In particular, we limit ourselves to the case $\ell\leq 2\pi L$ in the positive curvature theory. A discussion of the local and global stability of the symmetric solutions will be presented below, especially in Section \ref{oneloopSec}.

To summarize, the saddle-point solutions and saddle point values of the action are given as follows.

\subsubsection{Negative curvature} The metric being $g_{*}=\delta_{\ell}^{-}$ as in Eq.\ \eqref{delmindef}, we get
\begin{align}\label{JTminsaddlesol1} \Sigma_{*}^{-} & = \ln(2Lr_{0}) - \ln \bigl(1-r_{0}^{2}\rho^{2}\bigr)\\\label{JTminsaddlesol2} \varphi_{*}^{-} & = \frac{2(\mu L^{2} -1)r_{0}^{2}}{1+r_{0}^{2}}\frac{1-\rho^{2}}{1-r_{0}^{2}\rho^{2}}\,\cvp
\end{align}
with $r_{0}$ given in \eqref{ellminusrel}. This yields 
\be\label{SLsadJTmin}S_{\text L *}^{-} = \frac{\ell}{L}\, r_{0} + 4\pi\ln(2Lr_{0})\, .\ee
It is convenient to introduce the dimensionless parameters
\be\label{xydefJTmin} x = \frac{\ell}{2\pi L}\,\cvp \quad y = \mu L^{2}\, .\ee
Using the formula \eqref{Akmin} for the area, we get $S_{\text{eff}*}^{-} = S_{\text L*}^{-}+\mu A_{*}^{-}$, which yields the tree-level partition function $\ln Z^{-}_{\text{tree}} = -\frac{|c|}{24\pi}S_{\text{eff}*}^{-} = |c|f^{-}_{0}$, with the tree-level coefficient $f^{-}_{0}$ given by
\be\label{fzeromin} f_{0}^{-} =-\frac{1}{6}\ln\frac{\ell}{2\pi}  + \frac{1}{6}\ln\frac{1+\sqrt{1+x^{2}}}{2}- \frac{1}{12}(1+y)\bigl(\sqrt{1+x^{2}}-1\bigr)\, .\ee

\subsubsection{Zero curvature}

The metric is the flat disk \eqref{deltaldef}, which yields
\begin{align}\label{JTzerosaddlesol} &\Sigma_{*}^{0}  = \ln\frac{\ell}{2\pi}\,\cvp\quad \varphi_{*}^{0}  = \frac{\mu \ell^{2}}{8\pi^{2}}(1-\rho^{2})\, ,\\\label{SLsadzero}
& S_{\text L*}^{0} = 4\pi\ln\frac{\ell}{2\pi}\, \cvp\quad A_{*}^{0} = \frac{\ell^{2}}{4\pi}
\end{align}
and thus $\ln Z_{\text{tree}}^{0} = -\frac{|c|}{24\pi}S_{\text{eff}*}^{0} = -\frac{|c|}{24\pi}(S_{\text L*}^{0}+\mu A_{*}^{0}) = |c|f_{0}^{0}$ for a tree-level coefficient $f_{0}^{0}$ given by
\be\label{fzerozero} f_{0}^{0} = -\frac{1}{6}\ln\frac{\ell}{2\pi}-\frac{\mu\ell^{2}}{96\pi^{2}}\,\cdotp\ee
\begin{figure}
\centerline{\includegraphics[width=4in]{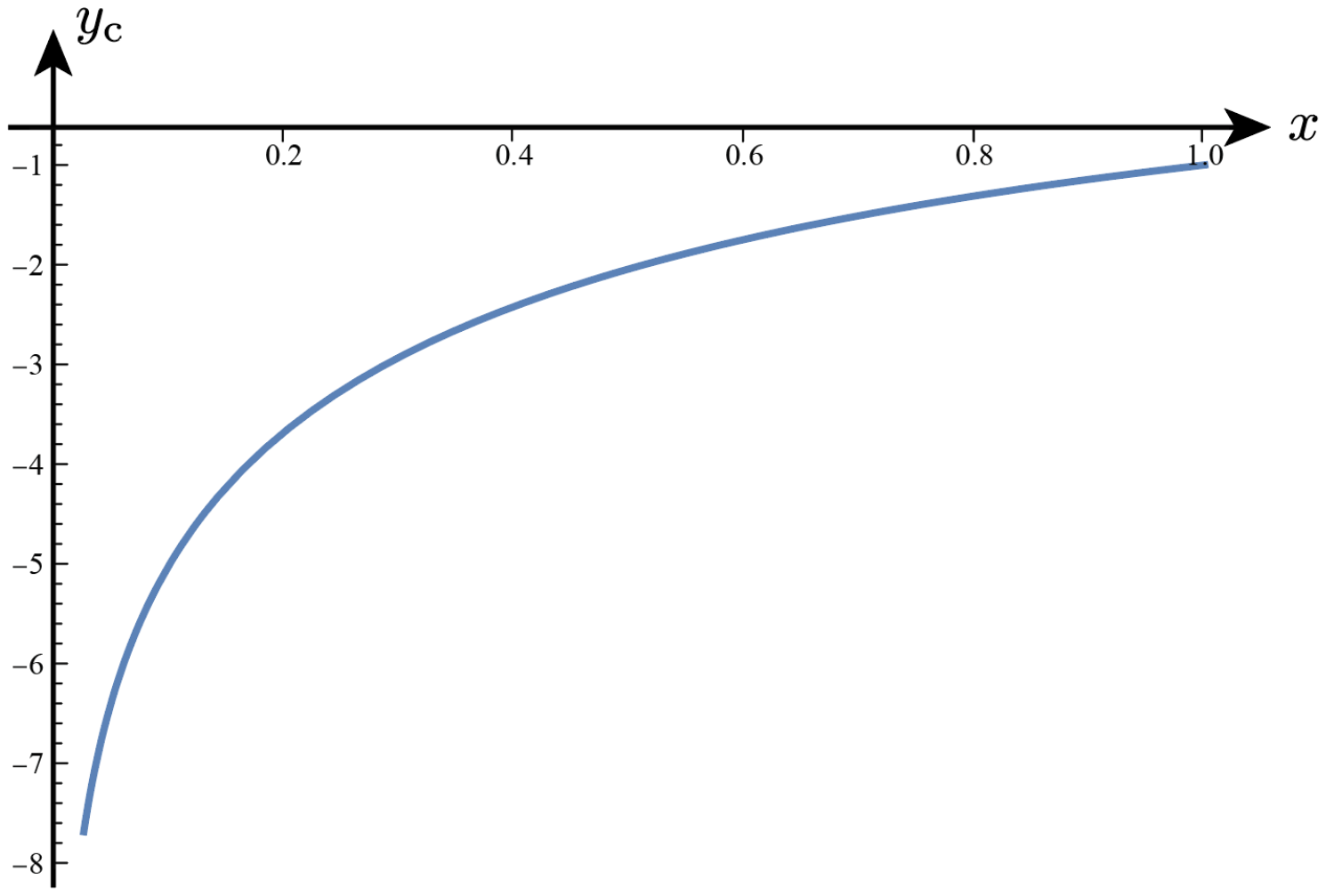}}
\caption{The critical line $y_{\text c}(x)$, defined by Eq.\ \eqref{smalldiskdom}, separating the region where the small disk dominates (above the curve) and the region where the large disk dominates (below the curve).\label{ycFig}}
\end{figure}

\subsubsection{Positive curvature} We find two distinct solutions in this case, corresponding to the large and small disks depicted in Fig.\ \ref{posdisFig}. The metric is of the form $g_{*}=\delta_{\ell}^{+}$ as in Eq.\ \eqref{delplusdef}, i.e.
\begin{align}\label{JTplussaddlesol1} \Sigma_{*}^{+} & = \ln(2Lr_{0}) - \ln \bigl(1+r_{0}^{2}\rho^{2}\bigr)\\\label{JTplussaddlesol2} \varphi_{*}^{+} & = \frac{2(\mu L^{2} +1)r_{0}^{2}}{1-r_{0}^{2}}\frac{1-\rho^{2}}{1+r_{0}^{2}\rho^{2}}\,\cvp
\end{align}
with $r_{0}$ given by one of the two formulas in \eqref{ellplusrel}. This yields 
\be\label{SLsadJTplus}S_{\text L *}^{+} = -\frac{\ell}{L}\, r_{0} + 4\pi\ln(2Lr_{0})\, .\ee
Using the formula \eqref{Akplus} for the area, we get $S_{\text{eff}*}^{+} = S_{\text L*}^{+}+\mu A_{*}^{+}$; in terms of the parameters $x$ and $y$ defined in \eqref{xydefJTmin}, this is
\begin{align}\label{Seffplussadbig} S_{\text{eff}*}^{+>} = 4\pi\ln L + 2\pi(y-1)\bigl(1+\sqrt{1-x^{2}}\bigr) - 4\pi\ln\frac{1-\sqrt{1-x^{2}}}{2x}\\
\label{Seffplussadsmall} S_{\text{eff}*}^{+<} = 4\pi\ln L + 2\pi(y-1)\bigl(1-\sqrt{1-x^{2}}\bigr) - 4\pi\ln\frac{1+\sqrt{1-x^{2}}}{2x}
\end{align}
for the large and small disk solutions, respectively. 

The relevant solution is the one which has the smallest effective action. One finds that the small disk dominates if and only if
\be\label{smalldiskdom} y>y_{\text c}(x) =1-\frac{1}{\sqrt{1-x^{2}}}\ln\frac{1+\sqrt{1-x^{2}}}{1-\sqrt{1-x^{2}}}\,\cvp\ee
the crossover between the small and the large disk occurring when the strict inequality in the above equation is replaced by an equality. We have plotted in Fig.\ \ref{ycFig} the critical line $y_{\text c}$ as a function of $x$. The interpretation is quite simple. If the only term present in the effective action were the area term, then the small disk would be favoured when $y>0$ and the large disk would be favoured when $y<0$. The plot shows that the Liouville piece in the effective action tends to stabilise the small disk, especially for small boundary length. Note that, if $y\geq y_{\text c}(1)=-1$, then the small disk always dominates, for all allowed boundary lengths. Stability issues are further discussed in the next section. 

When the small disk dominates, i.e.\ when $y>y_{\text c}$, the tree-level partition function is given by $\ln Z^{+}_{\text{tree}} = -\frac{|c|}{24\pi}S_{\text{eff}*}^{+<} = |c|f^{+}_{0}$, with a tree-level coefficient $f^{+}_{0}$ given by
\be\label{fzeroplus} f_{0}^{+} =-\frac{1}{6}\ln\frac{\ell}{2\pi}  + \frac{1}{6}\ln\frac{1+\sqrt{1-x^{2}}}{2}- \frac{1}{12}(y-1)\bigl(1-\sqrt{1-x^{2}}\bigr)\, ,\quad y>y_{\text c}\, .\ee
Similarly, when the large disk dominates, i.e.\ when $y<y_{\text c}$, the same tree level coefficient is given by
\be\label{fzeroplusbis} f_{0}^{+}=-\frac{1}{6}\ln\frac{\ell}{2\pi}  + \frac{1}{6}\ln\frac{1-\sqrt{1-x^{2}}}{2}- \frac{1}{12}(y-1)\bigl(1+\sqrt{1-x^{2}}\bigr)\, ,\quad y<y_{\text c}\, .\ee
\section{\label{quadSec}One-loop order: quadratic fluctuations}
\subsection{\label{quadratic action:Liouville}The quadratic action $\widetilde{S}^{(2)}_{\text{eff}}$ for Liouville gravity}

We expand the effective action \eqref{Seffdef} to quadratic order, using \eqref{Sigsadexp}, the field equation \eqref{eom1} and the explicit solution \eqref{Liousaddlesol}. This yields, in the notations of Eq.\ \eqref{Seffquad}, after using an integration by part,
\begin{align}\label{seff1L} \Seff^{(1)} & = \frac{2(1+r_0^2)}{1-r_0^2}\int_{\partial\disk}\xi\, \d \theta \, ,\\
\label{seff2L} \Seff^{(2)} & = \int_{\disk}\Big(\partial_{a}\Xi\partial_{a}\Xi +2\mu e^{2\Sigma_*}\Xi^2\bigr)\, \d^{2} x\ .
\end{align}
By performing the substitution indicated in \eqref{substidelta}, we get the purely quadratic action
\be\label{SefftilL} {\tilde S}_{\text{eff}}^{(2)} = \int_{\disk}\Big(\partial_{a}\Xi\partial_{a}\Xi +2\mu e^{2\Sigma_*}\Xi^2\bigr)\, \d^{2} x-\frac{1+r_0^2}{1-r_0^2}\int_{\partial\disk}\xi^2\, \d \theta\, .\ee

We now carry out the explicit splitting between the bulk and boundary degrees of freedom, in accordance with Eq.\ \eqref{sepBb}. We choose to define $\Sigma_{\sigma}$ by the conditions
\be\label{SsdefLiou} \partial_{a}\partial_{a}\Sigma_{\sigma} = \mu e^{2\Sigma_{\sigma}}\, ,\quad \Sigma_{\sigma|\partial\disk}=\sigma\, .\ee
As required, this prescription fixes uniquely $\Sigma_{\sigma}$ in terms of the boundary field $\sigma$. The advantage of this particular choice is that the bulk and the boundary degrees of freedom decouple in the action \eqref{SefftilL}. Indeed, expanding $\sigma$ as in \eqref{Sigsadexp2} yields a corresponding expansion
\be\label{Sigsexp} \Sigma_{\sigma} = \Sigma_{\sigma*} + \frac{\Xi_{\sigma}}{\sqrt{|c|}}\ee
for $\Sigma_{\sigma}$. At leading order at large $|c|$, which is all what we need for our one-loop calculation, the conditions \eqref{SsdefLiou} are then equivalent to
\be\label{Ssexpcond} \bigl(\Delta_{*}+2\mu\bigr)\Xi_{\sigma}=\bigl(-e^{-2\Sigma_{*}}\partial_{a}\partial_{a}+2\mu\bigr)\Xi_{\sigma} = 0\, ,\quad \Xi_{\sigma|\partial\disk} = \xi\, .\ee
Plugging $\Xi = \Xi_{\text B} + \Xi_{\sigma}$ in \eqref{SefftilL}, and using \eqref{Ssexpcond}, one then finds that the terms proportional to $\Xi_{\text B}\Xi_{\sigma}$ exactly cancel. We are left with
\be\label{SeffLBbsep} {\tilde S}_{\text{eff}}^{(2)}[\Xi_{\text B},\xi] = {\tilde S}_{\text{eff,B}}^{(2)}[\Xi_{\text B}]+{\tilde S}_{\text{eff,b}}^{(2)}[\xi]\, ,\ee
for purely bulk and boundary quadratic actions given by
\begin{align} \label{SeffLB} {\tilde S}_{\text{eff,B}}^{(2)}[\Xi_{\text B}] & = \int_{\disk}\d^{2}x\sqrt{g_{*}}\,\Xi_B \bigl( \Delta_*+2\mu \bigr)\Xi_B\, ,\\ 
\label{SeffLb}{\tilde S}_{\text{eff,b}}^{(2)}[\xi] & = \int_{\partial\disk} \xi (\partial_{\rho}\Xi_{\sigma})_{|\partial\disk} -\frac{1+r_0^2}{1-r_0^2}\int_{\partial\disk} \xi^2\, \d \theta\, .
\end{align}
Note that, by definition, and consistently with \eqref{SigBbc}, $\Xi_{\text B}$ satisfies Dirichlet boundary conditions,
\be\label{XiBDirichlet} \Xi_{\text B |\partial\disk} = 0\, .\ee
The boundary action can be evaluated explicitly as follows. We expand in Fourier modes
\begin{equation}
\xi(\theta)=\sum_{n\in\mathbb Z}\xi_{n} e^{in\theta}\, ,\quad
\Xi_{\sigma}(\rho,\theta)=\sum_{n\in\mathbb Z}\Xi_{\sigma, n}(\rho) e^{in\theta}\, .
\end{equation}
Since the fields are real, we have
\be\label{realitycond} (\xi_{n})^{*} = \xi_{-n}\, ,\quad (\Xi_{\sigma, n})^{*} = \Xi_{\sigma, -n}\, .\ee
The conditions \eqref{Ssexpcond} yields
\begin{equation}\label{equation for Fourier modes: Liouville}
\biggl[-\frac{(1-r_0^2 \rho^2)^2}{4 r_0^2}\Bigl(\frac{\d^2}{\d\rho^2}+\frac{1}{\rho}\frac{\d}{\d\rho}-\frac{n^2}{\rho^2}\Bigr)+2\biggr]\Xi_{\sigma, n}(\rho) =0\, ,\quad \Xi_{\sigma, n}(1) =\xi_{n}\, .
\end{equation}
This can be solved straightforwardly, imposing in particular the regularity of the solution at the origin $\rho=0$,
\begin{equation}\label{XisignLsol}
\Xi_{\sigma, n}(\rho)= \frac{1-r_0^2}{1+r_0^2+|n|(1- r_0^2)}\frac{1+r_0^2 \rho^2+|n|(1- r_0^2 \rho^2)}{1-r_0^2 \rho^2}\rho^{|n|}\xi_{n}.
\end{equation}
Substituting in Eq.\ \eqref{SeffLb}, we get
\begin{equation}\label{tSeffLb}
{\tilde S}_{\text{eff,b}}^{(2)}[\xi] =-2\pi\frac{1-r_{0}^{2}}{1+r_{0}^{2}}|\xi_{0}|^{2}+
4\pi \sum_{n\geq 2} \frac{(n^2-1)(1-r_0^2 )}{1+r_0^2+n(1-r_0^2)}|\xi_{n}|^{2}\, .\end{equation}
\subsection{\label{quadratic action:JT}The quadratic action $\widetilde{S}^{(2)}_{\text{eff}}$ for JT gravity}

As explained in Sec.\ \ref{JTpathsubSec}, in the case of JT gravity, the bulk field $\Sigma_{\text B}$ is set to zero and the Liouville field is entirely determined in terms of its boundary value, $\Sigma = \Sigma_{\sigma}$, by the condition \eqref{LioueqSig}. To compute the quadratic action, we thus proceed in two steps. We start by expanding as usual $\Sigma = \Sigma_{*}+|c|^{-1/2}\Xi$ and $\sigma = \sigma_{*}+|c|^{-1/2}\xi$, see Eq.\ \eqref{Sigsadexp} and \eqref{Sigsadexp2}, we express $\Xi$ in terms of $\xi$ up to terms of order $|c|^{-1/2}$, which will be needed,
\be\label{XiJTxiexp} \Xi = \Xi[\xi] = \Xi^{(0)}[\xi] + \frac{\Xi^{(1)}[\xi]}{\sqrt{|c|}} + O\bigl(|c|^{-1}\bigr)\, .\ee
The expansion \eqref{Seffquad} of the effective action \eqref{Seffdef} then reads
\begin{align}\label{Seff1JT}
\Seff^{(1)}
  & =2 \int_{\disk}\bigl(\partial_{a}\Sigma_{*}\partial_{a}\Xi^{(0)} + \mu e^{2\Sigma_{*}}\Xi^{(0)}\bigr)\, \d^{2} x + 2\int_{\partial\disk} \xi\, \d\theta\, 
\\\label{Seff2JT} \Seff^{(2)} & =2 \int_{\disk}\bigl(\partial_{a}\Sigma_{*}\partial_{a}\Xi^{(1)} + \mu e^{2\Sigma_{*}}\Xi^{(1)}\bigr)\, \d^{2} x + \int_{\disk}\bigl(\partial_{a}\Xi^{(0)}\partial_{a}\Xi^{(0)} + 2\mu e^{2\Sigma_{*}} (\Xi^{(0)})^{2}\bigr)\, \d^{2} x
\end{align}
and can be explicitly evaluated.

\subsubsection{Computation of $\Xi[\xi]$}

Expanding Eq.\ \eqref{LioueqSig} to the desired order in the $1/|c|^{1/2}$ expansion yields the following differential equations for $\Xi^{(0)}$ and $\Xi^{(1)}$,
\begin{align}\label{JTXi0cond} &\Bigl(\Delta_{*} -\frac{2\eta}{L^{2}}\Bigr)\Xi^{(0)} = 0\, ,
\\ \label{JTXi1cond} & \Bigl(\Delta_{*} -\frac{2\eta}{L^{2}}\Bigr)\Xi^{(1)} = \frac{2\eta}{L^{2}}\bigl(\Xi^{(0)}\bigr)^{2}\, .
\end{align}
As usual, $\smash{\Delta_{*} = e^{-2\Sigma_{*}}\partial_{a}\partial_{a}}$ refers to the positive Laplacian on the classical background geometries around which we expand. The relevant classical solutions for $\Sigma_{*}$ were given in Sec.\ \ref{clsolJTSec}. The equations \eqref{JTXi0cond} and \eqref{JTXi1cond} must be solved with the boundary conditions
\be\label{bcXi01} \Xi^{(0)}_{|\partial\disk} = \xi\, ,\quad \Xi^{(1)}_{|\partial\disk} = 0\ee
and imposing regularity at $\rho=0$.

The first equation \eqref{JTXi0cond} is similar to \eqref{Ssexpcond}, which was solved in \eqref{XisignLsol}. Expanding in Fourier modes
\be\label{FourXiJT012} \xi(\theta) = \sum_{n\in\mathbb Z}\xi_{n}e^{in\theta}\, ,\quad
\Xi^{(0)}(\rho,\theta) = \sum_{n\in\mathbb Z}\Xi^{(0)}_{n}(\rho)e^{in\theta}\, ,\ee
we find
\be\label{relationBBields}
\Xi^{(0)}_{n}(\rho)= \frac{1+\eta r_{0}^{2}}{1-\eta r_{0}^{2} + |n|(1+\eta r_{0}^{2})}\frac{1-\eta r_{0}^{2}\rho^{2} + |n|(1+\eta r_{0}^{2}\rho^{2})}{1+\eta r_{0}^{2}\rho^{2}}\rho^{|n|}\xi_{n}\, .\ee
Similarly, expanding
\be\label{FourerXioneJT} \Xi^{(1)}(\rho,\theta) = \sum_{n\in\mathbb Z}\Xi^{(1)}_{n}(\rho)e^{in\theta}\, ,\ee
one finds from \eqref{JTXi1cond} that the Fourier modes $\Xi^{(1)}_{n}$ satisfy an ordinary linear second order differential equation that may be solved by elementary methods. In the case of the zero mode $\Xi^{(1)}_{0}$, which is, as we shall see below, the only function we need to evaluate explicitly, the differential equation reads
\be\label{diffeqFourXione}
\biggl[-\frac{(1+\eta r_0^2 \rho^2)^2}{4 r_0^2}\Bigl(\frac{\d^2}{\d\rho^2}+\frac{1}{\rho}\frac{\d}{\d\rho}\Bigr)-2\eta\biggr]\Xi^{(1)}_{0}(\rho) = 2\eta\sum_{n\in\mathbb Z}\bigl|\Xi^{(0)}_{n}(\rho)\bigr|^{2}\, .\ee
Its solution, smooth at the origin and satisfying the boundary condition \eqref{bcXi01}, $\Xi^{(1)}_{0}(1)=0$, is given by
\begin{multline}\label{Xionezerosol} \Xi^{(1)}_{0}(\rho) = -\frac{2\eta r_{0}^{2}(1+\eta r_{0}^{2})}{1+\eta r_{0}^{2}\rho^{2}}\\\sum_{n\in\mathbb Z}\biggl[\frac{1+\eta r_{0}^{2}}{1+\eta r_{0}^{2}\rho^{2}}\,\rho^{2|n|+2}- \frac{1-\eta r_{0}^{2}\rho^{2}}{1-\eta r_{0}^{2}}\biggr]
\frac{|\xi_{n}|^{2}}{\bigl[1-\eta r_{0}^{2}+|n|(1+\eta r_{0}^{2})\bigr]^{2}}\,\cdotp
\end{multline}
For future reference, let us note that this yields
\be\label{derXione} {\frac{\d\Xi^{(1)}_{0}}{\d\rho}} \bigl(\rho=1\bigr) = -\frac{4\eta r_{0}^{2}}{1-\eta^2 r_{0}^{4}}\sum_{n\in\mathbb Z}\frac{1+\eta^2 r_{0}^{4}+|n|(1- \eta^2 r_{0}^{4})}{\bigl[1-\eta r_{0}^{2}+|n|(1+\eta r_{0}^{2})\bigr]^{2}}\, |\xi_{n}|^{2}\, .\ee
For $\eta=0$, the right-hand side of the above equation vanishes and for $\eta=\pm 1$ one has $\eta^{2}=1$.

\subsubsection{Quadratic action in the flat case}

When $\eta=0$, $\Sigma_{*} = \ln\frac{\ell}{2\pi}=\text{constant}$, $\Xi^{(0)}_{n} = \rho^{|n|}\xi_{n}$ and $\Xi^{(1)}=0$. The derivation of the quadratic action is then particularly simple. From \eqref{Seff1JT} and \eqref{Seff2JT} we immediately get
\begin{align}\label{Seff1flat} \Seff^{(1)} & = 2\Bigl(1+\frac{\mu\ell^{2}}{8\pi^{2}}\Bigr)\int_{0}^{2\pi}\xi\,\d\theta\, ,\\
\label{Seff2flat} \Seff^{(2)} & =2\pi\sum_{n\in\mathbb Z}\biggl(|n| + \frac{\mu\ell^{2}}{4\pi^{2}}\frac{1}{|n|+1}\biggr)|\xi_{n}|^{2}\, .
\end{align}
Using the rule \eqref{substidelta} then yields
\be\label{tSeffflat} {\tilde S}_{\text{eff}}^{(2)} = -2\pi\Bigl(1-\frac{\mu\ell^{2}}{8\pi^{2}}\Bigr)|\xi_{0}|^{2} + 4\pi\sum_{n\geq 2}\frac{n-1}{n+1}\Bigl(n+1-\frac{\mu\ell^{2}}{8\pi^{2}}\Bigr)|\xi_{n}|^{2}\, .\ee

\subsubsection{Quadratic action in the curved cases}

Integrating by part and using $\partial_{a}\partial_{a}\Sigma_{*} = -\eta e^{2\Sigma_{*}}/L^{2}$, we first rewrite \eqref{Seff1JT} as
\be\label{Seff1c1} \Seff^{(1)} = 2\Bigl(\mu + \frac{\eta}{L^{2}}\Bigr)\int_{\disk}e^{2\Sigma_{*}}\Xi^{(0)}\d^{2}x + 2\int_{0}^{2\pi}\Bigl(1+\frac{\d\Sigma_{*}}{\d\rho}(\rho=1)\Bigr)\xi\,\d\theta\, .\ee
Using the classical solutions \eqref{JTminsaddlesol1}, \eqref{JTplussaddlesol1} and the solution \eqref{relationBBields} for $\Xi^{(0)}$ then yields
\be\label{Seff1c2} \Seff^{(1)} =\frac{2(r_{0}^{4}+2\mu L^{2}r_{0}^{2}+1)}{1-r_{0}^{4}}\int_{0}^{2\pi}\xi\,\d\theta\, .\ee
Similarly, integrating by part, using in particular $\Xi^{(1)}_{|\partial\disk}=0$ and the equations of motion \eqref{JTXi0cond}, \eqref{JTXi1cond}, allows to rewrite \eqref{Seff2JT} as a purely boundary term,
\be\label{Seff2c1} \Seff^{(2)} = \int_{0}^{2\pi}\biggl[\xi\frac{\partial\Xi^{(0)}}{\partial\rho} - \bigl(1+\eta\mu L^{2}\bigr)\frac{\partial\Xi^{(1)}}{\partial\rho}\biggr]\,\d\theta\, .\ee
By using \eqref{relationBBields} and \eqref{derXione}, this is straightforwardly evaluated. Using \eqref{Seff1c2} and the substitution \eqref{substidelta} then yields
\begin{multline}\label{tSeffcurved} {\tilde S}_{\text{eff}}^{(2)} =-2\pi \frac{(1+\eta r_{0}^{2})(1-2(\mu L^{2}+2\eta)r_{0}^{2}+r_{0}^{4})}{(1-\eta r_{0}^{2})^{3}}\,|\xi_{0}|^{2}\\
+4\pi \frac{1+\eta r_{0}^{2}}{1-\eta r_{0}^{2}}\sum_{n\geq 2}\frac{(1-r_{0}^{4})n-2(\mu L^{2}+2\eta)r_{0}^{2}+1+r_{0}^{4}}{\bigl[1-\eta r_{0}^{2} + n (1+\eta r_{0}^{2})\bigr]^{2}}\, (n^{2}-1) |\xi_{n}|^{2}\, .
\end{multline}

Note that one gets, from either \eqref{ellminusrel} or \eqref{ellplusrel} (for $r_{0}=r_{0}^{<}$), that $r_{0}\sim\frac{\ell}{4\pi L}$ in the flat space $L\rightarrow\infty$, fixed $\ell$, limit. Then, as expected, we find the flat space result \eqref{tSeffflat} as the limit of \eqref{tSeffcurved}, starting either from $\eta=-1$ or from $\eta=+1$ for the small disk solution.

\subsection{\label{StabSec}Stability analysis}

Let us analyse the local stability of the classical geometries around which we expand. Clearly, local stability is a necessary condition for the validity of our analysis. It is equivalent to the convergence of the Gaussian integrals over the quadratic fluctuations that we perform in the next Section to derive the one-loop partition functions.

A first remark is in order. By inspecting the form of the quadratic fluctuations, Eq.\ \eqref{tSeffLb}, \eqref{tSeffflat} or \eqref{tSeffcurved}, one notices that the boundary modes $\xi_{\pm 1}$ are absent and that the sign in front of the zero mode contribution $|\xi_{0}|^{2}$ may be negative. These flat directions and/or zero mode instabilities are unphysical and harmless, because they are due to the $\PslR$ gauge redundancy of conformal gauge on the disk. It is indeed straightforward to check (see the next Section) that the $\delta$-functions in the path integrals \eqref{Zloop1} or \eqref{ZJTloop1} set $\xi_{0}=0$ and determine $\xi_{\pm 1}$ in terms of the other modes.

\subsubsection{Liouville theory}

In the Liouville theory, we have limited our analysis above to the case of a non-negative cosmological constant, because the theory is classically unstable for $\mu<0$ and the known solution, Eq.\ \eqref{exactWLioudiskc}, makes sense only for $\mu\geq 0$, suggesting that the full quantum theory exists only for $\mu\geq 0$.\footnote{The fact that the theory is ill-defined for $\mu<0$ can also be understood by looking at the moments $\langle A^{k}\rangle$ of the area in the $\mu=0$ theory. An obvious necessary condition for the existence of the theory at $\mu<0$ is that these moments must be finite for all $k$. However, they are predicted to diverge for sufficiently large values of $k$, both in the matrix model approach and by the exact formula \eqref{exactWLioudiskc}.}

Nevertheless, let us briefly discuss the case $\mu<0$ as well. When $\ell\leq 2\pi/\sqrt{-\mu}$, one finds two classical solutions, a small and a large disk, exactly as in the positive curvature JT theory with the identification $L=1/\sqrt{-\mu}$.  A comparison of the values of the classical action for these two solutions shows that the small disk always dominates over the large one. The computation of the quadratic fluctuations around these saddles, following exactly the same strategy as in Sec.\ \ref{quadratic action:Liouville}, then shows that the dominant saddle, i.e.\ the small disk, is locally stable.\footnote{One finds that the large disk is locally unstable, due to a negative eigenvalue of the operator $\Delta_*+2\mu$ appearing in the bulk action \eqref{SeffLB} for $\mu<0$ \cite{detpaper}.} 

We thus have an explicit example of a case for which the dominant saddle is locally stable,  allowing to define a consistent semiclassical large $|c|$ expansion, but which is globally unstable and do not correspond to a well-defined theory at the non-perturbative level. 

When $\mu > 0$, the bulk action \eqref{SeffLB} is positive-definite thanks to the Breitenlohner-Freedman bound and the boundary action \eqref{tSeffLb} is also manifestly positive-definite (after setting $\xi_{0}=0$ as explained above). These ensure that the saddle point in this case is locally stable. As we have already shown in section \ref{clSec}, the computation of the classical action at the saddle point correctly reproduces the leading term in the expansion of the FZZT result at the $c\rightarrow-\infty$ limit. This means that the saddle that we are considering is indeed the global minimum of the quantum effective action in the aforementioned regime.

\subsubsection{Flat space JT gravity}

For flat-space JT gravity, Eq.\ \eqref{tSeffflat} shows that our classical solution is stable if and only if
\be\label{flatstability} \mu\ell^{2} < 24\pi^{2}\, .\ee
This instablility is easily interpreted: when $\mu$ is too large and positive, the cosmological constant term in the action favours geometries with a branched-polymer structure, that have very small area, over the smooth inflated classical solution. Note that the full quantum theory is expected to be perfectly well-defined for arbitrary positive $\mu$. It is just no longer dominated by a smooth classical geometry. When $\mu<0$, the saddle is always stable. This is not surprising, since we are in a regime where the geometry is semi-classical, and the isoperimetric inequality ensures the classical stability of the model. Note that the existence of the model for any negative $\mu$ in the full quantum theory is a much more subtle issue, which remains to be settled but which can be tackled, in principle, with the methods presented in \cite{Fer1}; see also the discussion in \cite{Fer2}.

\subsubsection{Negative curvature JT}

The situation for negative curvature JT is similar to the flat case. The action \eqref{tSeffcurved} is positive-definite when $\xi_{0}=0$ if and only if $3-r_{0}^{4}-2(\mu L^{2}-2)r_{0}^{2}>0$ which, using the parameters $x$ and $y$ defined in \eqref{xydefJTmin}, is equivalent to
\begin{equation}\label{stabnegcurv}
 y <3+\frac{2}{x^{2}}\bigl(1+2\sqrt{1+x^2}\bigr)\, .
\end{equation}
This inequality gives the domain of stability of the saddle point, see Fig.\ \ref{neg_curv_stability}. For any given $\ell$, there is a strictly positive upper bound for the cosmological constant. The existence of this upper bound, as well as the stability for all negative $\mu$, can be interpreted exactly as in the flat case. The existence of the finite $c$ quantum theory for any negative $\mu$ is strongly believed to be valid but, strictly speaking, and as in the flat case, this remains to be proven \cite{Fer2}.

\begin{figure}
\centerline{\includegraphics[width=5in]{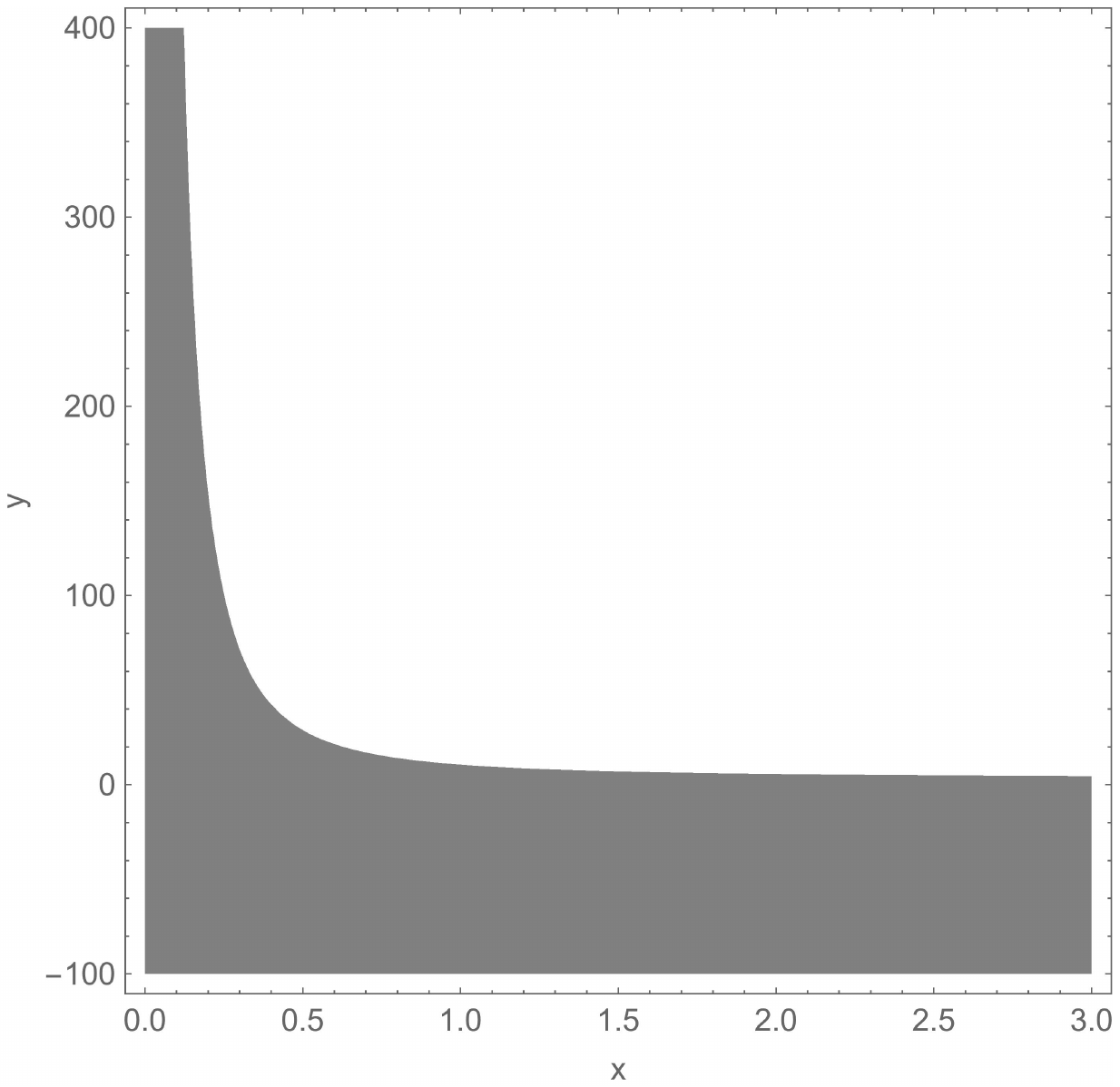}}
\caption{Domain of stability of the saddle point in negative curvature.}
\label{neg_curv_stability}
\end{figure}

\subsubsection{Positive curvature JT}

The case of positive curvature JT is the most intricate. The stability conditions, derived from \eqref{tSeffcurved}, are $3-r_{0}^{4}-2(\mu L^{2}+2)r_{0}^{2}>0$ for the small disk ($r_{0}<1$) and $3-r_{0}^{4}-2(\mu L^{2}+2)r_{0}^{2}<0$ for the large disk ($r_{0}>1$). In terms of the variables $x$ and $y$, this is equivalent to 
\begin{equation}
\label{domain of stability: positive curvature}
\begin{split}
& y<y_{\text s}(x) = -3+\frac{2}{x^{2}}\bigl(1+2\sqrt{1-x^2}\bigr)\quad \text{(small disk stability),}\\
& y>y_{\text l}(x) = -3+\frac{2}{x^{2}}\bigl(1-2\sqrt{1-x^2}\bigr)\quad \text{(large disk stability).}
\end{split}
\end{equation}
The domains of local stability are depicted in Fig.\ \ref{pos_curv_stability}, together with the critical line $y_{\text c}(x)$ separating the regions where the small disk dominates (above $y_{\text c}$) and the large disk dominates (below $y_{\text c}$). For any fixed allowed value of $x$, we find different regimes:

i) If $y>y_{\text s}(x)$, the large disk is the only locally stable solution, but we know that it is actually metastable, since the small disk has a lower action. The small disk itself is locally unstable. In this region of large positive cosmological constant, we expect a branched polymer phase and there is no dominating smooth classical geometry.

ii) If $y_{\text l}(x)<y<y_{\text s}(s)$, both the small and the large disks are locally stable, but the large disk is metastable for $y_{\text c}(x)<y<y_{\text s}(x)$ and the small disk is metastable for $y_{\text l}(x)<y<y_{\text c}(x)$. 

iii) If $y<y_{\text l}(x)$, the small disk is the only locally stable solution, but we know that it is metastable, since the large disk has a lower action. The large disk itself is locally unstable. The instability for $y<y_{\text l}(x)$ is specific to the positive curvature theory. This is consistent with the results in \cite{Fer2}, where it is explained that positive curvature JT gravity is actually unstable for any negative value of the cosmological constant.

So we see that our semiclassical, $c\rightarrow -\infty$ limit, analysis in positive curvature makes sense only in an interval $y\in ]y_{\text l}(x),y_{\text s}(x)[$ of values for the cosmological constant. It is tempting to conjecture that the full, finite $c$, quantum theory will exist only for sufficiently large cosmological constant. Below a certain critical value $\mu_{\text c}\leq 0$, the classical instability of the theory probably gives rise to the inconsistency of the full quantum theory \cite{Fer2}.

\begin{figure}[h!]
\centerline{{\includegraphics[width=5in]{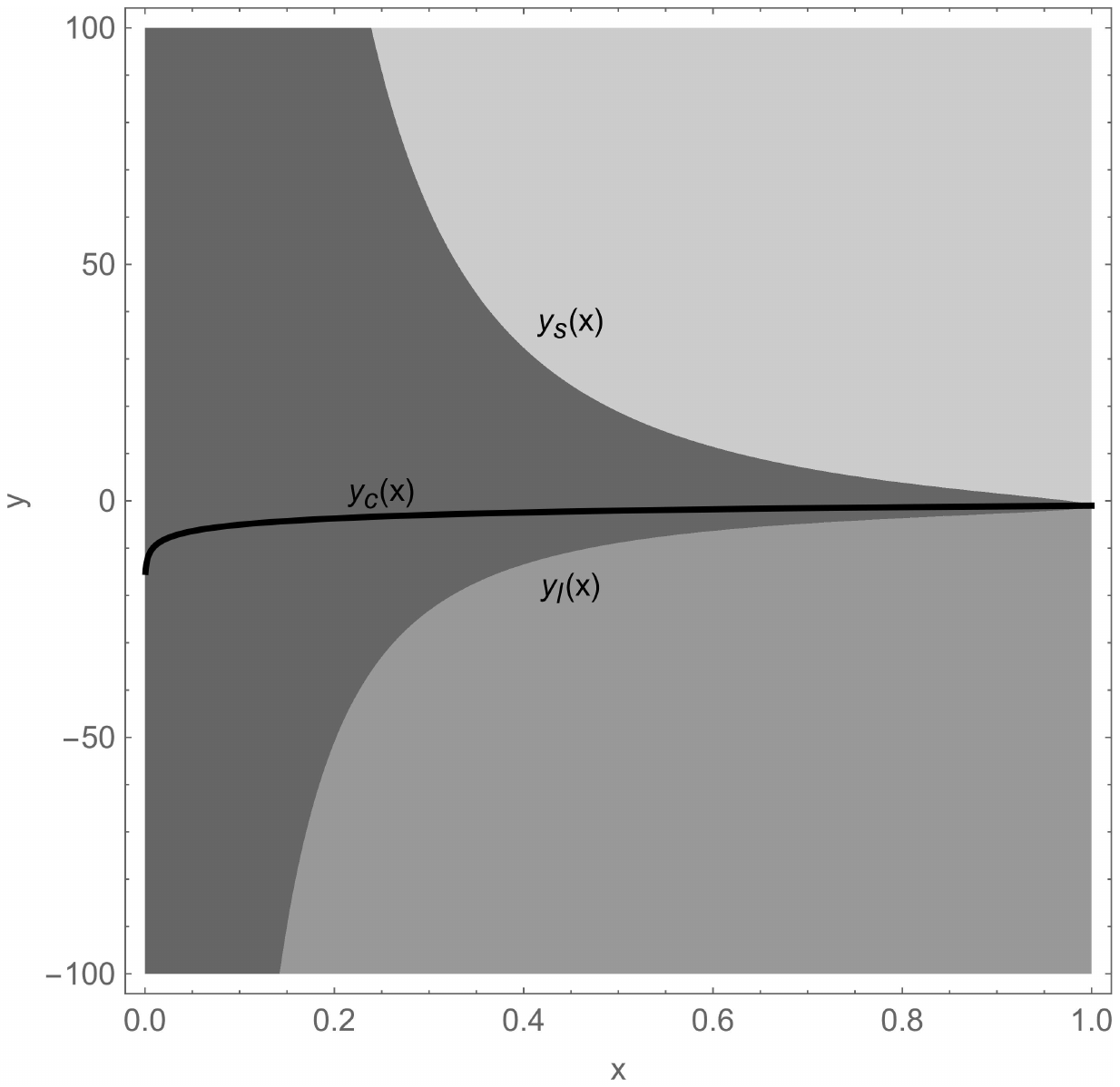}}}
\caption{Domains of stability of the saddle points in positive curvature. The critical line, in black, separates the regions where the small disk dominates (above) and the large disk dominates (below). The lightest shade of gray indicates the domain (upper-right region) in which only the large disk is locally stable. The slightly darker shade of gray indicates the domain (lower-right region) in which only the small disk is locally stable. The darkest shade of gray indicates  the domain in which both disks are locally stable. See the main text for a detailed discussion}
\label{pos_curv_stability}
\end{figure}
\section{\label{oneloopSec}One-loop order: path integrals}
\subsection{One-loop counterterms}

The general form of the counterterms were discussed in Sec.\ \ref{ctSec}. At one-loop, their contribution to the partition function is straightforwardly obtained by evaluating each term on the relevant classical backgrounds. The finite contributions that come from counterterms are fundamental ambiguities in the logarithm of the partition functions, that can be eliminated by redefining the parameters. It is thus useful to write down their explicit form. In the results presented in the next subsections, we shall often simply write ``$+\ \text{counterterms}$'' to indicate that these terms might be added, with undetermined coefficients that can depend on the local couplings but not on $\ell$ and thus not on $r_{0}$ or equivalently on $x$.

In the Liouville theory, the counterterm action \eqref{SLct}, written in terms of the parameter $r_{0}$ or equivalently in terms of the parameter $x$, yields
\be\label{SLctoneloop} S^{\text L}_{\text{c.t., 1-loop}} =  c_{1} +  c_{2}\frac{r_{0}}{1-r_{0}^{2}} + \frac{c_{3}}{1-r_{0}^{2}}= \tilde c_{1} +\tilde c_{2} x + \tilde c_{3}\sqrt{1+x^{2}}\ee
for arbitrary numerical constants $c_{i}$s and $\tilde c_{i}$s.

In flat JT gravity, the counterterms are of the same kind, and evaluating on the classical background yields
\be\label{SJTzeroctoneloop} S^{\text{JT}, 0}_{\text{c.t., 1-loop}} = c_{1} + c_{2} \ell + c_{3}\ell^{2}\ee
for numerical constants $c_{i}$s. 

In curved JT gravity, there are similar counterterms, but also a new contribution that comes from the last term in Eq.\ \eqref{SJTct}. This term can be evaluated by using \eqref{varphiPhirel} and the classical solutions \eqref{JTminsaddlesol2} and \eqref{JTplussaddlesol2},
\be\label{ctJTspecial} \int_{\disk}\d^{2}x\sqrt{g_{*}}\,\varphi_{*}^{\eta} = 4\pi(\mu L^2+\eta) L^{2}\frac{r_{0}^{4}}{1-r_{0}^{4}} \,\cdotp\ee
Overall, we get
\be\label{SJTpmctoneloop} S^{\text{JT}, \eta}_{\text{c.t., 1-loop}} = c_{1} + c_{2}\frac{r_{0}}{1+\eta r_{0}^{2}} + \frac{c_{3}}{1-r_{0}^{2}} + \frac{c_{4}}{1+r_{0}^{2}} =\tilde c_{1} + \tilde c_{2} x + \tilde c_{3}\sqrt{1-\eta x^{2}}+\frac{\tilde c_{4}}{\sqrt{1-\eta x^{2}}} \ee
for dimensionless constants $c_{i}$ or $\tilde c_{i}$ that may depend on $y$.

\subsection{\label{detSec}Bulk determinants}

In Liouville gravity, integrating out the bulk fluctuations $\Xi_{\text B}$ in the one-loop path integral \eqref{Zloop1}, by using the form \eqref{SeffLB} of the action for the quadratic fluctuations, produces a bulk contribution to the partition function of the form
\be\label{BulkLpathint} Z_{\text{L, 1-loop, Bulk}} = \frac{1}{\sqrt{\det_{\text D}(\Delta_{*}+2\mu)}}\, \cvp\ee
where $\Delta_{*}$ is the positive Laplacian on the negative curvature disk $R=-2\mu$ of fixed boundary length $\ell$ and the determinant is computed with Dirichlet boundary conditions.\footnote{This is because the bulk field $\Xi_{\text B}$ satisfies Dirichlet boundary conditions, see Eq.\ \eqref{XiBDirichlet}.} Similarly, the JT gravity one-loop path integral \eqref{ZJTloop1} has a factor
\be\label{BulkJTpathint} Z_{\text{1-loop, Bulk}}^{(\eta)} = \frac{1}{|\det_{\text D}(\Delta_{*}-\frac{2\eta}{L^{2}})|}\, \cvp\ee
where $\Delta_{*}$ is the positive Laplacian on the constant curvature disk $R=2\eta/L^{2}$ of fixed boundary length $\ell$ and the determinant is also computed with Dirichlet boundary conditions. Note that the determinant needed for the Liouville theory is the same as the one needed for negative curvature JT, with the identification $\mu = 1/L^{2}$.

We are thus facing a non-trivial problem: evaluating the determinants $\det_{\text D}(\Delta_{*}-\frac{2\eta}{L^{2}})$ on \emph{finite size} constant curvature disks. As far as we know, this calculation has never been done before. Since the result is in itself new and interesting (in non-zero curvature), and the developments required to derive the result are rather lengthy and of a different spirit from the ideas of the present paper, we have decided to present the derivation in a separate paper \cite{detpaper}. In \cite{detpaper}, we actually discuss more general determinants, of the form $\det_{\text D}(\Delta_{*}+M^{2})$, for arbitrary mass parameter $M$. The result of the calculations, made within the $\zeta$-function regularization scheme, are as follows.

\subsubsection{Flat bulk determinants}

For
\be\label{deltastarflat} \Delta_{*} = -\Bigl(\frac{2\pi}{\ell}\Bigr)^{2}\partial_{a}\partial_{a}\ee
the positive Laplacian on a flat disk of circumference $\ell$, we have
\begin{multline}\label{Bdetflat} \ln{\det}_{\text D}\biggl[\Bigl(\frac{\ell_{\text{IR}}}{2\pi}\Bigr)^{2}\bigl(\Delta_{*}+M^{2}\bigr)\biggr] =
\frac{1}{3}\ln 2 - \frac{1}{2}\ln(2\pi)-\frac{5}{12}-2\zeta_{\text R}'(-1)\\ +\frac{1}{2}\bigl(\gamma -1-\ln 2 \bigr)\Bigl(\frac{M\ell}{2\pi}\Bigr)^{2} + \biggl[-\frac{1}{3} + \frac{(M\ell)^{2}}{8\pi^{2}}\biggr]\ln\frac{\ell}{\ell_{\text{IR}}}\\
+\sum_{n\in\mathbb Z}\ln\biggl[2^{|n|}|n|! \Bigl(\frac{2\pi}{M\ell}\Bigr)^{|n|} I_{|n|}\Bigl(\frac{M\ell}{2\pi}\Bigr) \exp\Bigl(-\frac{M^{2}\ell^{2}}{16\pi^{2}(|n|+1)}\Bigr)\biggr]
\end{multline}
where $\zeta_{\text R}$ is the Riemann $\zeta$-function, $\gamma$ is Euler's constant and $I_{n}$ is the Bessel function of the second kind. We have introduced an arbitrary infrared renormalization scale $\ell_{\text{IR}}$ to make the determinants dimensionless.

For our purposes, we need the case $M=0$, which yields
\be\label{Bdetflat2} \ln{\det}_{\text D}\biggl[\Bigl(\frac{\ell_{\text{IR}}}{2\pi}\Bigr)^{2}\Delta_{*}\biggr] = - \frac{1}{3}\ln\ell + \text{counterterms}.\ee
Of course, this very simple result at $M=0$ can be derived straightforwardly from the conformal anomaly, since $\det_{\text D}\Delta_{*}$ is the partition function of a $c=-2$ conformal field theory. We have included the full formula \eqref{Bdetflat} for completeness but also to compare with the negative and positive curvature cases below, for which the conformal anomaly does not provide the result we need.

\subsubsection{Curved bulk determinants}

For 
\be\label{deltastarneg} \Delta_{*}^{\eta}= -\frac{(1+\eta r_{0}^{2}\rho^{2})^{2}}{4L^{2}r_{0}^{2} }\partial_{a}\partial_{a}\ee
the positive Laplacian on a constant curvature disk, $R=2\eta/L^{2}$, of circumference $\ell$ (see Eq.\ \eqref{ellminusrel} and Eq.\ \eqref{ellplusrel} for the relation between the parameters $r_{0}$ and $\ell$), we have
\begin{multline}\label{Bdetposneg}
\ln{\det}_{\text D}\biggl[\Bigl(\frac{\ell_{\text{IR}}}{2\pi}\Bigr)^{2}\bigl(\Delta_{*}^{\eta}+M^{2}\bigr)\biggr] = - \frac{1}{2}\ln(2\pi)-\frac{5}{12}-2\zeta_{\text R}'(-1)\\ 
+ \frac{\eta A^{\eta}}{3\pi L^{2}}
 +\frac{\gamma -1}{2\pi} A^{\eta}M^{2} +
 \biggl[-\frac{1}{3} + \frac{A^{\eta}M^{2}}{2\pi}\biggr]\ln\frac{2\pi Lr_{0}}{\ell_{\text{IR}}}\\
+\sum_{n\in\mathbb Z}\ln\Biggl[
F\Bigl(\frac{1}{2}+\frac{1}{2}\sqrt{1-4\eta (ML)^{2}},\frac{1}{2}-\frac{1}{2}\sqrt{1-4\eta (ML)^{2}},|n|+1,\frac{\eta A^{\eta}}{4\pi L^{2}}\Bigr)\\
\exp\biggl(-\frac{A^{\eta}M^{2}}{4\pi(|n|+1)}F\Bigl(1,1,|n|+2,\frac{\eta A^{\eta}}{4\pi L^{2}}\Bigr)\biggr)\Biggr]\, ,
\end{multline}
where $F$ is the usual hypergeometric function and $A^{\eta}$ is the area, given in terms of $r_{0}$ in Eq.\ \eqref{Akmin} and \eqref{Akplus}. Note that it is straightforward to check that the flat space limit of \eqref{Bdetposneg}, $L\rightarrow\infty$ at fixed $\ell$,  reproduces the flat space result \eqref{Bdetflat}. In particular, the hypergeometric functions yield the correct Bessel functions, as can be checked easily by analysing the limit of the coefficients in the defining series representations.

For our purposes, we need the special case $M^{2} = -2\eta/L^{2}$. The hypergeometric function  then drastically simplifies to
\begin{multline}\label{hypergeomsimple} F\Bigl(\frac{1}{2}+\frac{1}{2}\sqrt{1-4\eta (ML)^{2}},\frac{1}{2}-\frac{1}{2}\sqrt{1-4\eta (ML)^{2}},|n|+1,z\Bigr) \\= F\bigl(2,-1,|n|+1,z\bigr) = 1-\frac{2 z}{|n|+1}
\end{multline}
and the infinite sum in \eqref{Bdetposneg} can be explicitly evaluated \cite{detpaper},
\begin{multline}\label{infiniteSsum} \sum_{n\in\mathbb Z}\ln\Biggl[\biggl(1-\frac{\eta A^{\eta}}{2\pi L^{2}}
\frac{1}{|n|+1}\biggr)\exp\biggl(\frac{\eta A^{\eta}}{2\pi L^{2}}\frac{1}{|n|+1}F\Bigl(1,1,|n|+2,\frac{\eta A^{\eta}}{4\pi L^{2}}\Bigr)\biggr)\Biggr]\\
= \frac{\gamma\eta A^{\eta}}{\pi L^{2}} + \ln\frac{1-\eta r_{0}^{2}}{1+\eta r_{0}^{2}} - 
\frac{1-\eta r_{0}^{2}}{1+\eta r_{0}^{2}}\ln \bigl(1+\eta r_{0}^{2}\bigr)^{2} - 2\ln\Gamma\Bigl(\frac{2}{1+\eta r_{0}^{2}}\Bigr)\, .
\end{multline}
In the above equation, $\Gamma$ denotes Euler's Gamma function. This yields
\begin{multline}\label{Bdetposneg2}
\ln{\det}_{\text D}\biggl[\Bigl(\frac{\ell_{\text{IR}}}{2\pi}\Bigr)^{2}\bigl(\Delta_{*}^{\eta}-2\eta/L^{2}\bigr)\biggr] =
-\Bigl(\frac{1}{3} + \frac{\eta A^{\eta}}{\pi L^{2}}\Bigr)\ln r_{0}\\
+ \ln \frac{1-\eta r_{0}^{2}}{1+\eta r_{0}^{2}} - \frac{1-\eta r_{0}^{2}}{1+\eta r_{0}^{2}}
\ln \bigl(1+\eta r_{0}^{2}\bigr)^{2} - 2\ln\Gamma\Bigl(\frac{2}{1+\eta r_{0}^{2}}\Bigr) + \text{counterterms.}
\end{multline}
For future reference, it is useful to express this result in terms of the variable $x$, defined by Eq.\ \eqref{defxpara} in Liouville or by Eq.\ \eqref{xydefJTmin} in JT. For Liouville or for negative JT gravity, we get the same expression as a function of $x$,\footnote{For future convenience, we add constants and discard terms proportional to $\sqrt{1+x^{2}}$, which are counterterms.}
\begin{multline}\label{Bdetnegx}
\ln{\det}_{\text D}\biggl[\Bigl(\frac{\ell_{\text{IR}}}{2\pi}\Bigr)^{2}\bigl(\Delta_{*}^{-}+2/L^{2}\bigr)\biggr] =\Bigl(2\sqrt{1+x^{2}}-\frac{7}{3}\Bigr)\ln \frac{\ell}{2\pi}+\frac{7}{3}\ln\Bigl[\frac{1}{2}\bigl(1+\sqrt{1+x^{2}}\bigr)\Bigr]  \\+\frac{1}{2}\ln (1+x^{2})-2\ln\Gamma \bigl(1+\sqrt{1+x^{2}}\bigr) + \text{counterterms.}
\end{multline}
For positive curvature JT, with $\varepsilon=+1$ for the large disk and $\varepsilon=-1$ for the small disk,
\begin{multline}\label{Bdetposx}
\ln\biggl|{\det}_{\text D}\biggl[\Bigl(\frac{\ell_{\text{IR}}}{2\pi}\Bigr)^{2}\bigl(\Delta_{*}^{+}-2/L^{2}\bigr)\biggr]\biggr| =-\Bigl(\frac{7}{3}+2\varepsilon\sqrt{1-x^{2}}\Bigr)\ln\frac{\ell}{2\pi} +\frac{7}{3}\ln\Bigl[\frac{1}{2}\bigl(1-\varepsilon\sqrt{1-x^{2}}\bigr)\Bigr]  \\+\frac{1}{2}\ln (1-x^{2})-2\ln\Gamma \bigl(1-\varepsilon\sqrt{1-x^{2}}\bigr) + \text{counterterms.}
\end{multline}
Note that we have been careful in this case to indicate that we consider the absolute value of the determinant. Indeed, for the large disk solution, $r_{0}>1$, the operator $\Delta_{*}^{+}-2/L^{2}$ has one negative eigenvalue \cite{detpaper}, which is reflected in the negative sign of the argument of the logarithm of the first term on the second line of Eq.\ \eqref{Bdetposneg2}.

\subsection{Gauge fixing}

Let us now treat the gauge-fixing $\delta$-functions that appear in the one-loop path integrals for both Liouville, Eq.\ \eqref{Zloop1}, and JT, Eq.\ \eqref{ZJTloop1}. The analysis is simplified a little by rewriting the product of $\delta$-functions as
\be\label{delgfprod}
\delta\Bigl(\int_{0}^{\frac{\ell}{3}}\xi\, \d s_{*}\Bigr)\delta\Bigl(\int_{\frac{\ell}{3}}^{\frac{2\ell}{3}}\xi\, \d s_{*}\Bigr)\delta\Bigl(\int_{\frac{2\ell}{3}}^{\ell}\xi\,\d s_{*}\Bigr) = \delta\Bigl(\int_{0}^{\ell}\xi\, \d s_{*}\Bigr)\delta\Bigl(\int_{0}^{\frac{\ell}{3}}\xi\, \d s_{*}\Bigr)\delta\Bigl(\int_{\frac{\ell}{3}}^{\frac{2\ell}{3}}\xi\,\d s_{*}\Bigr)\, .\ee
Using the usual mode expansion $\xi=\sum_{n\in\mathbb Z}\xi_{n}e^{in\theta}$, the first $\delta$-function yields
\be\label{gfdelta1} \delta\Bigl(\int_{0}^{\ell}\xi\, \d s_{*}\Bigr) = \delta\bigl(2\pi e^{\sigma_{*}}\xi_{0}\bigr) = \frac{e^{-\sigma_{*}}}{2\pi}\delta (\xi_{0})\ee
and thus sets $\xi_{0}=0$ while producing a factor $e^{-\sigma_{*}}$ (we discard as usual the overall numerical constant, which is a counterterm). The other two $\delta$-functions produce in the same way two factors of $e^{-\sigma_{*}}$ and allow to express $\xi_{\pm 1}$ in terms of the $\xi_{n}$ for $|n|\geq 2$. But the relevant quadratic actions governing the fluctuations around the saddle, given in Eq.\ \eqref{tSeffLb} for Liouville and in Eq.\ \eqref{tSeffflat} and \eqref{tSeffcurved} for JT, do not depend on $\xi_{1}$ or $\xi_{-1}$. These variables can thus be trivially integrated out.

To get all the factors correctly, we must also take into account the form of the formal path integral measure over $\xi$. This measure is derived from the metric indicated in Eq.\ \eqref{oneloopmet}. It reads, modulo global unphysical numerical factors,\footnote{The ambiguities coming from arbitrary numerical factors multiplying the metrics from which the path integral measures are derived turn out, as expected, to be counterterms. This is  discussed below.}
\be\label{pathmeaxi} D_{*}\xi = \prod_{n\in\mathbb Z}\bigl(e^{\frac{\sigma_{*}}{2}}\d\xi_{n}\bigr)\, .\ee
The integration over the modes $\xi_{0}$, $\xi_{1}$ and $\xi_{-1}$ thus produce an overall factor
\be\label{zeromodefactor} e^{-3\sigma_{*}}\times e^{\frac{3\sigma_{*}}{2}} = 
e^{-\frac{3\sigma_{*}}{2}}\, ,\ee
with three factors of $e^{-\sigma_{*}}$ coming from the $\delta$-functions and three factors of $e^{\frac{\sigma_{*}}{2}}$ coming from the measure.

\subsection{\label{LoneSec}Liouville gravity}

\subsubsection{\label{bdpathLSec}Boundary path integral}

The Gaussian integral over the boundary degrees of freedom in Liouville theory can now be performed straightforwardly. From the form of the measure \eqref{pathmeaxi} and of the boundary action \eqref{tSeffLb}, taking into account the factor \eqref{zeromodefactor}, we get the one-loop boundary contribution to the partition function,
\be\label{bdLpathint1} Z_{\text{L, 1-loop, bd.}} = e^{-\frac{3\sigma_{*}}{2}}\prod_{n\geq 2}
\biggl[e^{\sigma_{*}}\frac{1+r_{0}^{2}+n (1-r_{0}^{2})}{(1-r_{0}^{2})(n^{2}-1)}\biggr]\, .\ee
The infinite product can be regularized and computed using standard $\zeta$-function techniques. This is explained in detail in the Appendix. Using the definition \eqref{gendetbdApp} and the result \eqref{BddetApp}, we get
\be\label{bdLpathint2} Z_{\text{L, 1-loop, bd.}} = e^{-\frac{3\sigma_{*}}{2}}D^{-1}\Bigl(1,3;2+\frac{1+r_{0}^{2}}{1-r_{0}^{2}};e^{-\sigma_{*}}\Bigr) = \sqrt{\frac{2}{\pi}}\,\frac{(e^{\sigma_{*}})^{-3+\frac{1+r_{0}^{2}}{1-r_{0}^{2}}}}{\Gamma\Bigl(2+\frac{1+r_{0}^{2}}{1-r_{0}^{2}}\Bigr)}\, \cdotp\ee
In terms of the variable $x=\frac{\ell\sqrt{\mu}}{2\pi}$, this reads
\be\label{bdLpathint3}\ln  Z_{\text{L, 1-loop, bd.}} = \Bigl(-3 + \sqrt{1+x^{2}}\Bigr)\ln\frac{\ell}{2\pi}-\ln\Gamma\bigl(2+\sqrt{1+x^{2}}\bigr) + \text{counterterms.}\ee

\noindent\emph{Remark:} one may wonder what happens if the bulk and boundary metrics in field space, Eq.\ \eqref{oneloopmet}, from which the formal path integral measures are defined, are multiplied by  arbitrary numerical constant $A_{\text B}$ and $A_{\text b}$. Since these numerical constants are not fixed by physics, one expects on general gounds that $A_{\text B}$ and $A_{\text b}$ can be cancelled by counterterms. For the bulk piece, the constant can clearly be absorbed in the IR cut-off $\ell_{\text{IR}}$ introduced in \eqref{Bdetposneg}. This cut-off appears in a contribution to the logarithm of the determinant which is proportional to the area and is thus a counterterm (this also explains why we have not kept it explicitly in Eq.\ \eqref{Bdetposneg2}). For the boundary piece, the constant $A_{\text b}$ can always be absorbed in the factor $z$ entering in the definition of the determinant given in the Appendix, Eq.\ \eqref{Appprop}. It thus yields a contribution proportional to $\zeta(0)$, given by Eq.\ \eqref{ZzeroApp}. In the present case, this is proportional to the area plus a constant and is thus, as required, a counterterm. Let us note as well that overall numerical factors in front of the actions, like the $\frac{1}{24\pi}$ factor sitting in front of \eqref{Zloop1} and \eqref{ZJTloop1}, can be absorbed by rescaling the integration variables. This is equivalent to multiplying the metric in field space by a constant and hence can be absorbed in  counterterms. This is the reason why we usually discard such terms.

\subsubsection{Putting everything together}

From \eqref{Zloop1}, the full one-loop partition function is the sum
\be\label{pulallL1} \ln Z_{\text{L, 1-loop}} = -\frac{13}{12\pi} S_{\text L*} + 2\sigma_{*} + \ln Z_{\text{L, 1-loop, Bulk}} + \ln Z_{\text{L, 1-loop, bd.}}\, .\ee
Using \eqref{SLsadJTmin}, the first piece yields a contribution
\be\label{piece1Lall}-\frac{13}{12\pi} S_{\text L*} + 2\sigma_{*} = -\frac{7}{3}\ln\frac{\ell}{2\pi} + \frac{13}{3}\ln \Bigl[\frac{1}{2}\bigl(1+\sqrt{1+x^{2}}\bigr)\Bigr] + \text{counterterms}.\ee
The bulk contribution is, from \eqref{BulkLpathint}, $-\frac{1}{2}$ times the logarithm of the determinant given in Eq.\ \eqref{Bdetnegx}, for $L^{2}=1/\mu$, and the boundary contribution is given in Eq.\ \eqref{bdLpathint3}. Adding up, recalling the identity $\Gamma(1+u) = u\Gamma(u)$, yields a one-loop term $\ln Z_{\text{L, 1-loop}}=f_{1}$ that precisely match the result given in Eq.\ \eqref{ZLloopcoeff1}, which was derived from the exact FZZT formula \eqref{exactWLioudiskc}.\footnote{Note that $\sqrt{1+x^{2}}-1 = \frac{\mu A^{-}}{2\pi}$ is an area counterterm and therefore the term proportional to it in Eq.\ \eqref{ZLloopcoeff1} can be discarded.} As far as we know, this provides the first non-trivial path integral test of FZZT. Maybe more importantly, it gives us confidence in the robustness of our methods, which we are now going to apply to JT gravity.

\subsection{Flat space JT gravity}

\subsubsection{Boundary path integral}

Similarly to the case of Liouville gravity, but now using the quadratic action \eqref{tSeffflat}, we get the boundary partition function for flat JT gravity,
\be\label{bdJTpathintzero}\begin{split} Z^{0}_{\text{1-loop, bd.}} & = e^{-\frac{3\sigma_{*}}{2}}\prod_{n\geq 2}\biggl[e^{\sigma_{*}}\frac{n+1}{(n-1)(n+1-\frac{\mu\ell^{2}}{8\pi^{2}})}\biggr] \\
& = e^{-\frac{3\sigma_{*}}{2}} D^{-1}\bigl(1,3-\frac{\mu\ell^{2}}{8\pi^{2}};3;e^{-\sigma_{*}}\bigr) = \frac{1}{2\sqrt{2\pi}}(e^{\sigma_{*}})^{-2+\frac{\mu\ell^{2}}{8\pi^{2}}}\, \Gamma\Bigl(3-\frac{\mu\ell^{2}}{8\pi^{2}}\Bigr)\, .\end{split}\ee

\noindent\emph{Remark:} let us check, as we have done at the end of Sec.\ \eqref{bdpathLSec} for the case of Liouville, that the arbitrary numerical constant $A_{\text B}$ and $A_{\text b}$ can be absorbed in counterterms. It works for the bulk constant, for exactly the same reason as in the case of Liouville. For the boundary constant, the $\zeta$-function \eqref{ZzeroApp} being such that $\zeta(0) = -\frac{1}{2}+\frac{\mu\ell^{2}}{8\pi^{2}} = -\frac{1}{2} + \frac{\mu A^{0}}{2\pi}$ is a counterterm, it works as well.

\subsubsection{Putting everything together}

From \eqref{ZJTloop1}, the logarithm of the full one-loop partition function is the sum of $\ln Z^{0}_{\text{1-loop, bd.}}$ and of the additional contribution
\begin{equation}
\begin{split}
 -\frac{13}{12\pi}S_{\text L*}^{0}+ 2\sigma_{*} -\ln{\det}_{D}\Delta_{*} 
&= \Bigl(-\frac{13}{3}+2+\frac{1}{3}\Bigr)\ln\frac{\ell}{2\pi} + \text{counterterms}\\
&= -2\ln\frac{\ell}{2\pi}+ \text{counterterms}
\end{split}
\label{piecezeroJT}
\end{equation}
that we have evaluated using \eqref{SLsadzero} and \eqref{Bdetflat2}. This yields
\be\label{Zzerofinal} \ln Z^{0}_{\text{1-loop}} = f_{1}^{0}= -\Bigl(4-\frac{\mu\ell^{2}}{8\pi^{2}}\Bigr)\ln\frac{\ell}{2\pi} + \ln \Gamma\Bigl(3-\frac{\mu\ell^{2}}{8\pi^{2}}\Bigr) + \text{counterterms.}\ee
Note that the result makes sense as long as $\mu\ell^{2}<24\pi^{2}$, consistently with the stability analysis in Sec.\ \ref{StabSec}, Eq.\ \eqref{flatstability}.

\subsection{Curved space JT gravity}

\subsubsection{Boundary path integral}

The quadratic action \eqref{tSeffcurved} yields
\be\label{bdJTpathintcurved}\begin{split} Z^{\eta}_{\text{1-loop, bd.}} & = e^{-\frac{3\sigma_{*}}{2}}\prod_{n\geq 2}\biggl[e^{\sigma_{*}}\frac{1-\eta r_{0}^{2}}{1+\eta r_{0}^{2}}\frac{\bigl[1-\eta r_{0}^{2} + n (1+\eta r_{0}^{2})\bigr]^{2}}{\bigl[(1-r_{0}^{4})n - 2(\mu L^{2}+2\eta)r_{0}^{2}+1+r_{0}^{4}\bigr](n^{2}-1)}\biggr] \\
& = e^{-\frac{3\sigma_{*}}{2}} D^{-1}\Bigl(\frac{3-2(\mu L^{2}+2\eta)r_{0}^{2}-r_{0}^{4}}{1-r_{0}^{4}},3,1;\frac{3+\eta r_{0}^{2}}{1+\eta r_{0}^{2}}, \frac{3+\eta r_{0}^{2}}{1+\eta r_{0}^{2}};e^{-\sigma_{*}}\Bigr)\\
& =\sqrt{\frac{2}{\pi}}\, (e^{2\sigma_{*}})^{-2+\frac{1+\mu L^{2}r_{0}^{2}}{1-r_{0}^{4}}}\,\Gamma\Bigl(\frac{3-2(\mu L^{2}+2\eta)r_{0}^{2}-r_{0}^{4}}{1-r_{0}^{4}}\Bigr)\biggl[\Gamma\Bigl(\frac{3+\eta r_{0}^{2}}{1+\eta r_{0}^{2}}\Bigr)\biggr]^{-2}\,\cdotp
\end{split}\ee
In terms of the usual variables $x=\frac{\ell}{2\pi L}$ and $y=\mu L^{2}$, this reads
\begin{multline}\label{bdJTpathint3}\ln  Z^{-}_{\text{1-loop, bd.}} = \Bigl(-3 + \sqrt{1+x^{2}} + \frac{(y-1)x^{2}}{2\sqrt{1+x^{2}}}\Bigr)\ln\frac{\ell}{2\pi}\\
+\ln\Gamma\Bigl(2+\sqrt{1+x^{2}} - \frac{(y-1)x^{2}}{2\sqrt{1+x^{2}}}\Bigr)
-2\ln\Gamma\bigl(2+\sqrt{1+x^{2}}\bigr) + \text{counterterms}
\end{multline}
in negative curvature, and
\begin{multline}\label{bdJTpathint4}\ln  Z^{+}_{\text{1-loop, bd.}} = \Bigl(-3 -\varepsilon \sqrt{1-x^{2}} -\varepsilon \frac{(y+1)x^{2}}{2\sqrt{1-x^{2}}}\Bigr)\ln\frac{\ell}{2\pi}\\
+\ln\Gamma\Bigl(2-\varepsilon\sqrt{1-x^{2}} +\varepsilon \frac{(y+1)x^{2}}{2\sqrt{1-x^{2}}}\Bigr)
-2\ln\Gamma\bigl(2-\varepsilon\sqrt{1-x^{2}}\bigr) + \text{counterterms}
\end{multline}
in positive curvature, where $\varepsilon = +1$ and $\varepsilon = -1$ correspond to the large disk and the small disk solutions respectively.

\noindent\emph{Remark:} let us check again that the numerical constants $A_{\text B}$ and $A_{\text b}$ that may multiply the flat space metrics can be absorbed in counterterms. It works for $A_{\text B}$ exactly as in the case of Liouville or flat space JT. As for $A_{\text b}$, it will work if and only if $\zeta(0)$, given by \eqref{ZzeroApp}, is a counterterm. In the present case, we find
\be\label{curvedJTzetazero} \zeta (0) = -\frac{5}{2} + \frac{2}{1-r_{0}^{4}}+ 2\mu L^{2}\frac{r_{0}^{2}}{1-r_{0}^{4}}\, \cdotp\ee
This is indeed a counterterm of the form given in Eq.\ \eqref{SJTpmctoneloop}.
 
\subsubsection{Putting everything together}

As in the flat case, using \eqref{ZJTloop1}, the logarithm of the full one-loop partition function is the sum of $\ln Z^{\eta}_{\text{1-loop, bd.}}$, of $-\frac{13}{12\pi}S_{\text L*}+ 2\sigma_{*}$, which is given by \eqref{piece1Lall} in negative curvature and by
\be\label{pieceSLJT} -\frac{13}{12\pi}S_{\text L*}^{+}+ 2\sigma_{*} = -\frac{7}{3}\ln\frac{\ell}{2\pi} +\frac{13}{3}\ln\Bigl[\frac{1}{2}\bigl(1-\varepsilon\sqrt{1-x^{2}}\bigr)\Bigr]+\text{counterterms}\ee
in positive curvature, and of $-\ln |\det_{\text D}(\Delta_{*}^{\eta}-2\eta/L^{2})|$ given by Eq.\ \eqref{Bdetnegx} and \eqref{Bdetposx}. We thus obtain the one-loop partition functions, $\ln Z_{\text{1-loop}}^{\eta} = f_{1}^{\eta}$,
\begin{align}
\label{Zminusfinal} 
\begin{split}
f_{1}^{-}  & = -\Bigl(3+\sqrt{1+ x^2}- \frac{(y-1) x^2}{2\sqrt{1+ x^2}}\Bigr)\ln\frac{\ell}{2\pi}-\frac{1}{2}\ln (1+x^2) \\& \hskip 3cm
+\ln \Gamma\Bigl(2+\sqrt{1+ x^2}- \frac{(y-1)  x^2}{2\sqrt{1+ x^2}}\Bigr) + \text{counterterms,}
\end{split}
\\ \label{Zplusfinal}
\begin{split}
f_{1}^{+}  & = -\Bigl(3- \varepsilon\sqrt{1- x^2}+ \varepsilon\frac{(y+1)  x^2}{2\sqrt{1-x^2}}\Bigr)\ln\frac{\ell}{2\pi}-\frac{1}{2}\ln (1-x^2) \\ & \hskip 3cm
 +\ln \Gamma\Bigl(2- \varepsilon\sqrt{1- x^2}+\varepsilon\frac{(y+1)  x^2}{2\sqrt{1-x^2}}\Bigr)+ \text{counterterms.}
\end{split}
\end{align}
As usual, $\varepsilon = +1$ for the large disk and $\varepsilon = -1$ for the small disk.

\section{\label{finalSec} Discussion}

\subsection{\label{flatdisSec}Flat space limit}

A simple consistency check amounts to taking the flat space limit 
\be\label{flatlimit} L\rightarrow\infty\, ,\quad \text{$\ell$ and $\mu$ fixed}\ee
of our JT gravity results in negative curvature or in positive curvature when the small disk dominates. In both cases, the tree level contributions, Eqs.\ \eqref{fzeromin} and \eqref{fzeroplus}, and the one-loop contributions, Eqs.\ \eqref{Zminusfinal} and \eqref{Zplusfinal}, have the correct limits yielding the flat space results \eqref{fzerozero} and \eqref{Zzerofinal}.

\subsection{\label{critdisSec}Critical exponents, improvement and all-loop predictions}

From \eqref{fzerozero} and \eqref{Zzerofinal}, the flat space partition function at one loop reads
\be\label{lnZzero} \ln Z^{0} =-\frac{|c|\mu\ell^{2}}{96\pi^{2}} -\Bigl(\frac{|c|}{6} + 4-\frac{\mu\ell^{2}}{8\pi^{2}}\Bigr)\ln\frac{\ell}{2\pi} + \ln \Gamma\Bigl(3-\frac{\mu\ell^{2}}{8\pi^{2}}\Bigr) + O\bigl(|c|^{-1}\bigr) + \text{c.t.,}\ee
where c.t.\ stands for possible arbitrary counterterms.

Let us first make an elementary remark to clarify a potential puzzle. The $\ln\frac{\ell}{2\pi}$ term seems to be inconsistent with dimensional analysis, since $\ell$ is dimensionful. This is a familiar feature of one-loop calculations. One can replace $\ln\frac{\ell}{2\pi}$ by $\ln\frac{\ell}{\ell_{0}}$ for an arbitrary length scale $\ell_{0}$ by adding counterterms, a constant and an area counterterm in the present case. In the $\mu=0$ theory, if one uses a cut-off regularization, the only length scale that can appear before renormalization is the UV cut-off scale, and thus $\ell_{0}$ may be naturally identified with the UV cut-off in this case. This cut-off can be absorbed in counterterms and replaced by an arbitrary renormalization length scale. 

Let us now use Eq.\ \eqref{lnZzero} for $\mu=0$. It implies that $Z^{0}\propto \ell^{-\frac{|c|}{6}-4+O(|c|^{-1})}$. This result can be compared with the exact prediction \eqref{Zzerexact}. Using \eqref{critexponents} and \eqref{varthetaexp}, we find, when $c\rightarrow -\infty$,
\be\label{nuvthetaexp} 2\nu\vartheta + 1 = \frac{|c|}{6} + 4 + O\bigl(|c|^{-1}\bigr)\, .\ee
Recalling that $\beta=\ell$ in our case, Eq.\ \eqref{betaeqell}, we find a perfect match. 

When $\mu\not = 0$, Eq.\ \eqref{lnZzero} can be rewritten as
\be\label{lnZzerobis} \ln Z^{0} =-\frac{|c|\mu}{24}\Bigl(\frac{\ell}{2\pi}\Bigr)^{2\nu} -\Bigl(\frac{|c|}{6} + 4\Bigr)\ln\frac{\ell}{2\pi} + \ln \Gamma\Bigl(3-\frac{\mu\ell^{2}}{8\pi^{2}}\Bigr) + O\bigl(|c|^{-1}\bigr) + \text{c.t.}\ee
for $\smash{\nu = 1-\frac{6}{|c|} + O(|c|^{-2})}$. We thus observe that the $\smash{\frac{\mu\ell^{2}}{8\pi^{2}}\ln\frac{\ell}{2\pi}}$ correction is absorbed in an anomalous dimension for $\ell$, with an exponent matching the critical exponent $\nu$ given in Eq.\ \eqref{nuexp}, at least up to the one-loop order. Combining with the $\ell$-dependence at $\mu=0$ discussed in the previous paragraph, we see that our one-loop calculation precisely reproduces the prediction of the exact formulas for both the exponents $\nu$ and $\vartheta$. Let us also note, for later reference, that the one loop anomalous dimension term in the partition function entirely comes from the boundary determinant contribution, Eq.\ \eqref{bdJTpathintzero}.

Actually, using the fact that the function $g^{0}$ in the exact scaling formula \eqref{Zzerexact} for the partition function can only depend on the dimensionless parameter $\mu(\frac{\ell}{2\pi})^{2\nu}$, we can trust the ``improved'' tree-level term $\smash{-\frac{|c|\mu}{24}(\frac{\ell}{2\pi})^{2\nu}}$ beyond one loop. We claim, following a standard lore in quantum field theory, that it will generate correctly the leading logarithmic contributions in $(\ln\ell)^{k}$ at any $k^{\text{th}}$ loop order. Similarly, the one-loop term involving the $\Gamma$ function can also be improved, following the same logic, to make it consistent with the exact scaling form \eqref{Zzerexact}, by replacing $\frac{\mu\ell^{2}}{4\pi^2}$ by $\mu(\frac{\ell}{2\pi})^{2\nu}$ in its argument. This yields non-trivial predictions for the subleading logarithms, at any loop order greater than or equal to two. For instance, we get terms
\be\label{twoloopspredict} \Biggl[-\frac{3\mu\ell^{2}}{4\pi^{2}}\bigl(\ln\frac{\ell}{2\pi}\bigr)^{2} + \frac{3\mu\ell^{2}}{2\pi^{2}}\biggl[\Psi\Bigl(3-\frac{\mu\ell^{2}}{8\pi^{2}}\Bigr) -\frac{3}{4}\biggr]\ln\frac{\ell}{2\pi}\Biggr]\frac{1}{|c|}\ee
at two loops for $\ln Z^{0}$, the leading logarithm coming from the improved tree-level contribution and the subleading logarithm coming from both the improved tree-level and one-loop contributions, with $\Psi = \Gamma'/\Gamma$ the Euler $\Psi$ function. Overall, our improved one-loop result for the partition function in flat space reads
\be\label{Zzeroimproved} Z^{0} = \ell^{-1-2\nu\vartheta}\Gamma\Bigl(3-\frac{1}{2}\mu\bigl(\frac{\ell}{2\pi}\bigr)^{2\nu}\Bigr)e^{-\frac{|c|\mu}{24}(\frac{\ell}{2\pi})^{2\nu}}\quad\text{(one-loop improved).}\ee
Of course, we do not claim that the above formula is ``exact'' in any sense. In particular, it displays the instability for $\mu(\frac{\ell}{2\pi})^{2\nu}\geq 6$, which is inherent to the $c\rightarrow -\infty$ expansion, but which is not present in the exact finite $c$ theory, where a crossover is expected to occur from a smooth inflated geometry to a branched polymer phase when $\mu$ is very large and positive.

The same logic can be applied in non-zero curvature, with the additional subtlety that we now have a natural IR scale $L$ given by the curvature of space. We focus on the case of negative curvature, for conciseness, since the case of positive curvature can be discussed exactly along the same lines. Improving the tree-level term \eqref{fzeromin} by replacing $x=\ell/(2\pi L)$ by the renormalized parameter
\be\label{renormx} \hat x = \frac{1}{L}\Bigl(\frac{\ell}{2\pi}\Bigr)^{\nu}\ee
to take into account the length anomalous dimension, and expanding at large $|c|$, produces the one-loop logarithmic contribution
\be\label{improve1} -\Bigl(1-\sqrt{1+ x^2}- \frac{(y-1) x^2}{2\sqrt{1+ x^2}}\Bigr)\ln\frac{\ell}{2\pi}\,\cdotp\ee
Note that, as usual, replacing $\ln\frac{\ell}{2\pi}$ by $\ln\frac{\ell}{\ell_{0}}$ for an arbitrary scale $\ell_{0}$ in this expression is immaterial, because the modification can be absorbed in counterterms, whose general forms are given in Eq.\ \eqref{SJTpmctoneloop}. As in zero curvature, this anomalous dimension contribution is generated by the boundary determinant, Eq.\ \eqref{bdJTpathint3}. There is an additional logarithmic piece
\be\label{addimp} 2\Bigl(1-\sqrt{1+ x^2}\Bigr)\ln\frac{\ell}{2\pi}\ee
in the full one-loop result, Eq.\ \eqref{Zminusfinal}, whose origin can be traced back to the contribution of the bulk determinant \eqref{Bdetnegx}. Such $\mu$-independent logarithmic terms can exist in non-zero curvature because there is a scale $L$ in the problem which allows to interpret them as contributions, proportional to $\ln\frac{\ell}{L}$, to the function $g^{-}(0,\beta/L^{1/\nu},c)$ defined in \eqref{Zpmerexact}. Overall, we can thus write an improved version of $\ln Z^{-}$ as
\begin{multline}\label{lnZminimpr} \ln Z^{-} = -\Bigl(\frac{|c|}{6}+4\Bigr)\ln\frac{\ell}{2\pi} +\frac{|c|}{6}\ln\frac{1+\sqrt{1+\hat x^{2}}}{2}- \frac{|c|}{12}(1+y)\bigl(\sqrt{1+\hat x^{2}}-1\bigr)\\
 + 2\Bigl(1-\sqrt{1+ \hat x^2}\Bigr)\ln\hat x-\frac{1}{2}\ln (1+\hat x^2) 
+\ln \Gamma\Bigl(2+\sqrt{1+ \hat x^2}- \frac{(y-1)  \hat x^2}{2\sqrt{1+ \hat x^2}}\Bigr) \\\text{(one-loop improved).}\end{multline}
The first term, proportional to $\ln\frac{\ell}{2\pi}$, gives the overall power of $\ell$ in the partition function, with the same critical exponent as in the flat case, see Eq.\ \eqref{Zzerexact}, \eqref{Zpmerexact} and \eqref{nuvthetaexp}, whereas the other terms contribute to $\ln g^{-}$, where $g^{-}$ is the dimensionless function defined in \eqref{Zpmerexact}.

\subsection{\label{semidisSec}Naive semiclassical limit and expected area}

At the classical level, and after integrating out the dilaton field $\Phi$ that imposes the constraint of constant curvature, the action for JT gravity reduces to the cosmological constant term given in Eq.\ \eqref{Scosmodef}. The limit $\La\rightarrow -\infty$ at fixed boundary length $\ell$ may therefore seem to correspond, at least superficially, to a semi-classical limit, with an expansion parameter $1/|\La|$ playing the role of $\hbar$. At leading order, the configurations that dominate must maximise the area for a given $\ell$. By the standard isoperimetric inequalities, for negative and zero curvature, this is achieved precisely by round disks embedded in hyperbolic space or in Euclidean space, which are the saddle point solutions that we have found in Section \ref{clsolJTSec}, Eqs.\ \eqref{JTminsaddlesol1} and \eqref{JTzerosaddlesol}. In positive curvature, the area is not bounded above, even at fixed $\ell$ \cite{Fer2}. For this reason, we restrict the discussion to the cases of zero and negative curvature in this subsection. 

We are going to show that our results strongly suggest that the limit $\La\rightarrow -\infty$ at fixed boundary length $\ell$ is actually \emph{not} a semi-classical limit in JT gravity. This confirms explicitly the general arguments and physical picture presented in \cite{Fer2}.

In the framework of the present paper, we take the $c\rightarrow -\infty$ limit first. The naive semi-classical limit (not to be confused with the $c\rightarrow -\infty$ limit!) is then
\be\label{naivesemiclassical} \mu\rightarrow - \infty\, ,\quad \ell\quad \text{fixed.}\ee
One could study directly the partition functions in this limit, but to better understand what is going on it is more physical and instructive to look at the expectation value of the area
\be\label{areavev} \langle A\rangle^{\eta} = -\frac{24\pi}{|c|}\frac{\partial\ln Z^{\eta}}{\partial\mu} = -24\pi\biggl( \frac{\partial f_{0}^{\eta}}{\partial\mu} + \frac{1}{|c|}\frac{\partial f_{1}^{\eta}}{\partial\mu} + O\bigl(|c|^{-2}\bigr)\biggr)\, .\ee

Let us first focus on the theory in flat space. From \eqref{fzerozero} and \eqref{Zzerofinal}, we get
\be\label{areaflat1} \langle A\rangle^{0} = \frac{\ell^{2}}{4\pi}-\frac{3\ell^{2}}{\pi|c|}\ln\frac{\ell}{2\pi} + \frac{3\ell^{2}}{\pi |c|}\Psi\Bigl(3-\frac{\mu\ell^{2}}{8\pi^{2}}\Bigr) + O\bigl(|c|^{-2}\bigr)\, .\ee
In the limit \eqref{naivesemiclassical}, this yields
\be\label{areaflat2} \langle A\rangle^{0} =\frac{\ell^{2}}{4\pi}-\frac{3\ell^{2}}{\pi|c|}\ln\frac{\ell}{2\pi}+ \frac{3\ell^{2}}{\pi|c|}\ln \frac{|\mu|\ell^{2}}{8\pi^{2}} +
O\bigl(|c|^{-2}\bigr) + O\bigl(|\mu|^{-1}\bigr)\, .\ee
The first logarithm on the right-hand side of this equation is interpreted, as in the previous subsection, as coming from the anomalous dimension of $\ell$. It is thus absorbed by replacing $(\frac{\ell}{2\pi})^{2}$ by $(\frac{\ell}{2\pi})^{2\nu}$ in the leading term $\ell^{2}/(4\pi)=\pi (\frac{\ell}{2\pi})^{2}$. The second logarithm, on the other hand, points to a divergence of the area when $|\mu|\rightarrow\infty$. Of course, strictly speaking, our one-loop calculation can be trusted only if the one-loop correction is much smaller than the leading correction. But the leading logarithms can be resummed, along the lines of the previous subsection, to yield an improved area expectation value that we may write in the following form,
\be\label{areaflatimproved} \langle A\rangle^{0} \underset{\mu\rightarrow -\infty}{\sim}\pi \Bigl(\frac{\ell}{2\pi}\Bigr)^{2\nu}\biggl[\frac{|\mu|}{2}\Bigl(\frac{\ell}{2\pi}\Bigr)^{2\nu}\biggr]^{2\xi}\, ,\ee
for a new critical exponent $\xi$ that we predict to be such that
\be\label{xiexpexp} \xi = \frac{6}{|c|} + O\bigl(|c|^{-2}\bigr)\, .\ee
One may notice that $\xi = 1-\nu$ at this order, but we have no rationale to believe that such a simple relation between $\nu$ and $\xi$ has to be valid to all orders in the $1/|c|$ expansion. Actually, one may conjecture the relation \eqref{xinurelconj} between $\xi$ and $\nu$, see \cite{Fer2} and the discussion below.

The divergence of the expected area in the limit \eqref{naivesemiclassical}, predicted by Eq.\ \eqref{areaflatimproved}, immediately implies that this limit is not dominated by a given smooth background. Thus, we do not have a semiclassical expansion in the limit \eqref{naivesemiclassical} in JT gravity \cite{Fer2}. A crucial physical feature at the origin of this effect is the fact that the boundary is a fractal in the finite $|c|$ theory. In particular, its geometrical length is ill-defined (infinite) and thus the isoperimetric inequality does not prevent the divergence of the enclosed area, even at fixed finite quantum length $\ell$.

We can repeat the above analysis in the negative curvature model. Using \eqref{fzeromin} and \eqref{Zminusfinal}, we find the large $|c|$ expansion
\begin{multline}\label{areacurv1} \langle A\rangle^{-} = 2\pi L^{2}\bigl(\sqrt{1+{x^{2}}}-1\bigr) \\ - \frac{3\ell^{2}}{\pi |c|}\frac{1}{\sqrt{1+x^{2}}}\biggl[\ln\frac{\ell}{2\pi} - \Psi\Bigl(2+\sqrt{1+x^{2}}+\frac{(1+|\mu|L^{2})x^{2}}{\sqrt{1+x^{2}}}\Bigr)\biggr] + O\bigl(|c|^{-2}\bigr)\, ,
\end{multline}
generalising \eqref{areaflat1}. Taking $\mu$ to be large and negative, without making any asumption on $x$ or on its renormalized version $\hat x$ defined in Eq.\ \eqref{renormx}, we can find an improved form of the expectation value,
\be\label{areacurvedimp} \langle A\rangle^{-} = 2\pi L^{2}\Biggl[ \biggl( 1 + \hat x^{2}\Bigl(\frac{|\mu|(\frac{\ell}{2\pi})^{2\nu}}{\sqrt{1+\hat x^{2}}}\Bigr)^{2\xi}\biggr)^{1/2}-1\Biggr]\quad \text{for $y= \mu L^{2}\ll -1$,} \ee
generalising \eqref{areaflatimproved}. This formula shows that the model can be in two qualitatively different regimes when $\mu L^{2}$ is large and negative. If $\hat x\ll 1$, the right-hand side of Eq.\ \eqref{areacurvedimp} goes to the flat space regime described by \eqref{areaflatimproved}. On the other hand, if 
\be\label{Schcross} \hat x\biggl(|\mu|\Bigl(\frac{\ell}{2\pi}\Bigr)^{2\nu}\biggr)^{\xi}\gg 1\, ,\ee
one gets
\be\label{areacurvedimpSch} \langle A\rangle^{-} \sim 2\pi L\Big(\frac{\ell}{2\pi}\Big)^{\nu} \biggl(\frac{|\mu|(\frac{\ell}{2\pi})^{2\nu}}{\sqrt{1+\hat x^{2}}}\biggr)^{\xi}\, .\ee
As in the flat space case, the expected area thus diverges when $\mu\rightarrow -\infty$, showing that the limit is not semi-classical in the negative curvature theory either. But the asymptotic regime is of a very different nature than in flat space, because the typical geometries are very large compared to the curvature length scale. If one assumes that $\hat x\gg 1$ on top of \eqref{Schcross}, the formula further simplifies to
\be\label{areaSchlim} \langle A\rangle^{-} \sim 2\pi L\Bigl(\frac{\ell}{2\pi}\Bigr)^{\nu} \biggl(|\mu|\Bigl(\frac{\ell}{2\pi}\Bigr)^{\nu}L\biggr)^{\xi}\, .\ee
The regime in which \eqref{areaSchlim} is valid is conjectured to match with the regime where an effective Schwarzian description of the model is possible \cite{Fer2}. In this regime, the curvature length scale $L$ actually plays the role of an ultraviolet cut-off!

The asymptotic formula \eqref{areaSchlim} allows to identify the macroscopic smooth length parameter $\ell_{\text S}$ \cite{Fer2}, that is traditionally used in the Schwarzian description of JT gravity, see Eq.\ \eqref{Schlimit} and below. The idea \cite{Fer2} is that for a large, near-hyperbolic effective smooth geometry, one has
\be\label{smoothmacA} \langle A\rangle^{-} \sim L\ell_{\text S}\, .\ee
Comparing with \eqref{areaSchlim} yields
\be\label{ellSmic} \ell_{\text S} =  2\pi \Bigl(\frac{\ell}{2\pi}\Bigr)^{\nu} \biggl(|\mu|\Bigl(\frac{\ell}{2\pi}\Bigr)^{\nu}L\biggr)^{\xi}\, .\ee
It is interesting to note that the overall power of $\ell$ in the right-hand side of this equation is one, at the order we have computed, since $(1+\xi)\nu = 1+O(|c|^{2})$. It is natural to conjecture \cite{Fer2} that this property will remain valid at finite $c$ and thus that
\be\label{xinurelconj} \xi = \frac{1}{\nu} -1\, .\ee
A simple argument in favour of this identity is that boundary cosmological constant counterterms must be proportional to the quantum length parameter $\ell$ in the microscopic theory \cite{Fer2}. Consistency with the treatment of such counterterms in the effective description then implies that $\ell_{\text S}$ must be proportional to $\ell$. More evidence for this relation is given in \cite{Fer2}.

Let us note that the above discussion follows closely and is perfectly consistent with the general discussion in \cite{Fer2}, which does not rely on the $c\rightarrow -\infty$ limit. To close this paper, we are now going to discuss further the Schwarzian limit of the negative curvature partition function in our framework and find nice confirmations of the interpretations we have just given.

\subsection{The Schwarzian limit}

The Schwarzian limit, or near-hyperbolic limit, or near-AdS limit, is a particular asymptotic regime of the negative curvature theory defined by \eqref{Schlimit}. In the limit, the disk geometry is usually assumed to be close to the geometry of the full hyperbolic space, with a smooth and gently wiggling boundary embedded in hyperbolic space, described by the so-called reparameterization ansatz and governed by the Schwarzian action. As emphasized in \cite{Fer1,Fer2}, and as should be clear from the discussion of the present paper, the Schwarzian description is nothing more than an \emph{effective} description of JT gravity, valid on length scales much greater than $L$. In particular, the parameter $\ell_{\text S}$ appearing in the very definition \eqref{Schlimit} of the limit, which is usually presented as being the length of the disk boundary, is actually an effective parameter itself. We have outlined, in the previous subsection, how the Schwarzian ``macroscopic'' regime can emerge, in the framework of the $c\rightarrow -\infty$ limit, with an effective length parameter given by \eqref{ellSmic}; see also the detailed discussion in \cite{Fer2}.

Our goal in this last subsection is to work out some more details on the limit, emphasizing the crucial difference between the microscopic description of the theory and the effective, long-wavelength description associated with the Schwarzian, and giving more evidence for the emergence of this description.

Let us first note that, taking into account the rescaling of the cosmological constant, Eq.\ \eqref{Zamolodlimit}, the Schwarzian limit corresponds in our case to 
\be\label{Schclim} y=\mu L^{2}\rightarrow -\infty\, ,\quad \ell_{\text S}\rightarrow\infty\, ,\quad \frac{3\ell_{\text S}}{|\mu| L^{3}} = b_{\text S} = \text{constant,}\ee
the fixed parameter $b_{\text S}$ being related to the usual parameter $\beta_{\text S}$ by
\be\label{bbetarel} b_{\text S} = |c|\beta_{\text S}\, .\ee
Let us emphasize that the scaling \eqref{Schclim} used to define the Schwarzian limit is a renormalized version of the naive scaling that would involve $\ell$ instead of $\ell_{\text S}$. Using this renormalized version is natural from the discussion of the previous subsection and from the general discussion in \cite{Fer2}, but the statement that this is the correct way to define the limit is a non-trivial claim. The naive ``bare'' parameter, defined by
\be\label{bbare} \tilde b_{\text S} = \frac{3\ell}{|\mu|L^{3}}\, \cvp\ee
is related to the ``renormalized'' parameter $b_{\text S}$ by
\be\label{bbref} b_{\text S} =\tilde b_{\text S}(|\mu| L)^{\xi}= \tilde b_{\text S}\Bigl(1 + \frac{6}{|c|}\ln\bigl(L^{-1}|y|\bigr)+ O\bigl(|c|^{-2}\bigr)\Bigr)\, .\ee
This is found by using \eqref{ellSmic}, \eqref{nuexp} and \eqref{xiexpexp}. Clearly, keeping  $\tilde b_{\text S}$ or $b_{\text S}$ fixed in the limit are two inequivalent prescriptions. We are going to check explicitly below that it is indeed the renormalized scaling, for which $b_{\text S}$ is kept fixed (and thus $\tilde b_{\text S}\rightarrow 0$), that yields a well-defined limit. For future reference, let us express the Schwarzian scaling in terms of our usual variables $x$ and $y$,
\be\label{Schxybrel} x\rightarrow \infty\, ,\quad y\rightarrow -\infty\, ,\quad b_{\text S} = \tilde b_{\text S}\bigl(L^{-1}|y|\bigr)^{\xi} = \frac{6\pi x}{|y|} \bigl(L^{-1}|y|\bigr)^{\xi} \quad\text{fixed.}\ee

\subsubsection{Leading order} The leading order partition function \eqref{fzeromin} reduces to
\be\label{f0Sch} f_{0}^{-} = -\frac{1}{12}(1-|y|)(x-1) -\frac{1}{6}\ln (2 L) +\frac{|y|}{24 x}+ o(1)\, ,\ee
where $o(1)$ represents terms that go to zero in the Schwarzian limit. The $x$-independent (or, equivalently, $\ell$-independent) terms, and the terms proportional to $x$ (or equivalently, proportional to $\ell$ or to $\ell_{\text S}$), are counterterms, interpreted as renormalizing the zero-temperature entropy and ground state energy of the holographic dual, respectively (we are using this language here, because it is traditional in the Schwarzian limit). We shall no longer indicate these terms explicitly. The physically relevant piece is the term $\frac{|y|}{24x}$. At leading order in the $c\rightarrow -\infty$ expansion, the background is smooth and classical, $\ell_{\text S}$ and $\ell$ are identified and there is no distinction between the naive and the renormalized scaling. This term then yields the familiar
\be\label{Schtreelevel} \frac{\pi |c|}{4 b_{\text S}} = \frac{\pi}{4\beta_{\text S}}\ee
tree-level contribution to the logarithm of the Schwarzian partition function. For future reference, let us also note that, taking into account the form of the renormalized scaling \eqref{Schxybrel}, the $\frac{|y|}{24x}$ contribution also produces a term
\be\label{renscaoneloop} \frac{3\pi}{2 b_{\text S}}\ln\bigl(L^{-1}|y|\bigr)
\ee
at order $|c|^{0}$.

\subsubsection{The Schwarzian quadratic action}

At one loop, it is useful, to help better understand the physics, to look at the Schwarzian limit of each individual contributions to the partition function. The Fadeev-Popov determinant for the $\PslR$ gauge fixing (prefactor in Eq.\ \eqref{Zloop1}) and the contribution from the zero modes, Eq.\ \eqref{zeromodefactor}, yield together a $\frac{1}{2}\ln\frac{\ell}{2\pi}$ term. The contribution of the on-shell Liouville action, given in Eq.\ \eqref{SLsadJTmin}, yields a $\ell$-independent term in the Schwarzian limit, as already mentioned. Similarly, the Schwarzian limit of the bulk determinant, Eq.\ \eqref{Bdetnegx}, yields $\ell$-independent terms or terms proportional to $\ell$. On top of the $\frac{1}{2}\ln\frac{\ell}{2\pi}$ mode, the only non-trivial contribution in the Schwarzian limit is thus coming from the integral over the boundary modes.

It is very instructive to look back at the action for the boundary quadratic fluctuations, Eq.\ \eqref{tSeffcurved}, and first implement the Schwarzian limit directy in this action. Discarding the three modes $\xi_{0}$ and $\xi_{\pm 1}$ that have already been integrated out, and taking the limit \eqref{Schxybrel}, we get the following action,
\be\label{SqSchmic} S^{(2)}_{\text{S},\,\text{mic}} = 4\pi\sum_{n\geq 2}
\frac{n+\frac{3\pi x^{2}}{b_{\text S}} + \frac{3}{2}x - \frac{3\pi}{2 b_{\text S}}}{(n+x)^{2}} (n^{2}-1)|\xi_{n}|^{2}\, .\ee
Note that since we are discussing here the one-loop contribution, and since we restrict our discussion to the one-loop order, we can use the naive form of the scaling to derive this action, with $b_{\text S}=6\pi x/|y|$. Corrections will only contribute to order $|c|^{-1}$ in the logarithm of the partition function.

The subscript ``S'' in $S^{(2)}_{\text{S},\,\text{mic}}$ indicates that we are in the Schwarzian limit, and ``mic'' that we have made no assumption on the ratio between the mode number $n$ and $x$ in deriving the action. In particular, this action has the correct microscopic properties, with a leading UV limit in $n|\xi_{n}|^{2}$, which predicts that $|\xi_{n}|\sim 1/\sqrt{n}$ at large $n$. This behaviour is a landmark feature of the microscopic definition of JT \cite{Fer1}. It is at the origin of the fractal nature of the boundary. 

Let us emphasize that the leading UV term $n|\xi_{n}|^{2}$ comes entirely from the Liouville piece in the action. The area term does not play any role, even though it is this area term, and only this area term, that is taken into account in the usual approach to JT, see the discussion below.

\subsubsection{The effective theory} If we take the Schwarzian $x\rightarrow\infty$ limit in each term of the sum in Eq.\ \eqref{SqSchmic}, we obtain a much simpler effective action
\be\label{SqSchmac} S^{(2)}_{\text{S},\,\text{eff}} = \frac{12\pi^{2}}{b_{\text S}}\sum_{n\geq 2} (n^{2}-1)|\xi_{n}|^{2}\, .\ee
Several important comments are here in order.

\noindent i) The procedure of taking the $x\rightarrow\infty$ limit term-by-term is of course not justified mathematically. For any $x$, even very large, there will always be UV modes for which the action \eqref{SqSchmac} is not appropriate. Explicitly, this action predicts a UV behaviour in $n^{2}|\xi_{n}|^{2}$, whereas the correct UV behaviour is in $n|\xi_{n}|^{2}$, as mentioned above. The simplified action is correct for modes that are such that $n\ll x$, that is to say, for modes having a wavelength $\ell/n$ much larger than the curvature scale $L$. In other words, it is a long-distance effective description for which $L$ plays the role of a UV cut-off scale. 

\noindent ii) The action $S^{(2)}_{\text{S},\,\text{eff}}$ comes entirely from the area term (it is proportional to $1/b_{\text S}$ which is itself proportional to the cosmological constant). The Liouville piece does not play any role in it. Even if it comes from the area term only, it is not the exact area term, that has a UV behaviour independent of $n$ (this can be read off for instance by looking at the large $n$ limit of the term proportional to $\mu$ in \eqref{tSeffcurved}). It corresponds to a long wavelength approximation to it. 

\noindent iii) The effective action \eqref{SqSchmac} matches precisely the usual Schwarzian action in the quadratic approximation (which is known to be exact for the Schwarzian theory, see e.g.\ the second reference in \cite{JTappli1}). This can be checked directly and rather straightforwardly by relating the boundary Liouville field to the circle diffeomorphism in the Schwarzian limit \cite{Fer3}, but we shall not try to explain this here since the ideas required to present the argument are too far removed from the ideas developed in the present paper; we refer to \cite{Fer3} for a detailed discussion. Instead, we directly check  that we indeed get the usual Schwarzian partition function by using the action \eqref{SqSchmac}. Indeed, this action produces the contribution
\be\label{Scheffdet} \prod_{n\geq 2}\biggl[\frac{b_{\text S}e^{\sigma_{*}}}{12\pi^{2}}\frac{1}{n^{2}-1}\biggr] = D^{-1}\Bigl(1,3;;\frac{12\pi^{2}}{b_{\text S}}e^{-\sigma_{*}}\Bigr) = 8\sqrt{27}\pi^{2}\bigl(e^{\sigma_{*}}b_{\text S}\bigr)^{-3/2}\ee
to the partition function. If we recall that $\ell\propto b_{\text S}$ and take into account the $\frac{1}{2}\ln\frac{\ell}{2\pi}$ term in the logarithm of the partition function coming from Faddeev-Popov and the integration over $\xi_{0}$ and $\xi_{\pm 1}$, together with the tree-level contribution given in \eqref{f0Sch},  we get, up to the usual counterterm ambiguities (here a multiplicative factor that does not depend on $\ell$, or on $b_{\text S}$),
\be\label{ZScheff} Z_{\text{S},\,\text{eff}} \propto b_{\text S}^{\frac{1}{2}-3}e^{\frac{\pi |c|}{4 b_{\text S}}}=b_{\text S}^{-\frac{5}{2}}e^{\frac{\pi}{4 \beta_{\text S}}}\, .\ee
This is the usual Schwarzian partition function, modulo an additional $1/b_{\text S}$ (or equivalently $1/\ell$) factor. This $1/\ell$ discrepancy comes from the fact that the partition functions considered in the present paper are the usual disk quantum gravity partition functions with no marked point on the boundary, whereas the Schwarzian partition function in its traditional form actually corresponds to having a marked point on the boundary.\footnote{This is explained in \cite{Fer2,Fer3}. Note that the precise meaning of the partition function, as the continuum limit of a precise enumeration problem, is known only if a microscopic definition of the model is given. Without such a precise definition, it is meaningless to ask whether the effective Schwarzian partition function corresponds to a counting with a marked point on the boundary or not. It turns out that the usual formula, with the partition function proportional to $\smash{b_{\text S}^{-3/2}}$, corresponds to a counting with a marked point on the boundary, and therefore has an additional factor $\ell\propto b_{\text S}$ compared with the usual quantum gravity partition function considered in the present work. This was also found to be the case in early attemps to provide a microscopic definition of the models, based on a Brownian boundary \cite{KitaevSuhmic} or on a self-avoiding loop boundary \cite{StanfordYangmic}.}

\subsubsection{The microscopic theory}

Let us now study the Schwarzian limit by including the one-loop corrections. One can equivalently use \eqref{SqSchmic} or study directly the limit on the exact one-loop result, starting either from \eqref{Zminusfinal} or from \eqref{lnZminimpr}. We found it slightly more convenient to start from Eq.\ \eqref{lnZminimpr}, in which some of the one-loop terms have been absorbed in the tree-level contribution by renormalizing $x$ to $\hat x$. Performing the expansion in the limit \eqref{Schxybrel}, one finds several diverging terms that, superficially, seem to indicate that the limit does not exist. As we now explain, remarkably, all these terms can be dealt with.

For instance, the most diverging terms in $\ln Z^{-}$, which, when $b_{\text S}$ is fixed, are of order $x^{2}\ln x$ and $x^{2}$, read
\be\label{Schdiv1} \frac{|c|}{12}|y|\hat x + \frac{3\pi x^{2}}{b_{\text S}}\biggl(\ln \frac{3\pi x^{2}}{b_{\text S}}-1\biggr)\,\cdotp\ee
These terms can be rewritten as
\be\label{Schdiv2} \frac{|c|}{12}|y|\hat x\Bigl(\frac{1}{2}|y|\hat x\Bigr)^{\xi}
-\frac{1}{2}|y|\hat x + O\bigl(|c|^{-1}\bigr) = \Biggl[\frac{|c||y|}{12 L}\biggl(\frac{|y|}{2L}\biggr)^{\xi} - \frac{|y|}{2L} \Biggr] \frac{\ell}{2\pi} + O\bigl(|c|^{-1}\bigr) \ee
and are thus proportional to $\ell$. They renormalize the ground state energy and can be discarded. The next-to-leading diverging terms are found to be in $x\ln x$ and $x$. They read
\be\label{Schdiv3} -\frac{|c|}{12}\bigl(\hat x+|y|\bigr) -2 x\ln x +\frac{3}{2}x\ln\frac{3\pi x^{2}}{b_{\text S}} = -\frac{|c|}{12}|y| -\frac{|c|}{12 L^{1+\xi}}\Bigl(\frac{|y|}{2}\Bigr)^{-3\xi}\frac{\ell}{2\pi}\ee
and can thus also be discarded, renormalizing both the ground state energy and the zero temperature entropy. One also find finite terms that are $\ell$-independent or proportional to $\ell$, that we can of course discard as well. There remains non-trivial constant terms and diverging $\ln (x|y|)$ terms,
\be\label{lnZmSf1} \ln Z^{-}\equiv -\frac{5}{2}\ln x + \frac{|c|}{24}\frac{|y|}{\hat x} - \frac{3\pi}{2 b_{\text S}}\ln\bigl(\frac{1}{2}x|y|\bigr) + o(1) + O\bigl(|c|^{-1}\bigr)\, .\ee
As before, $o(1)$ stands for terms that go to zero in the Schwarzian limit, and $\equiv$ means equality modulo terms that renormalize the zero temperature entropy or the ground state energy. 

The $-\frac{5}{2}\ln x$ contribution yields the correct prefactor $1/b_{\text S}^{5/2}$ in front of the partition function, with the additional $1/b_{\text S}$ with respect to the standard convention, consistently with Eq.\ \eqref{ZScheff} and the discussion below this equation. However, the diverging term in $\ln(x|y|)$ may seem problematic for the existence of the limit. If we were using the naive scaling, with $\tilde b_{\text S}$ defined in \eqref{bbare} fixed, we would indeed find an unwanted diverging piece $\smash{-\frac{3\pi}{2b_{\text S}}\ln|y|}$. But this piece is precisely cancelled when one uses the correct renormalized scaling \eqref{Schxybrel}, since it produces the additional contribution given in Eq.\ \eqref{renscaoneloop}. Overall, we thus obtain
\be\label{Schdiv4} \ln Z^{-}\equiv -\frac{5}{2}\ln b_{\text S} + \frac{2^{\xi}\pi |c|}{4 b_{\text S}} + O\bigl(|c|^{-1}\bigr)\, .\ee
The factor $2^{\xi}$ can of course be absorbed by redefining $b_{\text S}$. This result is a strong direct evidence that the usual effective Schwarzian theory emerges from the microcopic description of JT gravity. In our case, the UV effects coming from the action \eqref{SqSchmic} are taken into account by renormalizing the scaling limit according to \eqref{Schxybrel}. This is as expected from our previous discussion and consistent with the general ideas presented in \cite{Fer2}.

\section*{Acknowledgements}

This work is partially supported by the International Solvay Institutes and the Belgian Fonds  de la Recherche Scientifique F.R.S.-FNRS (convention IISN 4.4503.15). The work of S.C.\ is supported by a postdoctoral research fellowship of the Belgian F.R.S.-FNRS. F.F.\ would like to thank Junggi Yoon and the APCTP in Pohang, South Korea for providing him with an exceptional scientific and working environment during which part of this work was done.

\appendix\clearpage

\section{\label{bddetApp}Boundary determinants}

The boundary determinants encountered in the main text are of the following general form. Let $\hat N$ be the operator whose eigenvalues, which are non-degenerate, are the natural integers $n\in\mathbb N$. Denoting $\{a\}_{p} = (a_{1},\ldots,a_{p})$ and $\{b\}_{q} = (b_{1},\ldots,b_{q})$, we want to compute
\be\label{gendetbdApp} D\bigl(\{a\}_{p};\{b\}_{q};z\bigr) = 
\det\Biggl[z \frac{\prod_{i=1}^{p}\bigl( \hat N + a_{i} \bigr)}{
\prod_{i=1}^{q}\bigl( \hat N + b_{i} \bigr)}\Biggr]\, ,\ee
for $p > q$ and $z$, the $a_{i}$s and the $b_{i}$s are strictly positive real numbers. The determinant is defined in the $\zeta$-function scheme. Explicitly, the $\zeta$-function relevant for computing this determinant is $\tilde\zeta(s) = z^{-s}\zeta(s)$, where
\be\label{zetAppprod} \zeta (s) = \sum_{n\geq 0}\biggl[\frac{\prod_{i=1}^{q}(n+b_{i})}{\prod_{i=1}^{p}(n+a_{i})}\biggr]^{s}\, .\ee
Since $\tilde\zeta'(0) = -(\ln z)\zeta(0) + \zeta'(0)$, we have
\be\label{DdetdefApp} D\bigl(\{a\}_{p};\{b\}_{q};z\bigr) = e^{-\tilde\zeta'(0)} = z^{\zeta(0)}e^{-\zeta'(0)}\, .\ee

Let us start with the following special case.
\begin{proposition} 
\be\label{Dspe1} D(a;;z) = z^{\frac{1}{2}-a} D(a;;1) = z^{\frac{1}{2}-a}\frac{\sqrt{2\pi}}{\Gamma (a)}\,\cdotp\ee
\end{proposition}
\begin{proof}
The $\zeta$-function \eqref{zetAppprod} is simply the Hurwitz $\zeta$-function
\be\label{zetApp1def} \zeta(s) = \sum_{n\geq 0}\frac{1}{(n+a)^{s}} = \zeta_{\text H}(s;a)\ee
in this case. Eq.\ \eqref{Dspe1} then follows from \eqref{DdetdefApp} and
\be\label{Hurwitzid} \zeta_{\text H}(0;a) = \frac{1}{2}-a\, ,\quad \zeta_{\text H}'(0;a) = \ln\frac{\Gamma (a)}{\sqrt{2\pi}}\,\cdotp\ee
\end{proof}

To deal with the more general case, we have to take into account the fact that, in the $\zeta$-function regularization scheme, the determinant of a product of operators is not necessarily equal to the product of the determinants of the operators. This phenomenon is called the ``multiplicative anomaly'', see \cite{multanomaly}. We are going to prove the following result.\footnote{This is similar to a result relevant for shifted Laplace operators on compact Riemann surfaces proven in the Appendix A of the first reference in \cite{FerK2}.}
\begin{proposition}
For $p>q$,
\be\label{Appprop} 
\det\biggl[z \frac{\prod_{i=1}^{p}\bigl( \hat N + a_{i} \bigr)}{
\prod_{i=1}^{q}\bigl( \hat N + b_{i} \bigr)}\biggr] =z^{\frac{1}{2}-\frac{1}{p-q}(\sum_{i=1}^{p}a_{i}-\sum_{i=1}^{q}b_{i})} \frac{\prod_{i=1}^{p}\det (\hat N + a_{i})}{\prod_{i=1}^{q}\det (\hat N + b_{i})}\, \cdotp\ee
\end{proposition}
In particular, there is no multiplicative anomaly when $z=1$.
\begin{proof}
Consider
\be\label{ZdefApp} Z(s) = \zeta\Bigl(\frac{s}{p-q}\Bigr) - \frac{1}{p-q}\Bigl(\sum_{i=1}^{p}\zeta_{\text H}(s;a_{i}) - \sum_{i=1}^{q}\zeta_{\text H}(s;b_{i})\Bigr)\, .\ee
This function admits a series representation
\be\label{ZseriesApp} Z(s) = \sum_{n\geq 0} Z_{n}(s)\ee
for coefficients given by
\be\label{ZndefApp} Z_{n}(s) = \biggl[\frac{\prod_{i=1}^{q}(n+b_{i})}{\prod_{i=1}^{p}(n+a_{i})}\biggr]^{\frac{s}{p-q}} - \frac{1}{p-q}\biggl[\sum_{i=1}^{p}\frac{1}{(n+a_{i})^{s}}-\sum_{i=1}^{q}\frac{1}{(n+b_{i})^{s}} \biggr]\, .\ee
A straightforward calculation shows that $Z_{n}(s)= O(1/n^{s+2})$ when $n\rightarrow\infty$. This implies that the series representation \eqref{ZseriesApp} converges normally as soon as $\re s>-1$. This series representation can thus be used to evaluate $Z(0)$ and $Z'(0)$; since  $Z_{n}(0) = 0$ and $Z_{n}'(0) = 0$, we get $Z(0)=Z'(0)=0$. From \eqref{ZdefApp} we thus obtain
\begin{align}\label{ZzeroApp} & \zeta(0) = \frac{1}{p-q}\Bigl(\sum_{i=1}^{p}\zeta_{\text H}(0;a_{i}) - \sum_{i=1}^{q}\zeta_{\text H}(0;b_{i})\Bigr) = \frac{1}{2} - \frac{1}{p-q}\Bigl(\sum_{i=1}^{p} a_{i} - \sum_{i=1}^{q} b_{i}\Bigr)\, , \\
\label{ZpzeroApp} & \zeta'(0) = \sum_{i=1}^{p}\zeta_{\text H}'(0;a_{i}) - \sum_{i=1}^{q}\zeta_{\text H}'(0;b_{i})\, .
\end{align}
Together with Eq.\ \eqref{DdetdefApp}, this yields \eqref{Appprop}.
\end{proof}

Overall, we have thus obtained
\be\label{BddetApp} D\bigl(\{a\}_{p};\{b\}_{q};z\bigr) =(2\pi)^{\frac{p-q}{2}}z^{\frac{1}{2}-\frac{1}{p-q}(\sum_{i=1}^{p}a_{i}-\sum_{i=1}^{q}b_{i})}\frac{\prod_{i=1}^{q}\Gamma(b_{i})}{\prod_{i=1}^{p}\Gamma(a_{i})}\,\cvp\ee
a formula used in the main text, Section \ref{oneloopSec}, to evaluate the boundary path integrals in Liouville and JT gravity.


\begin{thebibliography}{99}
%

%
\bibitem{Fer1}{F.~Ferrari, ``Jackiw-Teitelboim Gravity, Random Disks of Constant Curvature, Self-Overlapping Curves and Liouville $\text{CFT}_{1}$,'' arXiv:2402.08052 [hep-th].}
%
\bibitem{Fer2}{F.~Ferrari, ``Random Disks of Constant Curvature: the Lattice Story,''  arXiv:2406.06875 [hep-th].}
%
\bibitem{Fer3}{F.~Ferrari, ``Random Disks of Constant Curvature: the Conformal Gauge Story,'' to appear.}
%
\bibitem{JTappli1}{A.~Almheiri and J.~Polchinski, ``Models of AdS$_{2}$ backreaction and holography,'' \jhep{11}{2015}{014}, arXiv:1402.6334,\\
D.~Stanford and E.~Witten, ``Fermionic Localization of the Schwarzian Theory,'' \jhep{10}{2017}{008}, arXiv:1703.04612 [hep-th],\\
K.~Jensen, ``Chaos in AdS$_2$ Holography,'' \prl{117}{2016}{111601}, arXiv:1605.06098,\\
J.~Maldacena, D.~Stanford and Z.~Yang, ``Conformal symmetry and its breaking in two dimensional Nearly Anti-de-Sitter space,'' \emph{PTEP} {\bf 12} (2016) 12C104 arXiv:1606.01857 [hep-th],\\
J.~Engels\"oy, T.~G.~Mertens and H.~Verlinde,``An investigation of AdS$_{2}$ backreaction and holography,'' \jhep{07}{2016}{139}, arXiv:1606.03438 [hep-th],\\
P.~Saad, S.~H.~Shenker and D.~Stanford, ``JT gravity as a matrix integral,'' arXiv:1903.11115.}
%
\bibitem{JTappli2}{S.~Sachdev and J.~Ye, ``Gapless Spin-Fluid Ground State in a Random Quantum Heisenberg Magnet,'' \prl{70}{1993}{3339}, cond-mat/9212030,\\
A.~Kitaev, ``A Simple Model of Quantum Holography,'' KITP Program \emph{Entanglement in Strongly-Correlated Quantum Matter}, unpublished, see http://online.kitp.ucsb.edu/online/entangled15/,\\
R.~Gurau, ``The $1/N$ expansion of colored tensor models,'' \emph{Annales Henri Poincar\'e} \textbf{12} (2011) 829, arXiv:1011.2726,\\
E.~Witten, ``An SYK-Like Model Without Disorder,'' \jpa{52}{2019}{47}, arXiv:1610.09758,\\
I.~R.~Klebanov and G.~Tarnopolsky, ``Uncolored random tensors, melon diagrams, and the Sachdev-Ye-Kitaev models,'' \prd{95}{2017}{046004}, arXiv:1611.08915,\\
F.~Ferrari, ``The Large $D$ Limit of Planar Diagrams,'' \emph{Ann.\ Inst.\ Henri Poincar\'e Comb.\ Phys.\ Interact.}\ \textbf{D6} (2019) 427, arXiv:1701.01171.}
%
\bibitem{Zamolod}{A.B.~Zamolodchikov, ``On the Entropy of Random Surfaces,'' \plb{117}{1982}{87}.}
%
\bibitem{WittLiousemicl}{D.~Harlow, J.~Maltz and E.~Witten, ``Analytic Continuation of Liouville Theory,'' \jhep{12}{2011}{071}, arXiv:1108.4417 [hep-th].}

\bibitem{Muhlmann}{D.~Anninos, T.~Bautista and B.~M\"uhlmann, ``The two-sphere partition function in two-dimensional quantum gravity,'' \jhep{09}{2021}{116}, arXiv:2106.01665 [hep-th],\\
B.~M\"uhlmann, ``The two-sphere partition function in two-dimensional quantum gravity at fixed area,'' \jhep{09}{2021}{189}, arXiv:2106.04532 [hep-th],\\
D.~Anninos and B.~M\"uhlmann, ``The semiclassical gravitational path integral and random matrices (toward a microscopic picture of a dS$_{2}$ universe),'' \jhep{12}{2021}{206}, arXiv:2111.05344 [hep-th],\\
B.~M\"uhlmann, ``The two-sphere partition function from timelike Liouville theory at three-loop order,'' \jhep{05}{2022}{057}, arXiv:2202.04549 [hep-th].}
%
\bibitem{StanLiousecl}{R.~Mahajan, D.~Stanford and C.~Yan, ``Sphere and disk partition functions in Liouville and in matrix integrals,'' \jhep{07}{2022}{132}, arXiv:2107.01172 [hep-th].}
%
\bibitem{FerK1}{F.~Ferrari, S.~Klevtsov and S.~Zelditch, ``Random geometry, quantum gravity and the K\"ahler potential,'' \plb{705}{2011}{375}, arXiv:1107.4022 [hep-th],\\
F.~Ferrari, S.~Klevtsov and S.~Zelditch, ``Gravitational Actions in Two Dimensions and the Mabuchi Functional,'' \npb{859}{2012}{341}, arXiv:1112.1352 [hep-th],\\
H.~Lacoin, R.~Rhodes and V.~Vargas, ``Path integral for quantum Mabuchi K-energy,'' \emph{Duke Math.\ J.\ }\textbf{171} (2022) no.3, 483, arXiv:1807.01758 [math-ph].}
%
\bibitem{FerK2}{A.~Bilal, F.~Ferrari and S.~Klevtsov, ``2D Quantum Gravity at One Loop with Liouville and Mabuchi Actions,'' \npb{880}{2014}{203}, arXiv:1310.1951 [hep-th],\\
A.~Bilal and L.~Leduc, ``2D quantum gravity on compact Riemann surfaces and two-loop partition function: A first principles approach,'' \npb{896}{2015}{360}, arXiv:1412.5189 [hep-th],\\
L.~Leduc and A.~Bilal, ``2D quantum gravity at three loops: a counterterm investigation,'' \npb{903}{2016}{226}, arXiv:1504.01738 [hep-th].}
%
\bibitem{FZZT}{V.~Fateev, A.~B.~Zamolodchikov and A.~B.~Zamolodchikov, ``Boundary Liouville field theory. 1. Boundary state and boundary two point function,''
arXiv:hep-th/0001012 [hep-th],\\
J.~Teschner, ``Remarks on Liouville theory with boundary,'' PoS \textbf{tmr2000} (2000) 041, 
arXiv:hep-th/0009138 [hep-th].}
%
\bibitem{DOZZ}{H.~Dorn and H.J.~Otto, ``On Correlation Functions for Non-critical Strings with $c \leq 1$ but $d\geq 1$,'' \plb{291}{1992}{39}, hep-th/9206053,\\
H.~Dorn and H.J.~Otto, ``Two and three point functions in Liouville theory,'' \npb{1994}{429}{375}, hep-th/9403141,\\
A.~B.~Zamolodchikov and A.~B.~Zamolodchikov, ``Structure constants and conformal bootstrap in Liouville field theory,'' \npb{477}{1996}{577}, arXiv:hep-th/9506136 [hep-th],\\
A.~Kupiainen, R.~Rhodes and V.~Vargas,``The DOZZ Formula from the Path Integral,'' \jhep{05}{2018}{094}, arXiv:1803.05418 [hep-th]].}
%
\bibitem{bcJTpapers}{A.~Goel, L.~V.~Iliesiu, J.~Kruthoff and Z.~Yang,
``Classifying boundary conditions in JT gravity: from energy-branes to $\alpha$-branes,'' \jhep{04}{2021}{069}, arXiv:2010.12592 [hep-th],\\
F.~Ferrari, ``Gauge Theory Formulation of Hyperbolic Gravity,'' \jhep{03}{2021}{046}, arXiv:2011.02108 [hep-th].}
%
\bibitem{Lioumath}{D.~Kraus and O.~Roth, ``Conformal Metrics,'' arXiv:0805.2235 [math.CV].}
%
\bibitem{detpaper}{S.~Chaudhuri and F.~Ferrari, ``Dirichlet Scalar Determinants On Two-Dimensional Constant Curvature Disks,'' arXiv:2405.14958 [hep-th].}
%
\bibitem{KitaevSuhmic}{A.~Kitaev and S.~J.~Suh, ``Statistical mechanics of a two-dimensional black hole,'' \jhep{05}{2019}{198}, arXiv:1808.07032 [hep-th].}
%
\bibitem{StanfordYangmic}{D.~Stanford and Z.~Yang, ``Finite-cutoff JT gravity and self-avoiding loops,'' arXiv:2004.08005 [hep-th].}
%
\bibitem{multanomaly}{C.~Kassel, \emph{Ast\'erisque} \textbf{177} (1989) 199 (s\'eminaire Bourbaki);\\
M.~Kontsevitch and S.~Vishik, Functional Analysis on the Eve of the 21st Century, vol.\ 1, pages 173--197 (1993),\\
E.~Elizalde, L.~Vanzo and S.~Zerbini, \emph{Zeta function regularization, the multiplicative anomaly and the Wodzicki residue}, \cmp{194}{1998}{613}, [hep-th/9701060];\\
E.~Elizalde, \emph{Applications in physics of the multiplicative anomaly formula involving some basic differential operators}, \npb{532}{1998}{407}, [hep-th/9804118];\\
E.~Elizalde and M.~Tierz, \emph{Multiplicative anomaly and zeta factorization}, \jmp{45}{2004}{1168}, [hep-th/0402186];\\
M.~Wodzicki, Non-Commutative Residues, Chapter 1, Lecture Notes in Mathematics, Springer-Verlag, Berlin (1987).}
%

\end{thebibliography}
\end{document}